\documentclass{vldb}

\usepackage{graphicx}

\usepackage{amsmath,amssymb,epsfig,latexsym,amssymb,graphics}

\usepackage[normalem]{ulem}
\usepackage{url}
\usepackage{epsfig}
\usepackage{xspace}
\usepackage{epstopdf}
\usepackage{color}
\usepackage{pstricks}
\usepackage{times}
\usepackage{epsfig}
\usepackage{graphicx}
\usepackage{xspace}
\usepackage{subfigure}
\usepackage{float}
\usepackage{latexsym}
\usepackage{algorithm}
\usepackage[noend]{algorithmic}
\usepackage{url}
\usepackage{color}
\usepackage{colortbl}
\usepackage[bookmarks=false]{hyperref}
\usepackage{fancynum}
\usepackage{xcolor,colortbl}
\usepackage[font=footnotesize]{caption}
\usepackage{textcomp}
\usepackage{mathrsfs}
\usepackage{tikz}
\usepackage{amsmath}
\usepackage{enumitem}
\usepackage{balance}

\newcommand{\topcaption}{%
 \setlength{\abovecaptionskip}{0.1cm}%
 \setlength{\belowcaptionskip}{0.1cm}%
 \caption}

\setlist[itemize]{leftmargin=0.4cm,itemsep=0.1cm,topsep=0.1cm}
\setlist{noitemsep}

\newtheorem{definition}{Definition}
\newtheorem{lemma}{Lemma}
\newtheorem{theorem}{Theorem}
\newtheorem{example}{Example}
\newtheorem{property}{Property}

\newcommand{\kw}[1]{{\ensuremath {\mathsf{#1}}}\xspace}
\newcommand{\kwnospace}[1]{{\ensuremath {\mathsf{#1}}}}

\newcommand{\stitle}[1]{\vspace{1ex} \noindent{\bf #1}}
\newcommand{\sstitle}[1]{\vspace{1ex} \noindent{\it #1}}
\newcommand{\rtitle}[1]{\vspace{1ex} \noindent{\textit{\textbf{#1}}}}

\newcommand{\reffig}[1]{Fig.~\ref{fig:#1}}
\newcommand{\refsec}[1]{Section~\ref{sec:#1}}
\newcommand{\refsubsec}[1]{Subsection~\ref{subsec:#1}}
\newcommand{\reftab}[1]{Table~\ref{tab:#1}}
\newcommand{\refalg}[1]{Algorithm~\ref{alg:#1}}
\newcommand{\refeq}[1]{Eq.~\ref{eq:#1}}
\newcommand{\refdef}[1]{Definiton~\ref{def:#1}}
\newcommand{\refthm}[1]{Theorem~\ref{thm:#1}}
\newcommand{\reflem}[1]{Lemma~\ref{lem:#1}}

\newcommand{\vc}[1]{\kappa(#1)\xspace}
\newcommand{\ec}[1]{\kappa'(#1)\xspace}
\newcommand{\mindegree}[1]{\delta(#1)\xspace}
\newcommand{\vcc}{{VCC}\xspace}
\newcommand{\vccs}{{VCCs}\xspace}
\newcommand{\ecc}{{ECC}\xspace}
\newcommand{\eccs}{{ECCs}\xspace}

\newcommand{\ccs}{{cores}\xspace}
\newcommand{\opartition}{\kw{OVERLAP\textrm{-}PARTITION}}
\newcommand{\cut}{{\cal S}\xspace}
\newcommand{\gdep}{g\textrm{-}deposit}
\newcommand{\prufunc}{\kw{SWEEP}}

\newcommand{\frk}{\kw{KVCC\textrm{-}ENUM}}
\newcommand{\fcbase}{\kw{GLOBAL\textrm{-}CUT}}
\newcommand{\fcstar}{\kw{GLOBAL\textrm{-}CUT^*}}
\newcommand{\lcut}{\kw{LOC\textrm{-}CUT}}

\newcommand{\bsa}{\kw{VCCE}}
\newcommand{\nsa}{\kw{VCCE\textrm{-}N}}
\newcommand{\gsa}{\kw{VCCE\textrm{-}G}}
\newcommand{\mra}{\kw{VCCE^*}}

\newcommand{\sfd}{\kwnospace{Stanford}\xspace}
\newcommand{\dblp}{\kwnospace{DBLP}\xspace}
\newcommand{\nd}{\kwnospace{ND}\xspace}
\newcommand{\google}{\kwnospace{Google}\xspace}
\newcommand{\cit}{\kwnospace{Cit}\xspace}
\newcommand{\cnr}{\kwnospace{Cnr}\xspace}
\newcommand{\ytb}{\kwnospace{Youtube}\xspace}

\title{Enumerating k-Vertex Connected Components in Large Graphs}

\author{
	Dong Wen$^{\natural}$, Lu Qin$^{\natural}$, Xuemin Lin$^{\ddag}$, Ying Zhang$^{\natural}$,  and Lijun Chang$^{\ddag}$ \vspace{1mm}\\
	\affaddr{$^{\natural}$CAI, University of Technology, Sydney, Australia}\\
	\affaddr{$^{\ddag}$The University of New South Wales, Australia}\\
	\ttfamily
	$^{\natural}$dong.wen@student.uts.edu.au; \{lu.qin, ying.zhang\}@uts.edu.au;\\
	\ttfamily
	$^{\ddag}$\{lxue, ljchang\}@cse.unsw.edu.au;
}

\begin{document}
\toappear{}
\maketitle


\begin{abstract}
Cohesive subgraph detection is an important graph problem that is widely applied in many application domains, such as social community detection, network visualization, and network topology analysis. Most of existing cohesive subgraph metrics can guarantee good structural properties but may cause the free-rider effect. Here, by free-rider effect, we mean that some irrelevant subgraphs are combined as one subgraph if they only share a small number of vertices and edges. In this paper, we study $k$-vertex connected component ($k$-\vcc) which can effectively eliminate the free-rider effect but less studied in the literature. A $k$-\vcc is a connected subgraph in which the removal of any $k-1$ vertices will not disconnect the subgraph. In addition to eliminating the free-rider effect, $k$-\vcc also has other advantages such as bounded diameter, high cohesiveness, bounded graph overlapping, and bounded subgraph number. We propose a polynomial time algorithm to enumerate all $k$-\vccs of a graph by recursively partitioning the graph into overlapped subgraphs. We find that the key to improving the algorithm is reducing the number of local connectivity testings. Therefore, we propose two effective optimization strategies, namely neighbor sweep and group sweep, to largely reduce the number of local connectivity testings. We conduct extensive performance studies using seven large real datasets to demonstrate the effectiveness of this model as well as the efficiency of our proposed algorithms.
\end{abstract}



\section{Introduction}
\label{sec:introduction}
Graphs have been widely used to represent the relationships of entities in the real world. With the proliferation of graph applications, research efforts have been devoted to many fundamental problems in mining and analyzing graph data. Recently, cohesive subgraph detection has drawn intense research interest \cite{PattilloYB13}. Such problem can be widely adopted in many real-world applications, such as community detection \cite{Cui2014, Huang2014}, network clustering \cite{verma2012}, graph visualization \cite{kalvarez2005, Zhang2012}, protein-protein network analysis \cite{bader2003}, and system analysis \cite{zhang2010}.

In the literature, a large number of cohesive subgraph models have been proposed. Among them, a clique, in which every pair of vertices are connected, guarantees perfect familiarity and reachability among vertices. Since the definition of the clique is too strict, clique-relaxation models are proposed in the literature including $s$-clique \cite{Mokken1979}, $s$-club \cite{Mokken1979}, $\gamma$-quasi-clique \cite{Zeng2006} and $k$-plex \cite{Berlowitz2015, Stephen1978}. Nevertheless, these models require exponential computation time and may lack guaranteed cohesiveness. To conquer this problem, other models are proposed such as $k$-core \cite{Batagelj2003}, $k$-truss \cite{cohen2008trusses, Wang2012, Shao2014}, $k$-mutual-friend subgraph \cite{zhao2012large} and $k$-\ecc ($k$-edge connected component) \cite{Zhou2012, Chang2013}, which require polynomial computation time and guarantee decent cohesiveness. For example, a $k$-core guarantees that every vertex has a degree at least $k$ in the subgraph, and a $k$-\ecc guarantees that the subgraph cannot be disconnected after removing any $k-1$ edges.

\begin{figure}[t]
\begin{center}
\includegraphics[width=0.99\hsize]{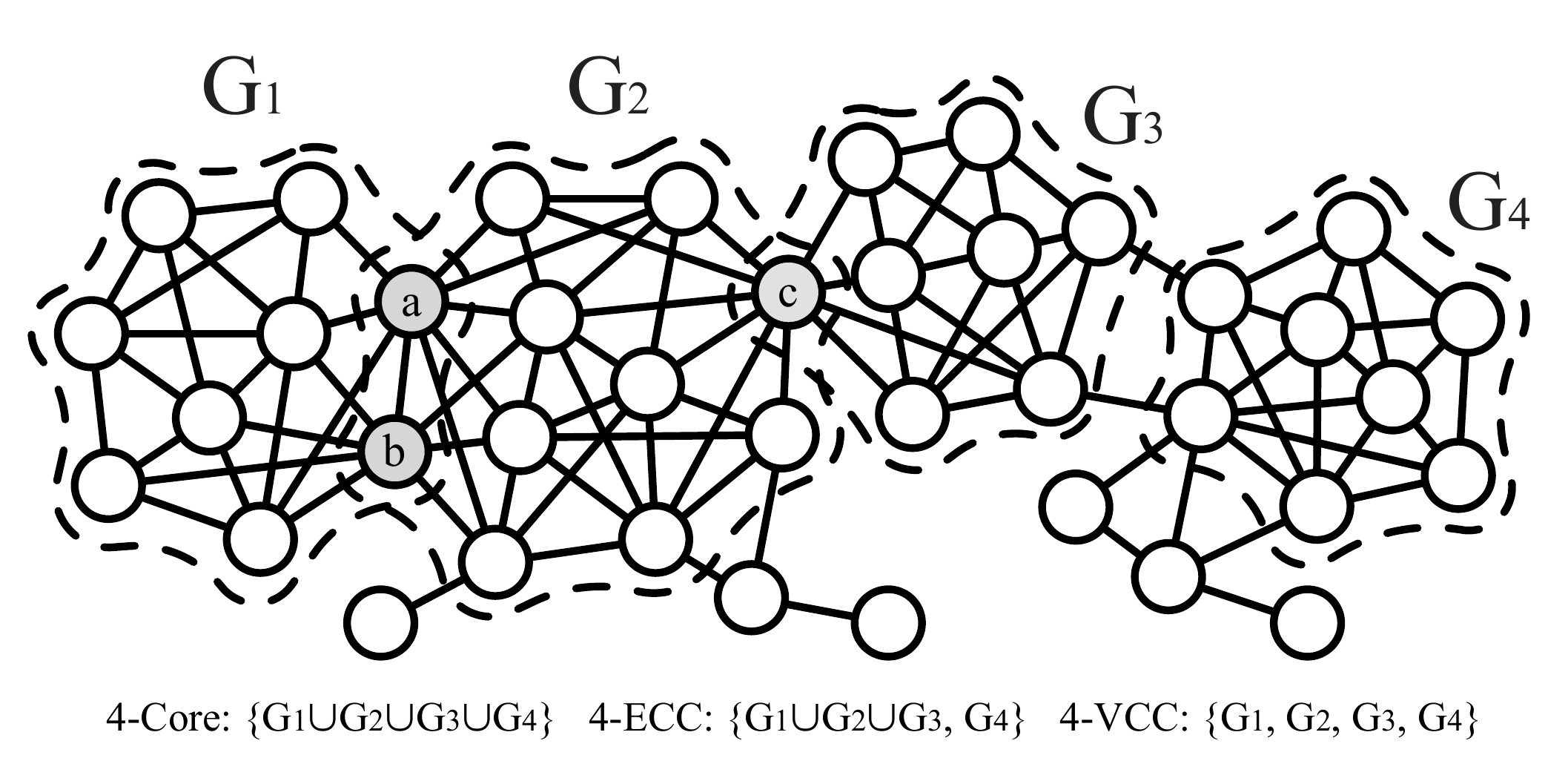}
\topcaption{Cohesive subgraphs in graph $G$.}
\label{fig:demo}
\end{center}
\end{figure}

\stitle{Motivation.} Despite the good structural guarantees in existing cohesive subgraph models,  we find that most of these models cannot effectively eliminate the free-rider effect. Here, by free-rider effect, we mean that some irrelevant subgraphs are combined as one result subgraph if they only share a small number of vertices and edges. To illustrate the free rider effect, we consider a graph $G$ shown in \reffig{demo}, which includes four subgraphs $G_1$, $G_2$, $G_3$, and $G_4$. The four subgraphs are loosely connected because: $G_1$ and $G_2$ share a single edge $(a,b)$; $G_2$ and $G_3$ share a single vertex $c$; and $G_3$ and $G_4$ do not share any edge or vertex. Let $k=4$. Based on the $k$-core model, there is only one $k$-core, which is the union of the four subgraphs $G_1$, $G_2$, $G_3$, and $G_4$, along with the two edges connecting $G_3$ and $G_4$. Based on the $k$-\ecc model, there are two $k$-\eccs, which are $G_4$ and the union of three subgraphs $G_1$, $G_3$, and $G_3$. Motivated by this, we aim to detect cohesive subgraphs and effectively eliminate the free-rider effect, i.e., to accurately detect $G_1$, $G_2$, $G_3$ and $G_4$ as result cohesive subgraphs in \reffig{demo}.

In the literature, a recent work \cite{Wu2015} aims to eliminate the free-rider effect in local community search. Given a query vertex, the algorithm in \cite{Wu2015} tries to eliminate the free-rider effect by weighting each vertex in the graph by its proximity to the query vertex. Based on the vertex weights, a query-biased subgraph is returned by considering both the density and the proximity to the query vertex. Unfortunately, such a query-biased local community model cannot be used in cohesive subgraph detection. 

\stitle{$k$-Vertex Connected Component.} Vertex connectivity, which is also named structural cohesion \cite{Moody00structuralcohesion}, is \textit{the minimum number of vertices that need to be removed to disconnect the graph}. It has been proved as an outstanding metric to evaluate the cohesiveness of a social group \cite{Moody00structuralcohesion, white2001}. We find this sociological conception can be used to detect cohesive subgraphs and effectively eliminate the free-rider effect. Given an integer $k$, a $k$-vertex connected component ($k$-\vcc) is \textit{a maximal connected subgraph in which the removal of any $k-1$ vertices cannot disconnect the subgraph}. Given a graph $G$ and a parameter $k$, we aim to detect all $k$-\vccs in $G$. In \reffig{demo} and $k=4$, there are four $k$-\vccs $G_1$, $G_2$, $G_3$, and $G_4$ in $G$. The subgraph formed by the union of $G_1$ and $G_2$ is not a $k$-\vcc because it will be disconnected by removing two vertices $a$ and $b$.

\stitle{Effectiveness.} $k$-\vcc effectively eliminates the free-rider effect by ensuring that each $k$-\vcc cannot be disconnected by removing any $k-1$ vertices. In addition, $k$-\vcc also have the following four good structural properties. 
\begin{itemize}
\item \textit{Bounded Diameter.} The diameter of a $k$-\vcc $G'(V', E')$ is bounded by $\lfloor \frac{|V'|-2}{\vc{G'}} \rfloor+1$ where $\vc{G'}$ is the vertex connectivity of $G'$. For example, we consider the $4$-\vcc $G_1$ with $9$ vertices in \reffig{demo}. The diameter of $G_1$ is bounded by $2$. 

\item \textit{High Cohesiveness.} We can guarantee that a $k$-\vcc is nested in a $k$-\ecc and a $k$-core. Therefore, a $k$-\vcc is generally more cohesive and inherits all the structural properties of a $k$-core and a $k$-\ecc. For example, each of the four $4$-\vccs in \reffig{demo} is also a $4$-core and a $4$-\ecc.

\item \textit{Subgraph Overlapping.} Unlike $k$-core and $k$-\ecc, $k$-\vcc model allows overlapping between $k$-\vccs, and we can guarantee that the number of overlapped vertices for any pair of $k$-\vccs is smaller than $k$. For example, the two $4$-\vccs $G_1$ and $G_2$ in \reffig{demo} overlap two vertices and an edge. 

\item \textit{Bounded Subgraph Number.} Even with overlapping, we can bound the number of $k$-\vccs to be linear to the number of vertices in the graph. This indicates that redundancies in the $k$-\vccs are limited. For example, the graph shown in \reffig{demo} contains four $4$-\vccs with three vertices $a$, $b$, and $c$ duplicated.  
\end{itemize}

\noindent The details of the four properties can be found in \refsec{pre:why}.

\stitle{Efficiency.} In this paper, we propose an algorithm to enumerate all $k$-\vccs in a given graph $G$ via overlapped graph partition. Briefly speaking, we aim to find a vertex cut with fewer than $k$ vertices in $G$. Here, a vertex cut of $G$ is a set of vertices the removal of which disconnects the graph. With the vertex cut, we can partition $G$ into overlapped subgraphs each of which contains all the vertices in the cut along with their induced edges. We recursively partition each of the subgraphs until no such cut exists. In this way, we compute all $k$-\vccs.  For example, suppose the graph $G$ is the union of $G_1$ and $G_2$ in \reffig{demo}, $k=4$, we can find a vertex cut with two vertices $a$ and $b$. Thus we partition the graph into two subgraphs $G_1$ and $G_2$ that overlap two vertices $a$, $b$ and an edge $(a,b)$. Since neither $G_1$ nor $G_2$ has any vertex cut with fewer then $k$ vertices, we return $G_1$ and $G_2$ as the final $k$-\vccs.  We theoretically analyze our algorithm and prove that \textit{the set of $k$-\vccs can be enumerated in polynomial time}. More details can be found in \refsec{base:complexity}.

Nevertheless, the above algorithm has a large improvement space. The most crucial operation in the algorithm is called local connectivity testing, which given two vertices $u$ and $v$, tests whether $u$ and $v$ can be disconnected in two components by removing at most $k-1$ vertices from $G$. To find a vertex cut with fewer than $k$ vertices, we need to conduct local connectivity testing between a source vertex $s$ and each of other vertices $v$ in $G$ in the worst case. Therefore, \textit{the key to improving algorithmic efficiency is to reduce the number of local connectivity testings in a graph.} Given a source vertex $s$, if we can avoid testing the local connectivity between $s$ and a certain vertex $v$, we call it as we can \textit{sweep vertex $v$}. We propose two strategies to sweep vertices.

\begin{itemize}
\item \textit{Neighbor Sweep.} If a vertex has certain properties, all its neighbors can be swept. Therefore, we call this strategy neighbor sweep. Moreover, we maintain a deposit value for each vertex, and once we finish testing or sweep a vertex, we increase the deposit values for its neighbors. If the deposit value of a vertex satisfies certain condition, such vertex can also be swept. 

\item \textit{Group Sweep.} We introduce a method to divide vertices in a graph into disjoint groups. If a vertex in a group has certain properties, vertices in the whole group can be swept. We call this strategy group sweep. Moreover, we maintain a group deposit value for each group. Once we test or sweep a vertex in the group, we increase the corresponding group deposit value. If the group deposit value satisfies certain conditions, vertices in such whole group can also be swept. 
\end{itemize}

\noindent Even though these two strategies are studied independently, they can be used together and boost the effectiveness of each other. With these two vertex sweep strategies, we can significantly reduce the number of local connectivity testings in the algorithm. Experimental results show the excellent performance of our sweep strategies. More details can be found in \refsec{opt} and \refsec{exp}.

\stitle{Contributions.} We make the following contributions in this paper.

\sstitle{(1) Theoretical analysis for the effectiveness of $k$-\vcc.} We present several properties to show the excellent quality of $k$-vertex connected component. Although the concept of vertex connectivity has been studied in the literature to evaluate the cohesiveness of a social group, this is the first work that aims to enumerate all $k$-\vccs and considers free-rider effect elimination in cohesive subgraph detection to the best of our knowledge.

\sstitle{(2) A polynomial time algorithm based on overlapped graph partition.} We propose an algorithm to compute all $k$-\vccs in a graph $G$. The algorithm recursively divides the graph into overlapped subgraphs until each subgraph cannot be further divided. We prove that our algorithm terminates in polynomial time. 

\sstitle{(3) Two effective pruning strategies.} We design two pruning strategies, namely neighbor sweep and group sweep, to largely reduce the number of local connectivity testings and thus significantly speed up the algorithm. 

\sstitle{(4) Extensive performance studies.} We conduct extensive performance studies on 7 real large graphs to demonstrate the effectiveness of $k$-\vcc and the efficiency of our proposed algorithms. 

\stitle{Outline.} The rest of this paper is organized as follows. \refsec{pre} formally defines the problem and presents its rationale. \refsec{kvcc} gives a framework to compute all $k$-\vccs in a given graph. \refsec{base} gives a basic implementation of the framework and analyzes the time complexity of the algorithm. \refsec{opt} introduces several strategies to speed up the algorithm. \refsec{exp} evaluates the model and algorithms using extensive experiments. \refsec{relatedwork} reviews related works and \refsec{conclusion} concludes the paper.


\section{Preliminary}
\label{sec:pre}
\subsection{Problem Statement}
\label{sec:pre:ps}

In this paper, we consider an undirected and unweighted graph $G(V,E)$, where $V$ is the set of vertices and $E$ is the set of edges.  We also use $V(G)$ and $E(G)$ to denote the set of vertices and edges of graph $G$ respectively. The number of vertices and the number of edges are denoted by $n = |V|$ and $m = |E|$ respectively. For simplicity and without loss of generality, we assume that $G$ is a connected graph. We denote neighbor set of a vertex $u$ by $N(u)$, i.e., $N(u) = \{ u \in V| (u,v) \in E \}$, and degree of $u$ by $d(u) = |N(u)|$. Given two graphs $g$ and $g'$, we use $g\subseteq g'$ to denote that $g$ is a subgraph of $g'$. Given a set of vertices $V_s$, the induced subgraph $G[V_s]$ is a subgraph of $G$ such that $G[V_s] = (V_s,\{(u,v) \in E | u,v \in V_s\})$. For any two subgraphs $g$ and $g'$ of $G$, we use $g\cup g'$ to denote the union of $g$ and $g'$, i.e., $g\cup g'=(V(g)\cup V(g'), E(g)\cup E(g'))$. Before stating the problem, we firstly give some basic definitions.

\begin{definition}
\textsc{(Vertex Connectivity)} The vertex connectivity of a graph $G$, denoted by $\vc{G}$, is defined as the minimum number of vertices whose removal results in either a disconnected graph or a trivial graph (a single-vertex graph).
\end{definition}

\begin{definition}
\label{def:kcon}
\textsc{(k-Vertex Connected)} A graph $G$ is $k$-vertex connected if: 1) $|V(G)|>k$; and 2) the remaining graph is still connected after removing any ($k-1$) vertices. That is, $\vc{G} \ge k$.
\end{definition}

We use the term $k$-\textit{connected} for short when the context is clear. It is easy to see that any nontrivial connected graph is at least $1$-connected. Based on \refdef{kcon}, we define the $k$-\textit{Vertex Connected Component} ($k$-\vcc) as follows.

\begin{definition}
\label{def:kvcc}
\textsc{(k-Vertex Connected Component)} Given a graph $G$, a subgraph $g$ is a $k$-vertex connected component ($k$-\vcc) of $G$ if: 1) $g$ is $k$-vertex connected; and 2) $g$ is maximal. That is, $\nexists g'\subseteq G$, such that $\vc{g'} \ge k$, $g \subseteq g'$.  
\end{definition}

\stitle{Problem Definition.} Given a graph $G$ and an integer $k$, we denote the set of all $k$-\vccs of $G$ as $\vcc_k(G)$. In this paper, we study the problem of efficiently enumerating all $k$-\vccs of $G$, i.e, to compute $\vcc_k(G)$ . 

\begin{example}
For the graph $G$ in \reffig{demo}, given parameter $k=4$, there are four $4$-\vccs: $\vcc_4(G)=\{$$G_1$, $G_2$, $G_3$, $G_4\}$. We cannot disconnect each of them by removing any $3$ or fewer vertices. Subgraph $G_1\cup G_2$ is not a $4$-\vcc because it will be disconnected after removing two vertices $a$ and $b$.
\end{example}

\subsection{Why k-Vertex Connected Component?}
\label{sec:pre:why}

$k$-\vcc model effectively reduces the free-rider effect by ensuring that each $k$-\vcc cannot be disconnected by removing any $k-1$ vertices. In this subsection, we show other good structural properties of $k$-\vcc in terms of bounded diameter, high cohesiveness, bounded overlapping and bounded component number. None of other cohesive graph models, such as $k$-core and $k$-Edge Connected Component ($k$-\ecc) can achieve these four goals simultaneously.

\stitle{Diameter.} Before discussing the diameter of a $k$-\vcc, we first quote Global Menger's Theorem as follows.

\begin{theorem}
\label{thm:menger}
A graph is $k$-connected if and only if any pair of vertices $u$,$v$ is joined by at least $k$ vertex-independent $u$-$v$ paths. \cite{menger1927}
\end{theorem}

This theorem shows the equivalence of vertex connectivity and the number of vertex-independent paths, both of which are considered as important properties for graph cohesion \cite{white2001}. Based on this theorem, we can bound the diameter of a $k$-\vcc, where the diameter of a graph $G$, denoted by $diam(G)$, is the longest shortest path between any pair of vertices in $G$:

\vspace*{-0.1cm}
\begin{equation}
\label{eq:diameter}
\small
diam(G) = max_{u,v \in V(G)}dist(u,v,G)
\end{equation}

\noindent Here, $dist(u,v,G)$ is the shortest distance of the pair of vertices $u$ and $v$ in $G$. Small diameter is considered as an important feature for a good community in \cite{EdacherySB99}. We give the diameter upper bound for a $k$-\vcc as follows. 

\begin{theorem}
Given any $k$-\vcc $G_i$ of $G$, we have: 

\begin{equation}
\small
diam(G_i) \leq \lfloor\frac{|V(G_i)|-2}{\vc{G_i}}\rfloor+1.
\end{equation}
\end{theorem}

\begin{proof}
Consider any two vertices $u$ and $v$ in $G_i$, we have $d(u,$ $v,$ $G_i)$ $\le$ $diam(G_i)$. \refthm{menger} indicates that there exist at least $\vc{G_i}$ vertex-disjoint paths between $u$ and $v$ in $G_i$, and in each path, we have at most $diam(G_i)-1$ internal vertices since $dist$$(u,$ $v,$ $G_i)$ $\le diam(G_i)$. Thus we have at most $\vc{G_i} \times (diam(G_i)-1)$ internal vertices between them. With two endpoints $u$ and $v$, we have $2+\vc{G_i}\times(diam(G_i)-1) \le |V(G_i)|$. Thus the upper bound of $diam(G_i)$ is $\lfloor \frac{|V(G_i)|-2}{\vc{G_i}}\rfloor+1$.
\end{proof}

\stitle{Cohesiveness.} To further investigate the quality of $k$-\vcc, we introduce the Whitney Theorem \cite{hassler1932}. Given a graph $g$, it analyzes the inclusion relation between vertex connectivity $\vc{g}$, edge connectivity $\ec{g}$ and minimum degree $\mindegree{g}$. The theorem is presented as follows.

\begin{theorem}
\label{thm:relation}
For any graph $g$, $\vc{g} \le \ec{g} \le \mindegree{g}$.
\end{theorem}

From this theorem, we know that for a graph $G$, every $k$-\vcc of $G$ is nested in a $k$-\ecc in $G$, and every $k$-\ecc of $G$ is nested in a $k$-core in $G$. Therefore, $k$-\vcc is generally more cohesive than $k$-\ecc and $k$-core.

\stitle{Overlapping.} The $k$-\vcc model also supports vertex overlap between different $k$-\vccs, which is especially important in social networks. We can easily deduce the following property from the definition of $k$-\vcc to bound the overlapping size.

\begin{property}
\label{pr:overlap}
Given two $k$-\vccs $G_i$ and $G_j$ in graph $G$, the number of overlapped vertices of $G_i$ and $G_j$ is less than $k$. That is, $|V(G_i) \cap V(G_j)| < k$.  
\end{property}

\begin{example}
In \reffig{demo}, we find two vertices $a$ and $b$ that are contained in two $4$-\vccs $G_1$ and $G_2$. For $k$-\ecc and $k$-core, two components will be combined if they have one vertex in common. For example, there is only one $4$-core, which is the union of $G_1$, $G_2$, $G_3$ and $G_4$.
\end{example}

\stitle{Component Number.} Once we allow overlapping between different components, the number of components can hardly be bounded. For example, the number of maximal cliques achieves $3^\frac{n}{3}$ for a graph with $n$ vertices \cite{Moon1965}. Nevertheless, we find that the number of $k$-\vccs in a graph $G$ can be bounded by a function that is linear to the number of vertices in $G$. That is:

\begin{equation}
\label{eq:cnum}
|\vcc_k(G)|\leq \frac{|V(G)|}{2}
\end{equation}

\noindent Detailed discussions of \refeq{cnum} can be found in \refsec{base}. It is worth noting that the linear number of $k$-\vccs allows us to design a polynomial time algorithm to enumerate all $k$-\vccs. We will also discuss this in detail in \refsec{base}.

\section{Algorithm Framework}
\label{sec:kvcc}
\subsection{The Cut-Based Framework}
\label{sec:kvcc:frk}

To compute all $k$-\vccs in a graph, we introduce a cut-based framework \frk in this section. We define \textit{vertex cut}. 

\begin{definition}
\label{def:vertexcut}
\textsc{(Vertex Cut)} Given a connected graph $G$, a vertex subset $\cut \subset V$ is a vertex cut if the removal of $\cut$ from $G$ results in a disconnected graph.
\end{definition}

From \refdef{vertexcut}, we know that the vertex cut may not be unique for a given graph $G$, and the vertex connectivity is the size of the minimum vertex cut. For a complete graph, there is no vertex cut since any two vertices are adjacent. The size of a vertex cut is the number of vertices in the cut. In the rest of paper, we use cut to represent vertex cut when the context is clear.

\begin{algorithm}[t]
\caption{$\frk(G,k)$}
\label{alg:frk}
\small
\begin{algorithmic}[1]
\REQUIRE a graph $G$ and an integer $k$;
\ENSURE all $k$-vertex connected components;
\vspace{2ex}
\STATE $\vcc_k(G) \leftarrow \emptyset$;
\STATE \textbf{while} $\exists u: d(u)<k$ \textbf{do} remove $u$ and incident edges;
\STATE identify connected components ${\cal G} = \{G_1, G_2,...,G_t\}$ in $G$;
\FORALL{connected component $G_i\in {\cal G}$}
	\STATE $\cut \leftarrow \fcbase(G_i,k)$;
	\IF{$\cut = \emptyset$}
		\STATE $\vcc_k(G)\leftarrow \vcc_k(G) \cup \{G_i\}$;
	\ELSE
		\STATE ${\cal G}_i\leftarrow \opartition(G_i,\cut)$;
		\FORALL{$G^j_i\in {\cal G}_i$}
		 		\STATE $\vcc_k(G)\leftarrow \vcc_k(G)\cup \frk(G^j_i,k)$;
		\ENDFOR
	\ENDIF
\ENDFOR
\STATE \textbf{return} $\vcc_k(G)$;

\vspace*{0.2cm}
\STATE \textbf{Procedure} $\opartition($Graph $G'$, Vertex Cut $\cut)$
\STATE ${\cal G}\leftarrow \emptyset$;
\STATE remove vertices in $\cut$ and their adjacent edges from $G'$;
\FORALL{connected component $G'_i$ of $G'$}
	\STATE ${\cal G}\leftarrow {\cal G} \cup \{G'[V(G'_i)\cup \cut]\}$; 
\ENDFOR
\STATE \textbf{return} ${\cal G}$;
\end{algorithmic}
\end{algorithm}

\stitle{The Algorithm.} Given a graph $G$, the general idea of our cut-based framework is given as follows. If $G$ is $k$-connected, $G$ itself is a $k$-\vcc. Otherwise, there must exist a qualified cut $\cut$ whose size is less than $k$. In this case, we find such cut and partition $G$ into overlapped subgraphs using the cut. We repeat the partition procedure until each remaining subgraph is a $k$-\vcc. From \refthm{relation}, we know that a $k$-\vcc must be a $k$-core (a graph with minimum degree no smaller than $k$). Thus we can compute all $k$-cores in advance to reduce the size of the graph. 

The pseudocode of our framework is presented in \refalg{frk}. In line 2, the algorithm computes the $k$-core by iteratively removing the vertices whose degree is less than $k$ and terminates once no such vertex exists. Then we identify connected components of the input graph $G$. 
For each connected component $G_i$ (line~4), we first find a cut of $G_i$ by invoking the subroutine $\fcbase$ (line~5). Here, we only need to find a cut with fewer than $k$ vertices instead of a minimum cut. The detailed implementation of $\fcbase$ will be introduced later. If there is no such cut, it means $G_i$ is $k$-connected and we add it to the result list $\vcc_k(G)$ (line~6-7). Otherwise, we partition the graph into overlapped subgraphs using the cut $\cut$ by invoking $\opartition$ (line~9). 
We recursively cut each of other subgraphs using the same procedure $\frk$ (line~11) until all remaining subgraphs are $k$-\vccs. Next, we introduce the subroutine $\kw{OVERLAP\_PARTITION}$, which partitions the graph into overlapped subgraphs by cut $\cut$.

\begin{figure}[t]
\begin{center}
\includegraphics[width=0.99\hsize]{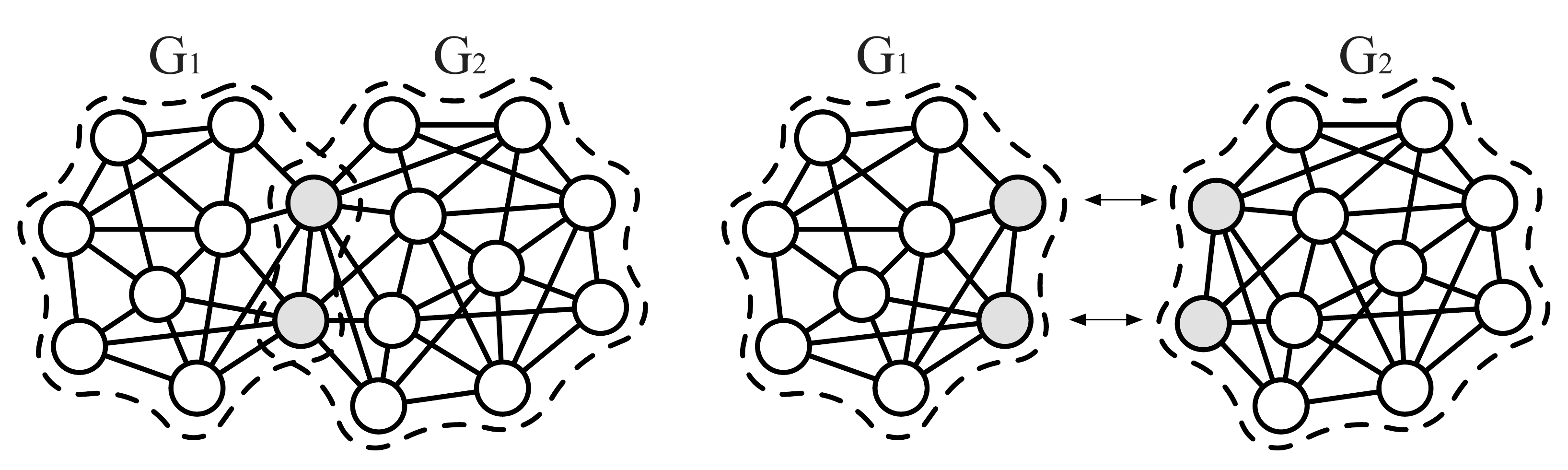}
\topcaption{An example of overlapped graph partition.}
\label{fig:partition}
\end{center}
\end{figure}

\stitle{Overlapped Graph Partition.} To partition a graph $G$ into overlapped subgraphs using a cut $\cut$, we cannot simply remove all vertices in $\cut$, since such vertices may be the overlapped vertices of two or more $k$-\vccs. Subroutine $\opartition$ is shown in line~13-18 of \refalg{frk}. We first remove the vertices in $\cut$ along with their adjacent edges from $G'$. $G'$ will become disconnected after removing $\cut$, since $\cut$ is a vertex cut of $G'$. We can simply add the cut $\cut$ into each connected component $G'_i$ of $G'$ and return induced subgraph $G'[V(G'_i)\cup \cut]$ as the partitioned subgraph (line~17-18). Partitioned subgraphs overlap each other since the cut $\cut$ is duplicated in these subgraphs. Below, we use an example to illustrate the partition process.

\begin{example}
We consider a graph $G$ on the left of \reffig{partition}. given the input parameter $k = 3$, we can find a vertex cut in which all vertices are marked by gray. These vertices belong to both $3$-\vccs, $G_1$ and $G_2$. Thus, given a cut $\cut$ of graph $G$, we partition the graph by duplicating the induced subgraph of cut $\cut$. As shown on the right of \reffig{partition}, we obtain two $3$-\vccs, $G_1$ and $G_2$, by duplicating the two cut vertices and their inner edges.
\end{example}


\subsection{Algorithm Correctness}
\label{sec:base:algcrt}

In this section, we prove the correctness of \frk using the following lemmas.

\begin{lemma} 
\label{lem:kcon}
Each of the subgraphs returned by \frk is $k$-vertex connected.
\end{lemma}

\begin{proof}
We prove it by contradiction. Assume one of the result subgraphs $G_i$ is not $k$-connected. $\fcbase$ in line 5 will find a vertex cut. $G_i$ will be partitioned in line 9 and cannot be returned, which contradicts that $G_i$ is in the result list.
\end{proof}

\begin{lemma} 
\label{lem:completeness}
(Completeness) The result returned by \frk contains all $k$-\vccs of the input graph $G$.
\end{lemma}

\begin{proof} Suppose graph $G$ is partitioned into overlapped subgraphs ${\cal G}'=\{$ $G'_1$, $G'_2$, $\ldots\}$ using a vertex cut $\cut$. We first prove that each $k$-\vcc $G_i$ of $G$ is contained in at least one  subgraph in ${\cal G}'$. We prove this by contradiction. We suppose that $G_i$ is not contained in any subgraph in ${\cal G}'$. Consider the computation of ${\cal G}'$, after we remove the vertices in $\cut$ and their adjacent edges from $G_i$, the remaining vertices in $G_i$ are contained in at least two graphs in ${\cal G}'$. This indicates that $\cut$ is a vertex cut of $G_i$. Since $|\cut|<k$, $G_i$ cannot be a $k$-\vcc, which contradicts that $G_i$ is a $k$-\vcc.  Therefore, we prove that we will not lose any $k$-\vcc. From \reflem{kcon}, we know that each of the returned subgraphs of \frk is $k$-connected. Therefore, all maximal subgraphs that are $k$-connected will be returned by \frk. In other words, all $k$-\vccs will be returned by \frk.
\end{proof}

\begin{lemma}
\label{lem:redundancyfree}
(Redundancy-Free) There does not exist two subgraphs $G_i$ and $G_j$ returned by \frk such that $G_i\subseteq G_j$. 
\end{lemma}

\begin{proof}
We prove it by contradiction. Suppose there are two subgraphs $G_i$ and $G_j$ returned by \frk such that $G_i\subseteq G_j$. On the one hand, we have $|V(G_i)\cap V(G_j)|=|V(G_i)|\geq k$. On the other hand, there must exist a partition ${\cal G}'=\{$ $G'_1$, $G'_2$, $\ldots\}$ of $G$ by a certain cut $\cut$ such that $G_i$ and $G_j$ are contained in two different graphs in ${\cal G}'$. From the partition procedure, we know that $G_i$ and $G_j$ have at most $k-1$ common vertices. This contradicts $|V(G_i)\cap V(G_j)|\geq k$. Therefore, the lemma holds. 
\end{proof}

\begin{theorem}
\label{thm:correctness}
\frk correctly computes all $k$-\vccs of $G$.
\end{theorem}

\begin{proof}
From \reflem{kcon}, we know that all subgraphs returned by \frk are $k$-connected.  From \reflem{completeness}, we know that all $k$-\vccs are returned by \frk. From \reflem{redundancyfree}, we know that all $k$-connected subgraphs returned by \frk are maximal, and no redundant subgraph will be produced. Therefore, \frk (\refalg{frk}) correctly computes all $k$-\vccs of $G$.
\end{proof}

Next, we show how to efficiently compute all $k$-\vccs following the framework in \refalg{frk}. From \refalg{frk}, we know that the key to improving algorithmic efficiency is to efficiently compute the vertex cut of a graph $G$. Below, we first introduce a basic algorithm in \refsec{base} to compute the vertex cut of a graph in polynomial time, and then we explore optimization strategies to accelerate the computation of the vertex cut in \refsec{opt}.

\section{Basic Solution}
\label{sec:base}
In the previous section, we propose a cut-based framework named $\frk$ to compute all $k$-\vccs. A key step in \refalg{frk} is $\fcbase$. Before giving the detailed implementation of $\fcbase$, we discuss techniques to find the edge-cut, which is highly related to the vertex-cut. Here, an edge-cut is a set of edges the removal of which will make the graph disconnected. We will show that these methods cannot be directly used to find the vertex-cut.

\stitle{Maximum Flow.} A basic solution to find edge cut is the maximum flow algorithm. With a given maximum flow, we can easily compute a minimum edge cut based on the \textit{Max-Flow Min-Cut Theorem}. However, the flow algorithm only considers capacity of each edge and does not have any limitation on that of vertex, which is obviously not suitable for finding the vertex cut.

\stitle{Min Edge-Cut.} Stoer and Wagner \cite{Stoer1997} proposed an algorithm to find global minimum edge cut in an undirected graph. The general idea is iteratively finding an edge-cut and merging a pair of vertices. It returns the edge-cut with the smallest value after $n-1$ merge operations. Given an upper bound $k$, the algorithm terminates once  an edge-cut with fewer then $k$ edges is found. However, this algorithm is not suitable for finding the vertex-cut since we do not know whether a vertex is included in the cut or not. Therefore, we cannot simply merge any two vertices in the whole procedure.

\subsection{Find Vertex Cut}
\label{sec:base:findcut}

We give some necessary definitions before introducing the idea to implement $\fcbase$.

\begin{definition}
\textsc{(Minimum $u$-$v$ Cut)} A vertex cut $\cut$ is a $u$-$v$ cut if $u$ and $v$ are in disjoint subsets after removing $\cut$, and it is a minimum $u$-$v$ cut if its size is no larger than that of other $u$-$v$ cuts.
\end{definition}

\begin{definition}
\label{def:lc}
\textsc{(Local Connectivity)} Given a graph $G$, the local connectivity of two vertices $u$ and $v$, denoted by $\vc{u,v,G}$, is defined as the size of the minimum $u$-$v$ cut. $\vc{u,v,G}=+\infty$ if no such cut exists. 
\end{definition}

Based on \refdef{lc}, we define two local $k$ connectivity relations as follows:

\begin{itemize}
\item  $u \equiv_G^k v$:  The local connectivity between $u$ and $v$ is not less than $k$ in graph $G$, i.e., $\vc{u,v,G} \ge k$.

\item $u \not\equiv_G^k v$:  The local connectivity between $u$ and $v$ is less than $k$ in graph $G$, i.e., $\vc{u,v,G} < k$.
\end{itemize}

We omit the suffix $G$, and use $u \equiv^k v$ and $u \not\equiv^k v$ to denote $u \equiv_G^k v$ and $u \not\equiv_G^k v$ respectively when the context is clear. Once $u \equiv^k v$, we say $u$ and $v$ is $k$-local connected. Obviously, $u \equiv^k v$ and $v \equiv^k u$ are equivalent.

\begin{algorithm}[t]
\caption{$\fcbase(G,k)$}
\label{alg:findcutbase}
\small
\begin{algorithmic}[1]
\REQUIRE a graph $G$ and an integer $k$;
\ENSURE a vertex cut with fewer than $k$ vertices;
\vspace{2ex}
\STATE compute a sparse certification ${\cal SC}$ of $G$;
\STATE select a source vertex $u$ with minimum degree;
\STATE construct the directed flow graph $\overline{\cal SC}$ of ${\cal SC}$;
\FORALL{$v \in V$}
	\STATE $\cut \leftarrow \lcut(u,v,\overline{\cal SC},{\cal SC})$;
	\STATE \textbf{if} $\cut \neq \emptyset$ \textbf{then} \textbf{return} $\cut$;
\ENDFOR
\FORALL{$v_a \in N(u)$}
	\FORALL{$v_b \in N(u)$}
		\STATE $\cut \leftarrow \lcut(v_a,v_b,\overline{\cal SC},{\cal SC})$;
		\STATE \textbf{if} $\cut \neq \emptyset$ \textbf{then} \textbf{return} $\cut$;
	\ENDFOR
\ENDFOR
\STATE return $\emptyset$;

\vspace{2ex}
\STATE \textbf{Procedure} $\lcut(u,v,\overline{G},G)$
\vspace{1ex}
\STATE \textbf{if} $v \in N(u)$ or $v = u$ \textbf{then} \textbf{return} $\emptyset$;
\STATE $\lambda \leftarrow $ calculate the maximum flow from $u$ to $v$ in $\overline{G}$;
\STATE \textbf{if} $ \lambda \ge k $ \textbf{then} \textbf{return} $\emptyset$;
\STATE compute the minimum edge cut in $\overline{G}$;
\STATE \textbf{return} the corresponding vertex cut in $G$;
\end{algorithmic}
\end{algorithm}

\stitle{The $\fcbase$ Algorithm.} We follow \cite{esfahanian1984} to implement $\fcbase$. Given a graph $G$, we assume that $G$ contains a vertex cut $\cut$ such that $|\cut| < k$. We consider an arbitrary source vertex $u$. There are only two cases: $(i)$ $u \not\in \cut$ and $(ii)$ $u \in \cut$. The general idea of algorithm $\fcbase$ considers two cases. In the first phase, we select a vertex $u$ and test the local connectivity between $u$ and all other vertices $v$ in $G$. We have either (a) $u\in \cut$ or (b) $G$ is $k$-connected if each local connectivity is not less than $k$. In the second phase, we consider the case $u\in \cut$ and test the local connectivity between any two neighbors of $u$ based on \reflem{cutnc}. More details can be found in \cite{esfahanian1984}. 

\begin{lemma}
\label{lem:cutnc}
Given a non-$k$-vertex connected graph $G$ and a vertex $u \in \cut$ where $\cut$ is a vertex cut and $|\cut|<k$, there exist $v, v' \in N(u)$ such that $v \not\equiv^k v'$.
\end{lemma}

The pseudocode is given in \refalg{findcutbase}. An optimization here is computing a sparse certificate of the original graph in line~1. Given a graph $G(V,E)$, a sparse certificate is a subset of edges $E' \in E$, such that the subgraph $G'(V,E')$ is $k$-connected if and only if $G$ is $k$-connected. Undoubtedly, the same algorithm is more efficient in a sparser graph. We will introduce the details of sparse certification in the next subsection. 

The first phase is shown in line~4-6. Once finding such cut $\cut$, we return it as the result. Similarly, the second phase is shown in line~7-10. Here, the procedure \lcut tests the local connectivity between $u$ and $v$ and returns the vertex cut if $u \not\equiv^k v$ (line 5 and line 9). To invoke \lcut, we need to transform the original graph $G$ into a directed flow graph $\overline{G}$. The details on how to construct the directed flow graph are introduced as follows.

\begin{figure}[t]
\begin{center}
\includegraphics[width=0.8\hsize]{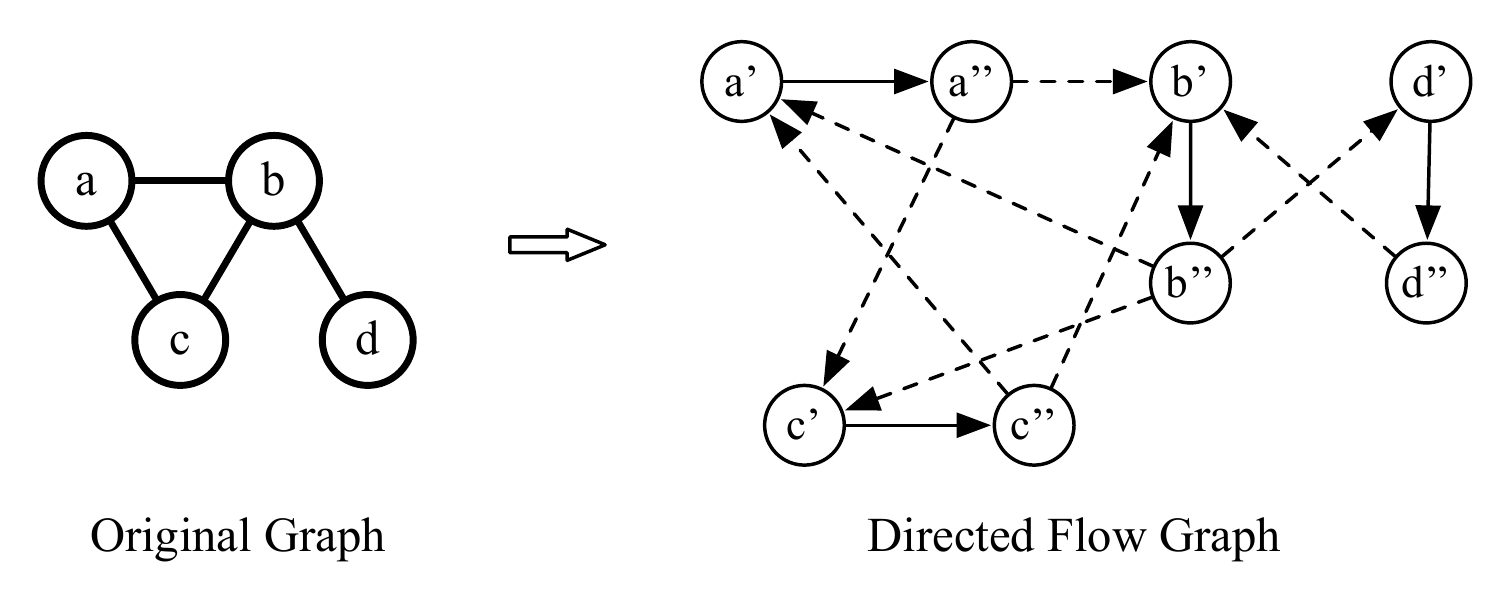}
\topcaption{An example of directed flow graph construction.}
\label{fig:dfgraph}
\end{center}
\end{figure}

\stitle{Directed Flow Graph.} The directed flow graph $\overline{G}$ of graph $G$ is an auxiliary directed graph which is used to calculate the local connectivity between two vertices. Given a graph $G$, we can construct the directed flow graph as follows. Each vertex $u$ in $G$ is represented by an directed edge $e_u$ in the directed flow graph $\overline{G}$. Let $u'$ and $u''$ denote the starting vertex and ending vertex of $e_u$. For each edge $(u,v)$ in $G$, we construct two directed edges: One is from $u''$ to $v'$, and the other is from $v''$ to $u'$. Consequently, we obtain $\overline{G}$ with $2n$ vertices and $n+2m$ edges and the capacity of every edge is $1$. 

\begin{example}
\reffig{dfgraph} gives an example of the directed flow graph construction. The solid lines in the directed flow graph represent vertices in the original graph, and the dashed lines in the directed flow graph represent edges in the original graph. The original graph contains 4 vertices and 4 edges, and the directed flow graph contains 8 vertices and 12 edges.
\end{example}

\stitle{The \lcut Procedure.} By using the directed flow graph, we convert vertex connectivity problem into edge connectivity problem. To calculate the local connectivity of two vertices $u$ and $v$, we perform the maximum flow algorithm on the directed flow graph. The value of the maximum flow is the local connectivity between $u$ and $v$.

The pseudocode of \lcut is given form line~12 to line~17 in \refalg{findcutbase}. It first checks whether $v$ is a neighbor of $u$ in line 13. If $u\in N(v)$, we always have $u \equiv^k v$ because of \reflem{nbrsave}.

\begin{lemma}
\label{lem:nbrsave}
$u \equiv^k v$ if $(u,v) \in E$.
\end{lemma}

Then the procedure computes the maximum flow $\lambda$ from $u$ to $v$ in $\overline{G}$ in line 14. If $\lambda \ge k$, we have $u \equiv^k v$ and the procedure returns $\emptyset$ in line 15. Otherwise, we compute the edge cut in $\overline{G}$ in line~16. Then we locate the corresponding vertices in the original graph $G$ for each edge in the edge cut and return them as the vertex cut of $G$ (line~16-17).

\subsection{Sparse Certificate}
\label{sec:base:cert}
We introduce the details of sparse certificate \cite{CheriyanKT93} in this section. In \refsec{opt}, we will show that the sparse certificate can not only be used to reduce the graph size, but also used to further reduce the local connectivity testings. 

\begin{definition}
\label{def:cert}
\textsc{(Certificate)} A certificate for the $k$-vertex connectivity of $G$ is a subset $E'$ of $E$ such that the subgraph $(V,E')$ is $k$-vertex connected if and only if $G$ is $k$-vertex connected.
\end{definition}

\begin{definition}
\label{def:scert}
\textsc{(Sparse Certificate)} A certificate for $k$-vertex connectivity of $G$ is called sparse if it has $O(k\cdot n)$ edges.
\end{definition}

From the definitions, we can see that a sparse certificate is equivalent to the original graph w.r.t $k$-vertex connectivity. It can also bound the edge size. We compute the sparse certificate (line 1 of \refalg{findcutbase}) according to the following theorem.

\begin{theorem}
\label{thm:maincert}
Let $G(V,E)$ be an undirected graph and let $n$ denote the number of vertices. Let $k$ be a positive integer. For $i = 1,2,...,k$, let $E_i$ be the edge set of a scan first search forest $F_i$ in the graph $G_{i-1} = (V,E-(E_1 \cup E_2 \cup ... \cup E_{i-1}))$. Then $E_1 \cup E_2 \cup ... \cup E_k$ is a certificate for the $k$-vertex connectivity of $G$, and this certificate has at most $k\times(n-1)$ edges \cite{CheriyanKT93}.
\end{theorem}

Based on \refthm{maincert}, we can simply generate the sparse certificate of $G$ using scan first search $k$ times, each of which creates a scan first search forest $F_i$. Below, we introduce how to perform a scan first search.

\stitle{Scan First Search.} In a scan first search of given graph $G$, for each connected component, we start from scanning a root vertex by marking all its neighbors. We scan an arbitrary marked but unscanned vertex each time and mark all its unvisited neighbors. This step is performed until all vertices are scanned. The resulting search forest forms the scan first forest of $G$. Obviously, a breath first search is a special case of scan first search.

\begin{figure}[t]
\begin{center}
\includegraphics[width=0.99\hsize]{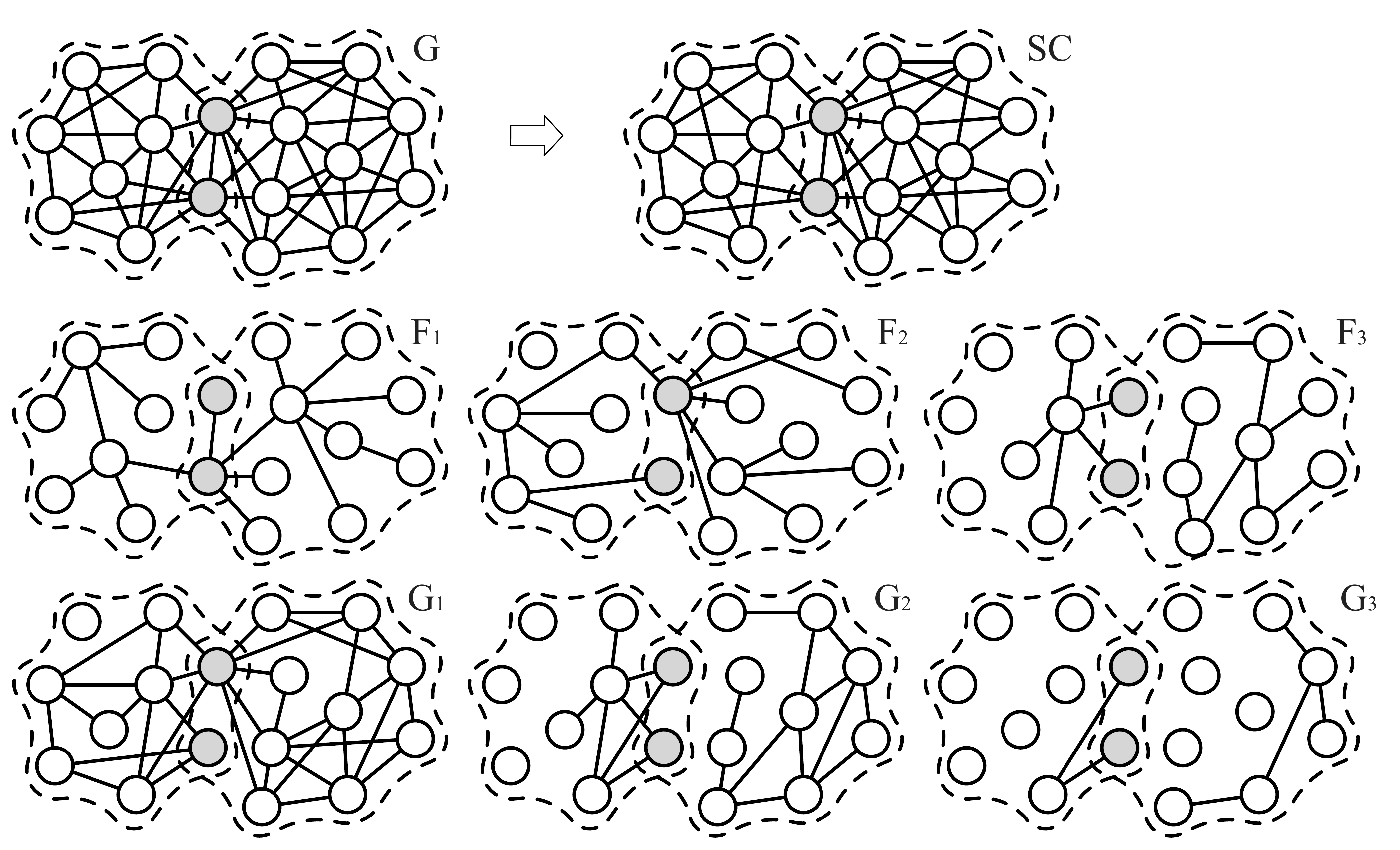}
\topcaption{The sparse certificate of given graph $G$ with $k = 3$}
\label{fig:sc}
\end{center}
\end{figure}

\begin{example}
\label{ex:sccn}
\reffig{sc} presents construction of a sparse certificate for the graph $G$. Let $k = 3$. For $i\in\{1,2,3\}$, $F_i$ denotes the scan first search forest obtained from $G_{i-1}$. $G_i$ is obtained by removing the edges in $F_i$ from $G_{i-1}$. $G_0$ is the input graph $G$. The obtained sparse certificate $SC$ is shown on the right side of $G$ with $SC = F_1 \cup F_2 \cup F_3$. All removed edges are shown in $G_3$.
\end{example}

\subsection{Algorithm Analysis}
\label{sec:base:complexity}

We analyze the basic algorithm in this section. In the directed flow graph, all edge capacities are equal to 1 and every vertex either has a single edge emanating from it or has a single edge entering it. For this kind of graph, the time complexity for computing the maximum flow is $O(n^{1/2}m)$ \cite{even1975network}. Note that we do not need to calculate the exact flow value in the algorithm. Once the flow value reaches $k$, we know that local connectivity between any two given vertices is at least $k$ and we can terminate the maximum flow algorithm. The time complexity for the flow computation is $O(\min(n^{1/2},k)\cdot m)$. Given a flow value and corresponding residual network, we can perform a depth first search to find the cut. It costs $O(m+n)$ time. As a result, we have the following lemma:

\begin{lemma}
\label{lem:tlcut}
The time complexity of $\lcut$ is $O(\min\\(n^{1/2},k)\cdot m)$.
\end{lemma}

Next we discuss the time complexity of $\fcbase$. The construction of both sparse certificate and directed flow graph costs $O(m+n)$ CPU time. Let $\delta$ denote the minimum degree in the input graph. We can easily get following lemma.

\begin{lemma}
\label{lem:lccnt}
$\fcbase$ invokes $\lcut$\\ $O(n+\delta^2)$ times in the worst case.
\end{lemma}

Next we discuss the CPU time complexity of the entire algorithm $\frk$. $\frk$ iteratively removes vertices with degree less than $k$ in line~2. This costs $O(m+n)$ time. Identifying all connected components can be performed by adopting a depth first search (line 3). This also need $O(m+n)$ time. To study the total time complexity spent by invoking $\fcbase$, we first give the following lemma.

\begin{lemma}
\label{lem:dpcnt}
For each subgraph created by the overlapped partition, at most $k-1$ vertices and $\frac{(k-1)(k-2)}{2}$ edges are increased after the partition. 
\end{lemma}

\begin{proof}
The vertex cut $\cut$ contains not more than $k-1$ vertices, and only these vertices exist in the overlapped part. Therefore, at most $\frac{(k-1)(k-2)}{2}$ incident edges are duplicated. 
\end{proof}

\begin{lemma}
\label{lem:cccnt}
Given a graph $G$ and an integer $k$, for each connected component $C$ obtained by overlapped partition in \refalg{frk}, $|V(C)| \ge k+1$.
\end{lemma}

\begin{proof}
Let $\cut$ denote a vertex cut in an overlapped partition. $C$ is one of the connected components obtained in this partition. Let $H$ denote the vertex set of all vertices in $V(C)$ but not in $\cut$, i.e., $H = \{u | u \in V(C), u\not\in \cut \}$. We have $H\neq \emptyset$.  Note that each vertex in the graph has a degree at least $k$ in $G$ (line 5 in \refalg{frk}). There exist at least $k$ neighbors for each vertex $u$ in $H$ and therefore for each neighbor $v$ of $u$ we have $v \in C$ according to \reflem{nbrsave}. Thus, we have $|V(C)| \ge k+1$.
\end{proof}

\begin{lemma}
\label{lem:cutcnt}
Given a graph $G$ and an integer $k$, the total number of overlapped partitions during the algorithm $\frk$ is no larger than $\frac{n-k-1}{2}$.
\end{lemma}

\begin{proof}
Suppose that $\lambda$ is the total number of overlapped partitions during the whole algorithm $\frk$. This  generates at least $\lambda+1$ connected components. We know from \reflem{cccnt} that each connected component contains at least $k+1$ vertices. Thus, we have at least $(\lambda+1)(k+1)$ vertices in total. 

On the other hand, we increase at most $k-1$ vertices in each subgraph obtained by an overlapped partition according to \reflem{dpcnt}. Thus, at most $\lambda(k-1)$ vertices are added. We obtain the following formula.

\begin{center}
$(\lambda+1)(k+1) \le n+\lambda(k-1)$
\end{center}

Rearranging the formula, we have $\lambda \le \frac{n-k-1}{2}$.
\end{proof}

Next, we prove the upper bound for number of $k$-\vccs.

\begin{theorem}
\label{thm:numkvcc}
Given a graph $G$ and an integer $k$, there are at most $\frac{n}{2}$ $k$-\vccs, i.e., $|\vcc_k(G)| < \frac{|V(G)|}{2}$.
\end{theorem}

\begin{proof}
Similar to the proof of \reflem{cutcnt}, let $\lambda$ be the times of overlapped partitions in the whole algorithm $\frk$. At most $\lambda(k-1)$ vertices are increased. Let $\sigma$ be the number of connected components obtained in all partitions. We have $\sigma > \lambda$. Each connected component contains at least $k+1$ vertices according to \reflem{cccnt}. Note that each connected component is either a $k$-\vcc or a graph that does not contain any $k$-\vcc. Otherwise, the connected component will be further partitioned. Let $x$ be the number of $k$-\vccs and $y$ be the number of connected components that do not contain any $k$-\vcc, i.e., $x+y = \sigma$. We know that a $k$-\vcc contains at least $k+1$ vertices. Thus there are at least $x(k+1)+y(k+1)$ vertices after finishing all partitions. We have following formula.

\begin{center}
$x(k+1)+y(k+1) \le n+\lambda(k-1)$
\end{center}

Since $\lambda < \sigma$ and $\sigma = x+y$, we rearrange the formula as follows.

\begin{center}
$x(k+1)+y(k+1) < n+x(k-1)+y(k-1)$

$x(k+1) < n+x(k-1)$
\end{center}

Therefore, we have $x < \frac{n}{2}$.
\end{proof}

\begin{theorem}
\label{thm:time}
The total time complexity of $\frk$ is $O(\min(n^{1/2},$ $k)\cdot m  \cdot (n+\delta^2) \cdot n)$.
\end{theorem}

\begin{proof}
The total time complexity of $\frk$ is dependent on the number of times $\fcbase$ is invoked. Suppose $\fcbase$ is invoked $p$ times during the whole $\frk$ algorithm, the number of overlapped partitions during the whole $\frk$ algorithm is $p_1$ and the total number of $k$-\vccs is $p_2$. It is easy to see that $p=p_1+p_2$. From \reflem{cutcnt}, we know that $p_1\leq \frac{n-k-1}{2} < \frac{n}{2}$. From \refthm{numkvcc}, we know that $p_2 < \frac{n}{2}$. Therefore, we have $p=p_1+p_2 < n$. According to \reflem{tlcut} and \reflem{lccnt}, the total time complexity of $\frk$ is $O(\min(n^{1/2},$ $k)\cdot m\cdot (n+\delta^2) \cdot n )$.
\end{proof}

\stitle{Discussion.} \refthm{time} shows that all $k$-\vccs can be enumerated in polynomial time. Although the time complexity is still high, it performs much better in practice. Note that the time complexity is the product of three parts:
\begin{itemize} 
\item The first part $O(\min(n^{1/2},$ $k)\cdot m)$ is the time complexity for $\lcut$ to test whether there exists a vertex cut of size smaller than $k$. In practice, the graph to be tested is much smaller than the original graph $G$ since (1) The graph to be tested has been pruned using the $k$-core technique and sparse certification technique. (2) Due to the graph partition scheme, the input graph is partitioned into many smaller graphs. 
\item The second part $O(n+\delta^2)$ is the number of times such that $\lcut$ (local connectivity testing) is invoked by the algorithm $\fcbase$. We will discuss how to significantly reduce the number of local connectivity testings in \refsec{opt}.
\item The third part $O(n)$ is the number of times $\fcbase$ is invoked. In practice, the number can be significantly reduced since the number of $k$-\vccs is usually much smaller than $\frac{n}{2}$.
\end{itemize}

In the next section, we will explore several search reduction techniques to speed up the algorithm.

\section{Search Reduction}
\label{sec:opt}
In the previous section, we introduce our basic algorithm. Recall that in the worst case, we need to test local connectivity between the source vertex $u$ and all other vertices in $G$ using $\lcut$ in$ \fcbase$, and we also need to test local connectivity for every pair of neighbors of $u$. For each pair of vertices, we need to compute the maximum flow in the directed flow graph. Therefore, the key to improving the algorithm is to reduce the number of local connectivity testings ($\lcut$). In this section, we propose several techniques to avoid unnecessary testings. We can avoid testing local connectivity of a vertex pair $(u,v)$ if we can guarantee that $u \equiv^k v$. We call such operation a sweep operation. Below, we introduce two ways to efficiently prune unnecessary testings, namely neighbor sweep and group sweep, in \refsec{opt:ns} and \refsec{opt:gs} respectively.

\subsection{Neighbor Sweep}
\label{sec:opt:ns}

In this section, we propose a neighbor sweep strategy to prune unnecessary local connectivity testings ($\lcut$) in the first phase of $\fcbase$. Generally speaking, given a source vertex $u$, for any vertex $v$, we aim to skip testing the local connectivity of $(u,v)$ according to the information of the neighbors of $v$. Below, we explore two neighbor sweep strategies, namely neighbor sweep using side-vertex and neighbor sweep using vertex deposit.

\subsubsection{Neighbor Sweep using Side-Vertex}
\label{sec:opt:ns:sv}

We first define \textit{side}-\textit{vertex} as follows. 

\begin{definition}
\label{def:sd}
\textsc{(Side-Vertex)} Given a graph $G$ and an integer $k$, a vertex $u$ is called a side-vertex if there does not exist a vertex cut $\cut$ such that $|\cut| < k$ and $u \in \cut$.
\end{definition}

Based on \refdef{sd}, we give the following lemma to show the transitive property regarding the local $k$ connectivity relation $\equiv^k$. 

\begin{lemma}
\label{lem:transitive}
Given a graph $G$ and an integer $k$,  suppose $a \equiv^k b$ and $b \equiv^k c$, we have $a \equiv^k c$ if $b$ is a side-vertex.
\end{lemma}

\begin{proof}
We prove it by contradiction. Assume that $b$ is a side-vertex and $a \not\equiv^k c$. There exists a vertex cut with $k-1$ or fewer vertices between $a$ and $c$. $b$ is not in any such cut since it is a side-vertex. Then we have either $b \not\equiv^k a$ or $b \not\equiv^k c$. This contradicts the precondition that $a \equiv^k b$ and $b \equiv^k c$.
\end{proof}

A wise way to use the transitive property of the local connectivity relation in \reflem{transitive} can largely reduce the number of unnecessary testings.  Consider a selected source vertex $u$ in algorithm $\fcbase$. We assume that $\lcut$ (line~5) returns $\emptyset$ for a vertex $v$, i.e., $u \equiv^k v$. We know from \reflem{transitive} that the vertex pair $(u,w)$ can be skipped for local connectivity testing if $(i)$ $v \equiv^k w$ and $(ii)$ $v$ is a side-vertex. For condition $(i)$, we can use a simple necessary condition according to \reflem{nbrsave}, that is,  for any vertices $v$ and $w$, $v \equiv^k w$ if $(v,w) \in E$. In the following, we focus on condition $(ii)$ and look for necessary conditions to efficiently check whether a vertex is a side-vertex. 



\stitle{Side-Vertex Detection.} To check whether a vertex is a side-vertex, we can easily obtain the following lemma based on \refdef{sd}.

\begin{lemma}
\label{lem:issv}
Given a graph $G$, a vertex $u$ is a side-vertex if and only if $\forall v, v' \in N(u)$, $v \equiv^k v'$.
\end{lemma}

Recall that two vertices are $k$-local connected if they are neighbors of each other. For the $k$-local connectivity of non-connected vertices, we give another necessary condition below.

\begin{lemma}
\label{lem:commonnbr}
Given two vertices $u$ and $v$, $u \equiv^k v$ if $|N(u) \cap N(v)| \ge k$.
\end{lemma}

\begin{proof}
$u$ and $v$ cannot be disjoint after removing any $k-1$ vertices since they have at least $k$ common neighbors. Thus $u$ and $v$ must be $k$-local connected.
\end{proof}

Combining \reflem{issv} and \reflem{commonnbr}, we derive the following necessary condition to check whether a vertex is a  side-vertex.

\begin{theorem}
\label{thm:necessaryside}
A vertex $u$ is a side-vertex if $\forall v, v' \in N(u)$, either $(v, v') \in E$ or $|N(v) \cap N(v')| \ge k$. 
\end{theorem}

\begin{proof}
The theorem can be easily verified using \reflem{nbrsave}, \reflem{issv} and \reflem{commonnbr}.
\end{proof}

\begin{definition}
\label{def:ssv}
\textsc{(Strong Side-Vertex)} A vertex $u$ is called a strong side-vertex if it satisfies the conditions in \refthm{necessaryside}.
\end{definition}

Using strong side-vertex, we can define our first rule for neighbor sweep as follows.

\rtitle{(Neighbor Sweep Rule 1)} \textit{Given a graph $G$ and an integer $k$, let $u$ be a selected source vertex in algorithm $\fcbase$ and $v$ be a strong side-vertex in the graph. We can skip the local connectivity testings of all pairs of $(u,w)$ if we have $u \equiv^k v$ and $w \in N(v)$.}

We give an example to demonstrate \textit{neighbor sweep rule 1} below.

\begin{figure}[t]
\begin{center}
\begin{tabular}[t]{c}
\subfigure[A strong side-vertex $s$]{
	\includegraphics[width=0.49\columnwidth]{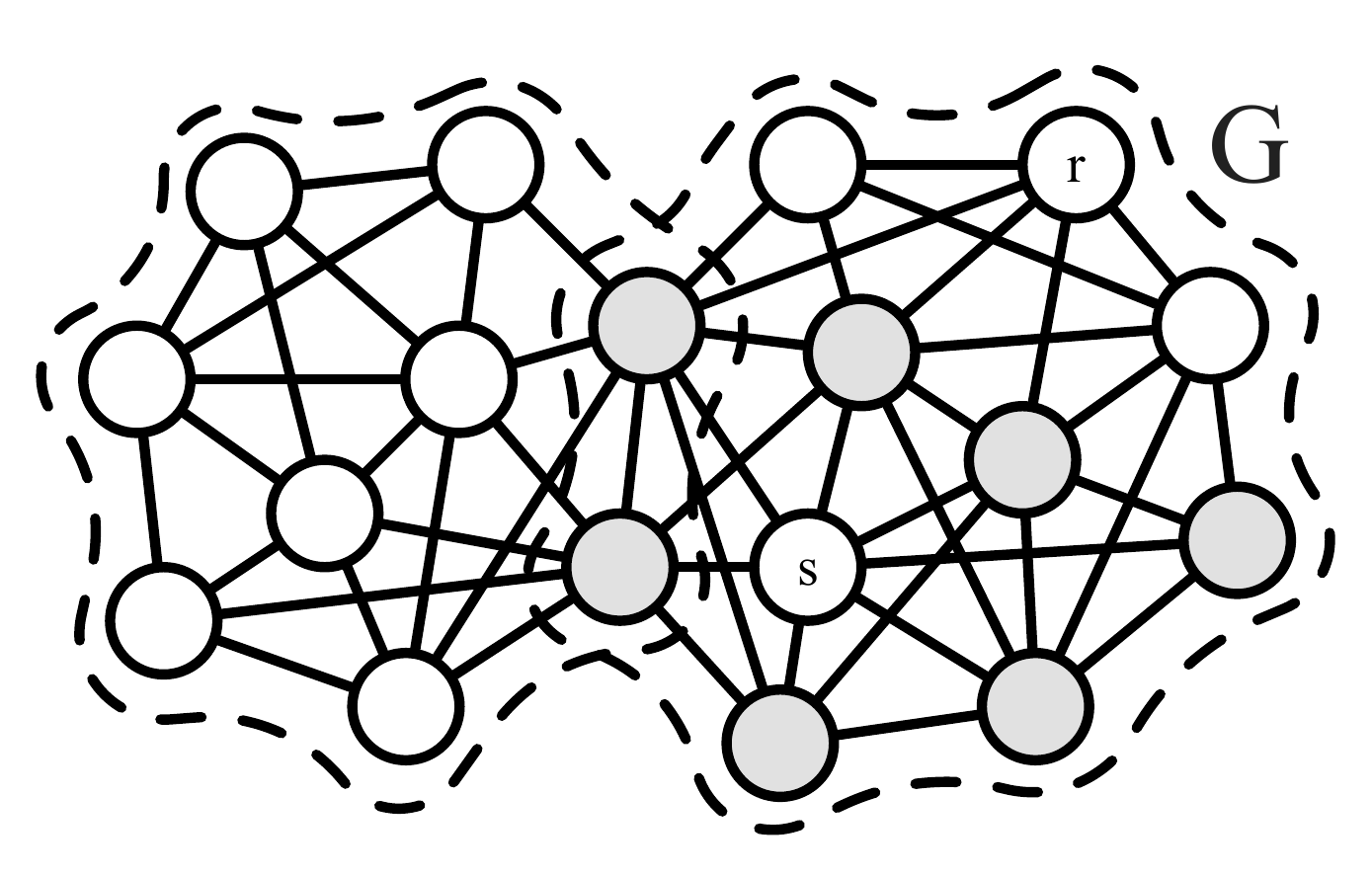}
}
\subfigure[Vertex deposit]{
	\includegraphics[width=0.49\columnwidth]{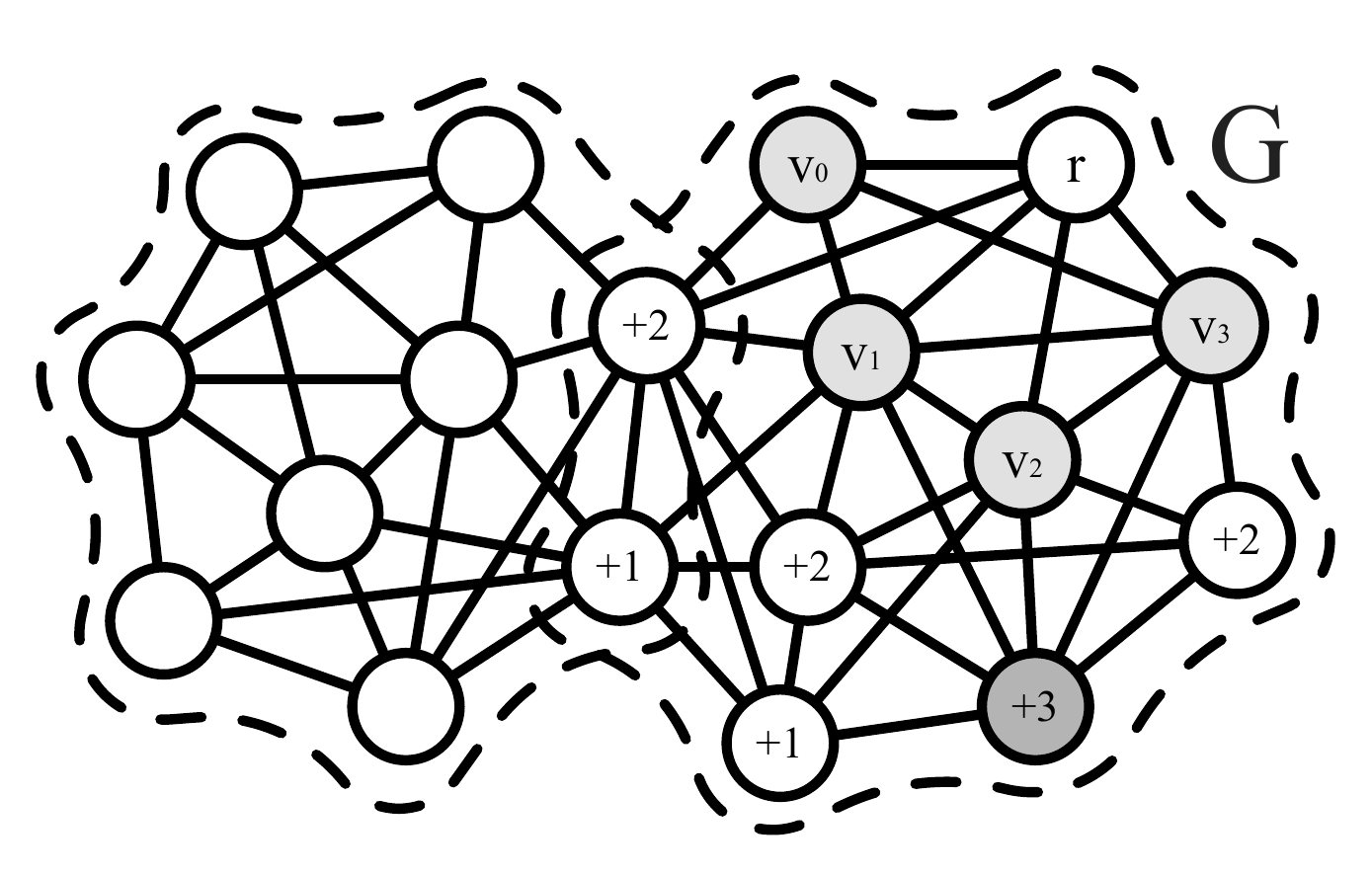}
}
\end{tabular}
\topcaption{Strong side-vertex and vertex deposit when $k=3$}
\label{fig:ssv}
\end{center}
\end{figure}

\begin{example}
\reffig{ssv} (a) presents a strong side-vertex $s$ in graph $G$ while parameter $k = 3$. Assume that $r$ is the source vertex. Any two neighbors of $s$ are either connected by an edge or have at least $3$ common neighbors. If first test the local connectivity between $r$ and $s$ and $r \equiv^k s$, we can safely sweep all neighbors of $s$, which are marked by the gray color in \reffig{ssv} (a).
\end{example}

Below, we discuss how to efficiently detect the strong side-vertices and maintain strong side-vertices while the graph is partitioned in the whole algorithm.

\stitle{Strong Side-Vertex Computation.} Following \refthm{necessaryside}, we can compute all strong side-vertices $v$ in advance and skip all neighbors of $v$ once $v$ is $k$ connected with the source vertex (line~5 in $\fcbase$). We can derive the following lemma.

\begin{lemma}
The time complexity of computing all strong side-vertices in graph $G$ is $O(\sum_{w\in V(G)}d(w)^2)$.
\end{lemma}

\begin{proof}
To compute all strong side-vertices in a graph $G$, we first check all $2$-hop neighbors $v$ for each vertex $u$. Since $v$ and $u$ share a common vertex of $1$-hop neighbor, we can easily obtain all vertices which have $k$ common neighbors with $u$. Any vertex $w$ is considered as $1$-hop neighbor of other vertices $u$ $d(w)$ times. We use $d(w)$ steps to obtain $2$-hop neighbors of $u$ which share a common vertex $w$ with $u$. This phase costs $O(\sum_{w\in V(G)}d(w)^2)$ time.

Now for each given vertex $u$, we have all vertices $v$ sharing $k$ common neighbors with it. For each vertex $w$, we check whether any two neighbors of $w$ have $k$ common neighbors. This phase also costs $O(\sum_{w\in V(G)}d(w)^2)$ time. Consequently, the total time complexity is $O(\sum_{w\in V(G)}d(w)^2)$.
\end{proof}

After computing all strong side-vertices for the original graph $G$, we do not need to recompute the strong side-vertices for all vertices in the partitioned graph from scratch. Instead, we can find possible ways to reduce the number of strong side-vertex checks by making use of the already computed strong side-vertices in $G$. We can do this based on \reflem{sidevertexpart} and \reflem{sidevertexonly} which are used to efficiently detect non-strong side-vertices and strong side-vertices respectively. 


\begin{lemma}
\label{lem:sidevertexpart}
Let $G$ be a graph and $G_i$ be one of the graphs obtained by partitioning $G$ using $\opartition$ in \refalg{frk}, a vertex is a strong side-vertex in $G$ if it is a strong side-vertex in $G_i$. 
\end{lemma}

\begin{proof}
The strong side-vertex $u$ requires at least $k$ common neighbors between any two neighbors of $u$. The lemma is obvious since $G$ contains all edges and vertices in $G_i$.
\end{proof}

From \reflem{sidevertexpart}, we know that a vertex is not a strong side-vertex in $G_i$ if it is not a strong side-vertex in $G$. This property allows us checking limited number of vertices in $G_i$, which is the set of strong side-vertices in $G$.

\begin{lemma}
\label{lem:sidevertexonly}
Let $G$ be a graph, $G_i$ be one of the graphs obtained by partitioning $G$ using $\opartition$ in \refalg{frk}, and $\cut$ is a vertex cut of $G$, for any vertex $v \in V(G_i)$, if $v$ is a strong side-vertex in $G$ and $N(v)\cap \cut =\emptyset$, then $v$ is also a strong side-vertex in $G_i$.
\end{lemma}

\begin{proof}
The qualification of a strong side-vertex of vertex $v$ requires the information about two-hop neighbors of $v$. Vertices in $\cut$ are duplicated when partitioning the graph. Given a strong side-vertex $v$ in $G$, if $N(v)\cap \cut =\emptyset$, the two-hop neighbors of $v$ are not affected by the partition operation, thus the relationships between the vertices in $N(v)$ are not affected by the partition operation. Therefore, $v$ is still a strong side-vertex in $G_i$ according to \refdef{ssv}.
\end{proof}

With \reflem{sidevertexpart} and \reflem{sidevertexonly}, in a graph $G_i$ partitioned from graph $G$ by vertex cut $\cut$, we can reduce the scope of strong side-vertex checks from the vertices in the whole graph $G_i$ to the vertices $u$ satisfying following two conditions simultaneously: 

\begin{itemize}
\item $u$ is a strong side-vertex in $G$; and
\item $N(u) \cap \cut \neq \emptyset$.
\end{itemize}

\subsubsection{Neighbor Sweep using Vertex Deposit}
\label{sec:opt:ns:vd}

\stitle{Vertex Deposit.} The strong side-vertex strategy heavily relies on the number of strong side-vertices. Next, we investigate a new strategy called vertex deposit, to further sweep vertices based on neighbor information. We first give the following lemma:

\begin{lemma}
\label{lem:2hop}
Given a source vertex $u$ in graph $G$, for any vertex $v\in V(G)$, we have $u\equiv^k v$ if there exist $k$ vertices $w_1, w_2, \ldots, w_k$ such that $u\equiv^k w_i$ and $w_i\in N(v)$ for any $1\leq i \leq k$.
\end{lemma}

\begin{proof}
We prove it by contradiction. Assume that $u \not\equiv^k v$. There exists a vertex cut $\cut$ with $k-1$ or fewer vertices between $u$ and $v$.  For any $w_i (1\leq i \leq k)$, we have $w_i\equiv^k v$ since $w_i\in N(v)$ (\reflem{nbrsave}) and we also have $w_i \equiv^k u$. Since $u \not\equiv^k v$, $w_i$ cannot satisfy both $w_i \equiv^k u$ and $w_i \equiv^k v$ unless $w_i \in \cut$. Therefore, we obtain a cut $\cut$ with at least $k$ vertices $w_1$, $w_2$, $\ldots$, $w_k$. This contradicts $|\cut| < k$.
\end{proof}


Based on \reflem{2hop}, given a source vertex $u$, once we find a vertex $v$ with at least $k$ neighbors $w_i$ with $u\equiv^k w_i$, we can obtain $u\equiv^k v$ without testing the local connectivity of $(u,v)$. To efficiently detect such vertices $v$, we define the deposit of a vertex $v$ as follows.

\begin{definition} 
\label{def:deposit}
(Vertex Deposit) Given a source vertex $u$, the deposit for each vertex $v$, denoted by $deposit(v)$, is the number of neighbors $w$ of $v$ such that the local connectivity of $w$ and $u$ has been computed with $w\equiv^k u$.
%
\end{definition}

According to \refdef{deposit}, suppose $u$ is the source vertex and for each vertex $v$, $deposit(v)$ is a dynamic value depending on the number of processed vertex pairs. To maintain the vertex deposit, we initialize the deposit to $0$ and once we know $w \equiv^k u$ for a certain vertex $w$, we can increase the deposit $deposit(v)$ for each vertex $v\in N(w)$ by $1$. We can obtain the following theorem according to \reflem{2hop}.

\begin{theorem}
\label{thm:deposit}
Given a source vertex $u$, for any vertex $v$, we have $u \equiv^k v$ if $deposit(v) \ge k$.
\end{theorem}


Based on \refthm{deposit}, we can derive our second rule for neighbor sweep as follows.

\rtitle{(Neighbor Sweep Rule 2)} \textit{Given a selected source vertex $u$, we can skip the local connectivity testing of pair $(u,v)$ if $deposit(v) \ge k$.}

We show an example below.

\begin{example}
\reffig{ssv} (b) gives an example of our vertex deposit strategy. Given the graph $G$ and parameter $k = 3$, let vertex $r$ be the selected source vertex. We assume that $v_0, v_1, v_2$ and $v_3$ are tested vertices. All these vertices are local $k$-connected with vertex $r$, i.e., $r \equiv^k v_i, i\in{0,1,2,3}$, since $v_0, v_1, v_2$ and $v_3$ are neighbors of $r$. We deposit once for the neighbors of each tested vertex. The deposit value for all influenced vertices are given in the figure. We mark the vertices with deposit no less than $3$ by dark gray. The local connectivity testing between $r$ and such a vertex can be skipped. 
\end{example}

To increase the deposit of a vertex $v$, we only need any neighbor of $v$ is local $k$-connected with the source vertex $u$. We can also use vertex deposit strategy when processing the strong side-vertex. Given a source vertex $u$ and a strong side-vertex $v$, we sweep all $w \in N(v)$ if $u \equiv^k v$ according to the side-vertex strategy. Next we increase the deposit for each non-swept vertex $w' \in N(w)$. In other words, for a strong side-vertex, we can possibly sweep its $2$-hop neighbors by combining the two neighbor sweep strategies. An example is given below.

\begin{figure}[t]
\begin{center}
\begin{tabular}[t]{c}
\subfigure[$2$-hop deposit]{
	\includegraphics[width=0.49\columnwidth]{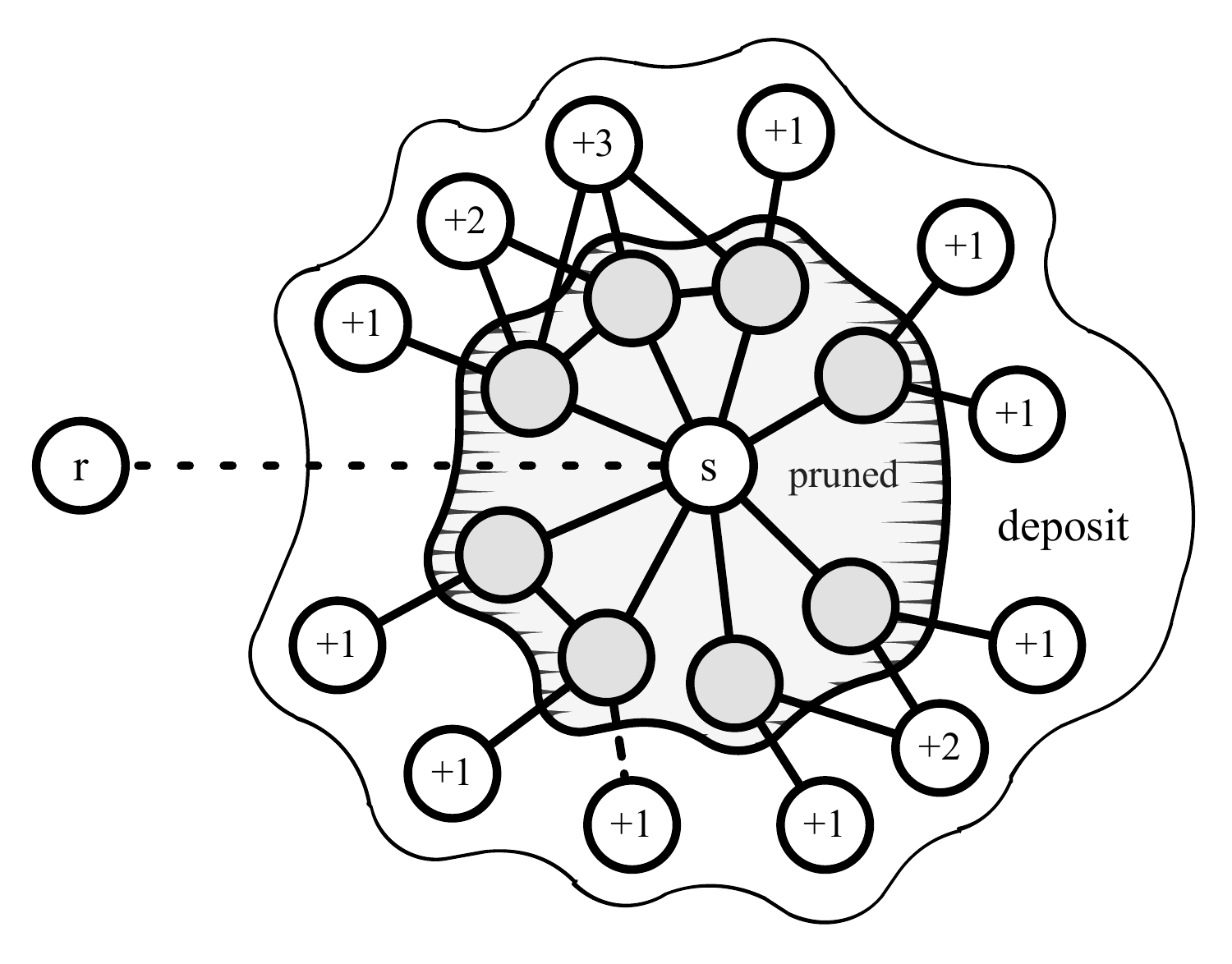}
}
\subfigure[group sweep and deposit]{
	\includegraphics[width=0.49\columnwidth]{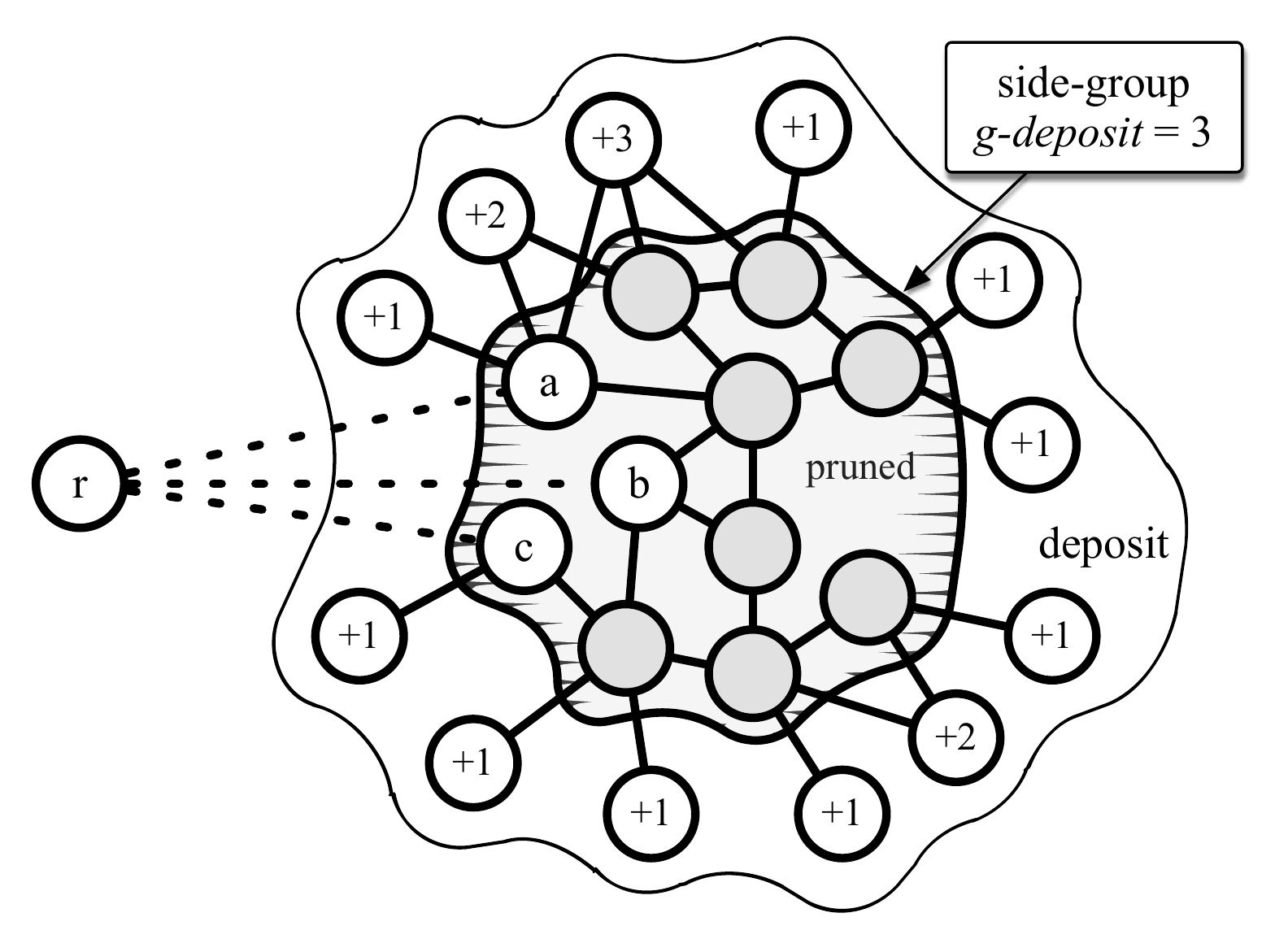}
}
\end{tabular}
\topcaption{Increasing deposit with neighbor and group sweep}
\label{fig:gsd}
\end{center}
\end{figure}

\begin{example}
\reffig{gsd} (a) shows the process for a strong side-vertex $s$. Given a source vertex $r$, assume $s$ is a strong side-vertex and $r \equiv^k s$. All neighbors of $s$ are swept and all $2$-hop neighbors of $s$ increase their deposits accordingly. The increased value of the deposit for each vertex depends on the number of connected vertices that are swept.
\end{example}

\subsection{Group Sweep}
\label{sec:opt:gs}

The neighbor sweep strategy can only prune unnecessary local connectivity testings  in the first phase of $\fcbase$ by using the neighborhood information. In this subsection, we introduce a new pruning strategy, namely group sweep, which can prune unnecessary local connectivity testings in a batch manner. In group sweep, we do not limit the skipped vertices to the neighbors of certain vertices. More specifically, we aim to partition vertices into vertex groups and sweep a whole group when it satisfies certain conditions. In addition, our group sweep strategy can also be applied to reduce the unnecessary local connectivity testings in both phases of $\fcbase$.

First, we define a new relation regarding a vertex $u$ and a set of vertices $C$ as follows. 

\sstitle{$u \equiv^k C$}: For all vertices $v \in C, u \equiv^k v$.

Given a source vertex $u$ and a side-vertex $v$, we assume $u \equiv^k v$. According to the transitive relation in \reflem{transitive}, we can skip testing the pairs of vertices $u$ and $w$ for all $w$ with $w \equiv^k v$. In our neighbor sweep strategy, we select all neighbors of $v$ as such vertices $w$, i.e., $u \equiv^k N(v)$. To sweep more vertices each time, we define the side-group.

\begin{definition}
\textsc{(Side-Group)} Given a graph $G$ and an integer $k$, a vertex set ${\cal CC}$ in $G$ is a side-group if $\forall u, v \in C, u \equiv^k v$.
\end{definition}

Note that it is possible that a side-group contains vertices in a certain vertex cut $\cut$ with $|\cut|<k$.  Next, we introduce how to construct the side-groups in graph $G$, and then discuss our group sweep rules.

\stitle{Side-Group Construction.} \refsec{base:cert} introduces sparse certificate to bound the graph size. Let $F_i$ and $G_i$  be the notations defined in \refthm{maincert}. Assume that $G$ is not $k$-connected and there exists a vertex cut $\cut$ such that $|\cut|<k$. According to \cite{CheriyanKT93}, we have the following lemma.

\begin{lemma}
\label{lem:nopath}
$F_k$ does not contain a simple tree path $P_k$ whose two end points are in different connected components of $G-\cut$. 
\end{lemma}

Based on \reflem{nopath}, we can obtain the following theorem.

\begin{theorem}
\label{thm:cc}
Let ${\cal CC}$ denote the vertex set of any connected component in $F_k$. ${\cal CC}$ is a side-group.
\end{theorem}

\begin{proof}
Assume that $u \not\equiv^k v$ in ${\cal CC}$. All simple paths from $u$ to $v$ will cross the vertex cut $\cut$. This contradicts \reflem{nopath}.
\end{proof}

\begin{example}
Review the construction of a sparse certificate in \reffig{sc}. Given $k=3$, two connected components with more than one vertex are obtained in $F_3$. The number of vertices in the two connected components are $6$ and $9$ respectively.  Each of them is a side-group and any two vertices in the same connected component is local $3$-connected. Note that the connected component with $6$ vertices contains two vertices in the vertex cut as marked by gray. 
\end{example}

We denote all the side-groups as ${\cal CS}=\{{\cal CC}_1, {\cal CC}_2, \ldots, {\cal CC}_t\}$.  According to \refthm{cc}, ${\cal CS}$ can be easily computed as a by-product of the sparse certificate. With ${\cal CS}$, according to the transitive relation in \reflem{transitive}, we can easily obtain the following pruning rule.

\rtitle{(Group Sweep Rule 1)} \textit{Let $u$ be the source vertex in the algorithm $\fcbase$, given a side-group ${\cal CC}$, if there exists a strong side-vertex $v\in {\cal CC}$ such that $u\equiv^k v$, we can skip the local connectivity testings of vertex pairs $(u,w)$ for all $w\in {\cal CC}-\{v\}$.}

The above group sweep rule relies on the successful detection of a strong side-vertex in a certain side-group. In the following, we further introduce a deposit based scheme to handle the scenario that no strong side-vertex exists in a certain side-group. 

\stitle{Group Deposit.} Similar with the vertex deposit strategy, the group deposit strategy aims to deposit the values in a group level. To show our group deposit scheme, we first introduce the following lemma.

\begin{lemma}
\label{lem:ccprune}
Given a source vertex $u$, an integer $k$, and a side-group ${\cal CC}$, we have $u \equiv^k {\cal CC}$ if  $|\{v|v \in {\cal C}, u \equiv^k v\}| \ge k$.
\end{lemma}

\begin{proof}
We prove it by contradiction. Assume that there exists a vertex $w$ in ${\cal CC}$ such that $u \not\equiv^k w$. A vertex cut $\cut$ exists with $|\cut|<k$. Let $v_0, v_1, ..., v_{k-1}$ be the $k$ vertices in ${\cal CC}$ such that $u \equiv^k v_i, 0 \le i \le k-1$. We have $w \equiv^k v_i$ based on the definition of a side-group. Each $v_i$ must belong to $\cut$ since $u \not\equiv^k w$. As a result, the size of $\cut$ is at least $k$. This contradicts $|\cut|<k$.
\end{proof}

Based on \reflem{ccprune}, given a source vertex $u$, once we find a side-group ${\cal CC}$ with at least $k$ vertices $v$ with $u\equiv^k v$, we can get $u\equiv^k {\cal CC}$ without testing the local connectivity from $u$ to other vertices in ${\cal CC}$. To efficiently detect such side-groups ${\cal CC}$, we define the group deposit of a side-group ${\cal CC}$ as follows.

\begin{definition}
\label{def:gdeposit}
(Group Deposit) The group deposit for each side-group ${\cal CC}$, denoted by $\gdep({\cal CC})$, is the number of vertices $v\in {\cal CC}$ such that the local connectivity of $v$ and $u$ has been computed with $v\equiv^k u$.
\end{definition}

According to \refdef{gdeposit}, suppose $u$ is the source vertex, for each side-group ${\cal CC}\in {\cal CS}$, $\gdep({\cal CC)}$ is a dynamic value depending on the already processed vertex pairs. To maintain the group deposit for each side-group ${\cal CC}$, we initialize the group deposit for ${\cal CC}$ to $0$. Once $v\equiv^k u$ for a certain vertex $v\in {\cal CC}$, we can increase $\gdep({\cal CC})$ by 1. We obtain the following theorem according to \reflem{ccprune}. 

\begin{theorem} 
\label{thm:ccprune}
Given a source vertex $u$, for any side-group ${\cal CC}\in {\cal CS}$, we have $u\equiv^k{\cal CC}$ if $\gdep({\cal CC})\geq k$.
\end{theorem}

Based on \refthm{ccprune}, we can derive our second rule for group sweep as follows.

\rtitle{(Group Sweep Rule 2)} \textit{Given a selected source vertex $u$, we can skip the local connectivity testings between $u$ and vertices in ${\cal CC}$ if $\gdep({\cal CC})\geq k$.}

Note that a group sweep operation can further trigger a neighbor sweep operation and vice versa, since both operations result in new local $k$-connected vertex pairs. We show an example below.

\begin{example}
\reffig{gsd} (b) presents an example of group sweep. Suppose $k = 3$ and the gray area is a detected side-group. Given a source vertex $r$, assume that $a, b, c$ are the tested vertices with $r \equiv^k a, r \equiv^k b$ and $r \equiv^k c$ respectively. According to \refthm{ccprune}, we can safely sweep all vertices in the same side-group. Also, we apply the vertex deposit strategy for neighbors outside the side-group. The increased value of deposit is shown on each vertex.
\end{example}

Next we show that the side-groups can also be used to prune the local connectivity testings in the second phase of $\fcbase$. Recall that in the second phase of $\fcbase$, given a source vertex $u$, we need to test the local connectivity of every pair $(v_a,v_b)$ of the neighbors of $u$. With side-groups, we can easily obtain the following group sweep rule.

\rtitle{(Group Sweep Rule 3)} \textit{Let $u$ be the source vertex, and $v_a$ and $v_b$ be two neighbors of $u$. If $v_a$ and $v_b$ belong to the same side-group, we have $v_a\equiv^k v_b$ and thus we do not need to test the local connectivity of $(v_a, v_b)$ in the second phase of $\fcbase$.}

The detailed implementation of the neighbor sweep and group sweep techniques is given in the following section.

\subsection{The Overall Algorithm}
\label{sec:opt:overall}

In this section, we combine our pruning strategies and give the implementation of optimized algorithm $\fcstar$. The pseudocode is presented in \refalg{findcutstar}. We can replace $\fcbase$ with $\fcstar$ in $\frk$ to obtain our final algorithm to compute all $k$-\vccs. 

The $\fcstar$ algorithm still follows the similar idea of $\fcbase$ that consider a source vertex $u$, and then compute the vertex cut  in two phases based on whether $u$ belongs to the vertex cut $\cut$. Given a source vertex $u$, phase 1 (line~8-15) considers the case that $u \notin \cut$. Phase 2 (line~16-21) considers the case that $u\in \cut$. If in both phase, the vertex cut $\cut$ is not found, there is no such a cut and we simply return $\emptyset$ in line~22.

We compute the side-groups ${\cal CS}$ while computing the sparse certificate (line~1). Note that here we only consider the side-group whose size is larger than $k$, since the group can be swept only if at least $k$ vertices in the group are swept according to \refthm{ccprune}. Then we compute all strong side-vertices, ${\cal SV}$ based on \refthm{necessaryside} (line~3). Here, the strong side-vertices are computed based on the method discussed in \refsec{opt:ns:sv}. If ${\cal SV}$ is not empty, we can select one inside vertex as source vertex $u$ and do not need to consider the phase 2, because $u$ cannot be in any cut $\cut$ with $|\cut|<k$ in this case. Otherwise, we still select the source vertex $u$ with the minimum degree (line~4-7).

In phase 1 (line~8-15), we initialize the group deposit for each side-group, which is number of swept vertices in the side-group, to $0$ (line~8). Also, we initialize the local deposit for each vertex to $0$ and $pru$ for each vertex to $\textit{false}$ (line~9). Here, $pru$ is used to mark whether a vertex can be swept. Since the source vertex $u$ is local $k$-connected with itself, we first apply the sweeping rules on the source vertex by invoking $\prufunc$ procedure (line~10). Intuitively, a vertex that is close to the source vertex $u$ tends to be in the same $k$-\vcc with $u$. In other words, a vertex $v$ that is far away from $u$ tends to be separated from $u$ by a vertex cut $\cut$. Therefore, we process vertices $v$ in $G$ according to the non-ascending order of $dist(u,v,G)$ (line~11). We aim to find the vertex cut by processing as few vertices as possible.  
For each vertex $v$ to be processed in phase 1, we skip it if $pru(v)$ is $\text{true}$ (line~12). Otherwise, we test the local connectivity of $u$ and $v$ using $\lcut$ (line~13). If there is a cut $\cut$ with size smaller than $k$, we simply return $\cut$ (line~15). Otherwise, we invoke $\prufunc$ procedure to sweep vertices using the sweep rules introduced in \refsec{opt}. We will introduce the $\prufunc$ procedure in detail later. 

In phase 2 (line~16-21), we first check whether the source vertex $u$ is a strong side-vertex. If so, we can skip phase 2 since a strong side-vertex is not contained in any vertex cut with size smaller than $k$. Otherwise, we perform pair-wise local connectivity testings for all vertices in $N(u)$. Here, we apply the \textit{group sweep rule 3} and skip testing those pairs of vertices that are in the same side-group (line~19).

\stitle{Procedure $\prufunc$.} The procedure $\prufunc$ is shown in \refalg{sweep}. To sweep a vertex $v$, we set $pru(v)$ to be true. This operation may result in neighbor sweep and group sweep of other vertices as follows.

\begin{itemize}
\item (Neighbor Sweep) In line~1-5, we consider the neighbor sweep. For all the neighbors $w$ of $v$ that have not been swept, we first increase $deposit(w)$ by $1$ based on \refdef{deposit}. Then we consider two cases. The first case is that $v$ is a strong side-vertex. According to \textit{neighbor sweep rule 1} in \refsec{opt:ns:sv}, $w$ can be swept since $w$ is a neighbor of $v$. The second case is $deposit(w)>k$. According to \textit{neighbor sweep rule 2} in \refsec{opt:ns:vd}, $w$ can be swept. In both cases, we invoke $\prufunc$ to sweep $w$ recursively. (line~4-5) 
\item (Group Sweep) In line~6-11, we consider the group sweep if $v$ is contained in a side-group ${\cal CC}_i$. We first increase $\gdep({\cal CC}_i)$ by $1$ based on \refdef{gdeposit}. Then we consider two cases. The first case is that $v$ is a strong side-vertex. According to \textit{group sweep rule 1} in \refsec{opt:gs}, we can sweep all vertices in ${\cal CC}_i$. The second case is that $\gdep\geq k$. According to \textit{group sweep rule 2} in \refsec{opt:gs}, we can sweep all vertices in ${\cal CC}_i$. In both cases we recursively invoke $\prufunc$ to sweep each unswept vertex in ${\cal CC}_i$ (line~8-11).
\end{itemize}

\begin{algorithm}[h]
\caption{$\fcstar(G,k)$}
\label{alg:findcutstar}
\small
\begin{algorithmic}[1]
\REQUIRE a graph $G$ and an integer $k$;
\ENSURE a vertex cut with size smaller than $k$;
\vspace{2ex}
\STATE compute a sparse certification ${\cal SC}$ of $G$ and collect all side-groups as ${\cal CS} = \{{\cal CC}_1,...,{\cal CC}_t\}$;
\STATE construct the directed flow graph $\overline{\cal SC}$ of ${\cal SC}$;
\STATE ${\cal SV} \leftarrow $ compute all strong side vertices in ${\cal SC}$;
\IF{${\cal SV} = \emptyset$}
	\STATE select a vertex $u$ with minimum degree;
\ELSE
	\STATE randomly select a vertex $u$ from ${\cal SV}$;
\ENDIF
\vspace{2ex}
\STATE \textbf{for all} ${\cal CC}_i$ in ${\cal CS}$: $\gdep({\cal CC}_i) \leftarrow 0$;
\STATE \textbf{for all} $v$ in $V$: $deposit(v) \leftarrow 0, pru(v) \leftarrow \text{false}$;
\STATE $\prufunc(u, pru, deposit, \gdep, {\cal CS})$;

\FORALL{$v\in V$ in non-ascending order of $dist(u,v,G)$}
	\STATE \textbf{if} $pru(v) = \text{true}$ \textbf{then continue}; 
	\STATE $\cut \leftarrow \lcut(u,v,\overline{\cal SC},{\cal SC})$;
	\STATE \textbf{if} $\cut \neq \emptyset$ \textbf{then} \textbf{return} $\cut$;
	\STATE $\prufunc(v, pru, deposit, \gdep, {\cal CS})$;

\ENDFOR

\vspace{2ex}
\IF{$u$ is not a strong side-vertex}
	\FORALL{$v_a \in N(u)$}
		\FORALL{$v_b \in N(u)$}
			\STATE \textbf{if} $v_a$ and $v_b$ are in the same ${\cal CC}_i$  \textbf{then continue};
			\STATE $\cut \leftarrow \lcut(u,v,\overline{\cal SC},{\cal SC})$;
			\STATE \textbf{if} $\cut \neq \emptyset$ \textbf{then} \textbf{return} $\cut$;
		\ENDFOR
	\ENDFOR
\ENDIF
\vspace{2ex}
\STATE \textbf{return} $\emptyset$;

\end{algorithmic}
\end{algorithm}

\begin{algorithm}[h]
\caption{$\prufunc(v, pru, deposit, \gdep, {\cal CS})$}
\label{alg:sweep}
\small
\begin{algorithmic}[1]

\STATE $pru(v) \leftarrow \text{true}$;

\FORALL{$w \in N(v)$ s.t. $pru(w)=\text{false}$}
	\STATE $deposit(w)$++;
	\IF{ $v$ is a strong side-vertex \textbf{or} $deposit(w)\geq k$ }
		\STATE $\prufunc(w, pru, deposit, \gdep, {\cal CS})$;
	\ENDIF
\ENDFOR
	
\IF{$v$ is contained in a ${\cal CC}_i$ and ${\cal CC}_i$ has not been processed}
	\STATE $\gdep({\cal CC}_i)$++;
	\IF{$v$ is a strong side-vertex  \textbf{or} $\gdep({\cal CC}_i)\geq k$}
		\STATE mark ${\cal CC}_i$ as processed;
		\FORALL{$w \in {\cal CC}_i$ s.t. $pru(w)=\text{false}$}
			\STATE $\prufunc(w, pru, deposit, \gdep, {\cal CS})$
		\ENDFOR
	\ENDIF
\ENDIF
\end{algorithmic}
\end{algorithm}

\section{Experiments}
\label{sec:exp}
In this section, we experimentally evaluate the performance of our proposed algorithms. 

All algorithms are implemented in C++ using gcc complier at -O3 optimization level. All the experiments are conducted under a Linux operating system running on a machine with an Intel Xeon 3.4GHz CPU, 32GB 1866MHz DDR3-RAM. The time cost of algorithms is measured as the amount of wall-clock time elapsed during program execution.

\stitle{Datasets.} We use 7 publicly available real-world networks to evaluate the algorithms. The network statistics is shown in \reftab{datasets}. 

\begin{table}[t]
\topcaption{\sc Network Statistics} \label{tab:datasets}
\begin{center}
{\small
\begin{tabular}{l|c|c|c|c}
\hline
{\bf Datasets} & $|V|$ & $|E|$ & Density & Max Degree \\\hline\hline
\sfd & 281,903 & 2,312,497 & 8.20 & 38,625 \\
\dblp & 317,080 & 1,049,866 & 3.31 & 343 \\
\cnr & 325,557 & 3,216,152 & 9.88 & 18,236 \\
\nd & 325,729 & 1,497,134 & 4.60 & 10,721 \\
\google & 875,713 & 5,105,039 & 5.83 & 6,332 \\
\cit & 3,774,768 & 16,518,948 & 4.38 & 793 \\\hline
\end{tabular}
}
\end{center}
\end{table}

\sfd is a web graph where vertices represent pages from Stanford University (\url{stanford.edu}) and edges represent hyperlinks between them. \dblp is a co-authorship network of DBLP. \cnr is a small crawl of the Italian CNR domain. \nd is a web graph where vertices represent pages from University of Notre Dame (\url{nd.edu}) and edges represent hyperlinks between them. \google is a web graph from Google Programming Contest. \ytb is a social network from the video-sharing web site Youtube. \cit is a citation network maintained by National Bureau of Economic Research. All datasets can be downloaded from SNAP\footnote{\small \url{http://snap.stanford.edu/index.html}}.

\subsection{Effectiveness Evaluation}
\label{subsec:exp:qe}

We adopt the following three quality measures for effectiveness evaluation:

\begin{itemize}
\item\textit{Diameter $diam$.} The diameter definition is shown in \refeq{diameter}.

\item\textit{Edge Density $\rho_e$.} Edge density is the ratio of the number of edges in a graph to the number of edges in a complete graph with the same set of vertices. Formal equation is given as follows:

\begin{equation}
\label{eq:edge_density}
\small
\rho_e(g)=\frac{2 |E(g)|}{|V(g)| \cdot (|V(g)|-1)}
\end{equation}

\item\textit{Clustering Coefficient ${\cal C}$.} The local clustering coefficient $c(u)$ for a vertex $u$ is the ratio of the number of triangles containing $u$ to the number of triples centered at $u$, which is defined as:

\begin{equation}
\label{eq:local_cc}
\small
c(u)=\frac{|\{(v,w) \in E| v \in N(u), w \in N(u)\}|}{|N(u)| \cdot (|N(u)|-1) / 2}
\end{equation}

The clustering coefficient of a graph is the average local clustering coefficient of all vertices:

\begin{equation}
\label{eq:cc}
\small
{\cal C}(G) = \frac{1}{|V|}\sum_{u \in V}c(u)
\end{equation}
\end{itemize}

In our effectiveness testings, given a graph $G$ and a parameter $k$, we calculate the diameter, edge density and clustering coefficient respectively for every $k$-\vcc of $G$. We show the average value of all $k$-\vccs for each parameter $k$. We compute the same statistics for all $k$-\eccs and $k$-\ccs of $G$ as comparisons.

\begin{figure}[t!]
\begin{center}
\begin{tabular}[h]{c}
\hspace{-1em}
\subfigure[\ytb]{
	\includegraphics[width=0.5\columnwidth]{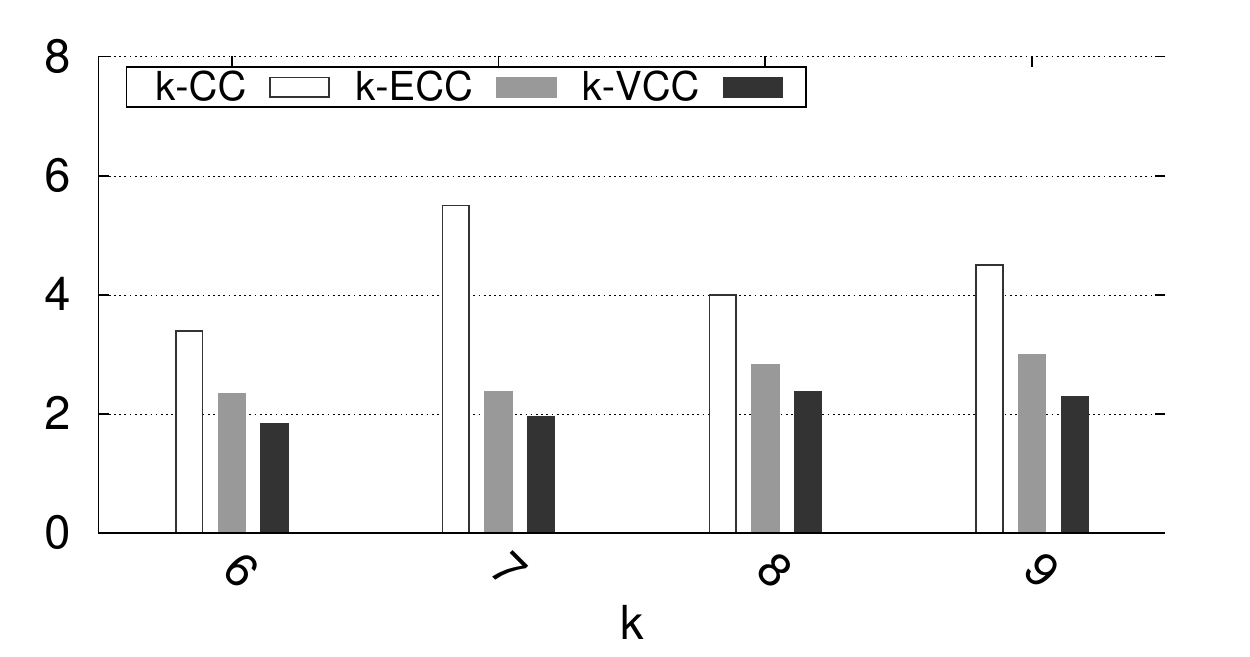}
}
\hspace{-1em}
\subfigure[\dblp]{
	\includegraphics[width=0.5\columnwidth]{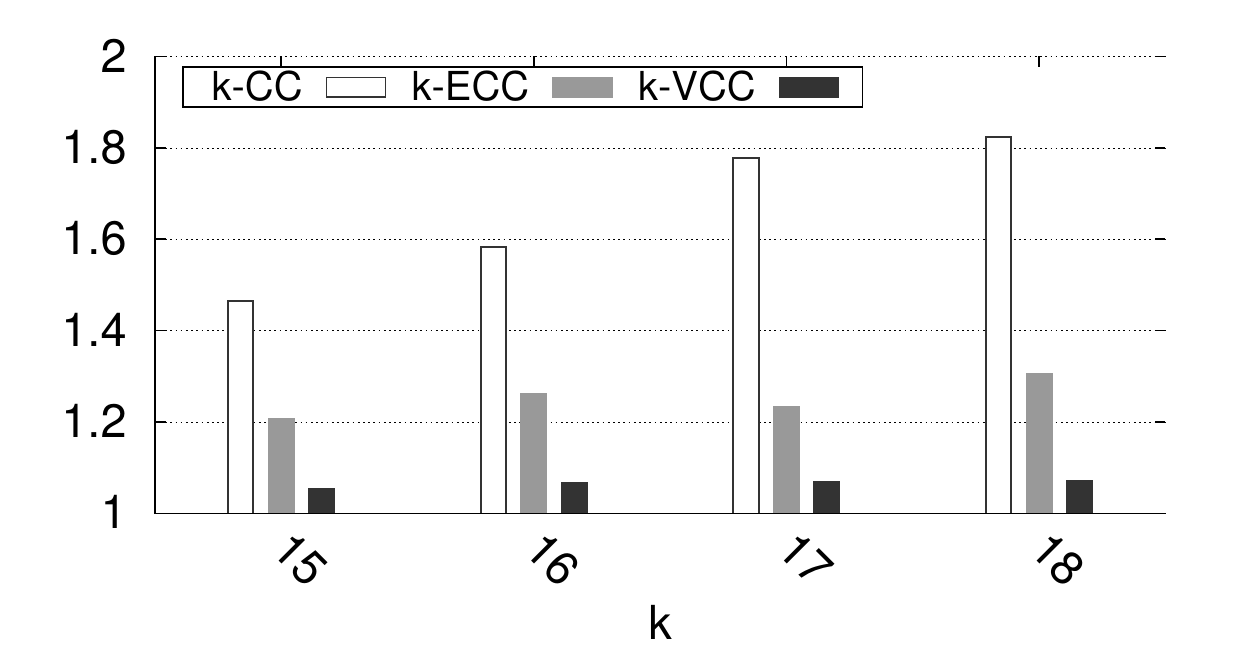}
}
\\
\hspace{-1em}
\subfigure[\google]{
	\includegraphics[width=0.5\columnwidth]{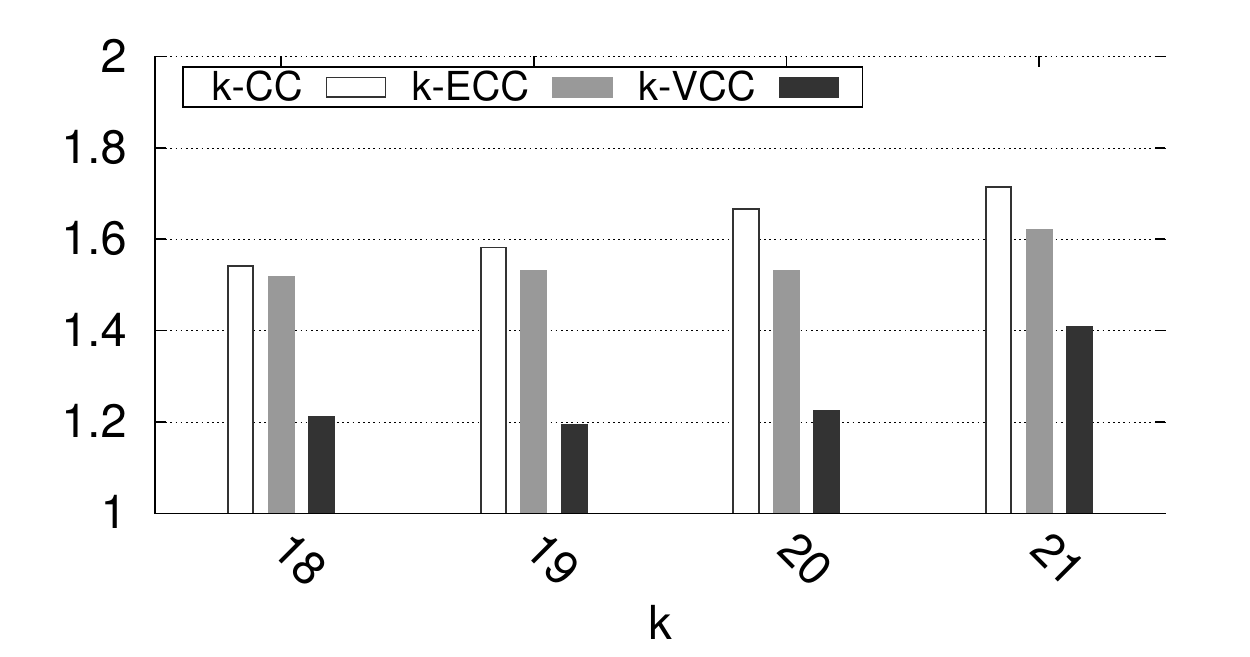}
}
\hspace{-1em}
\subfigure[\cnr]{
	\includegraphics[width=0.5\columnwidth]{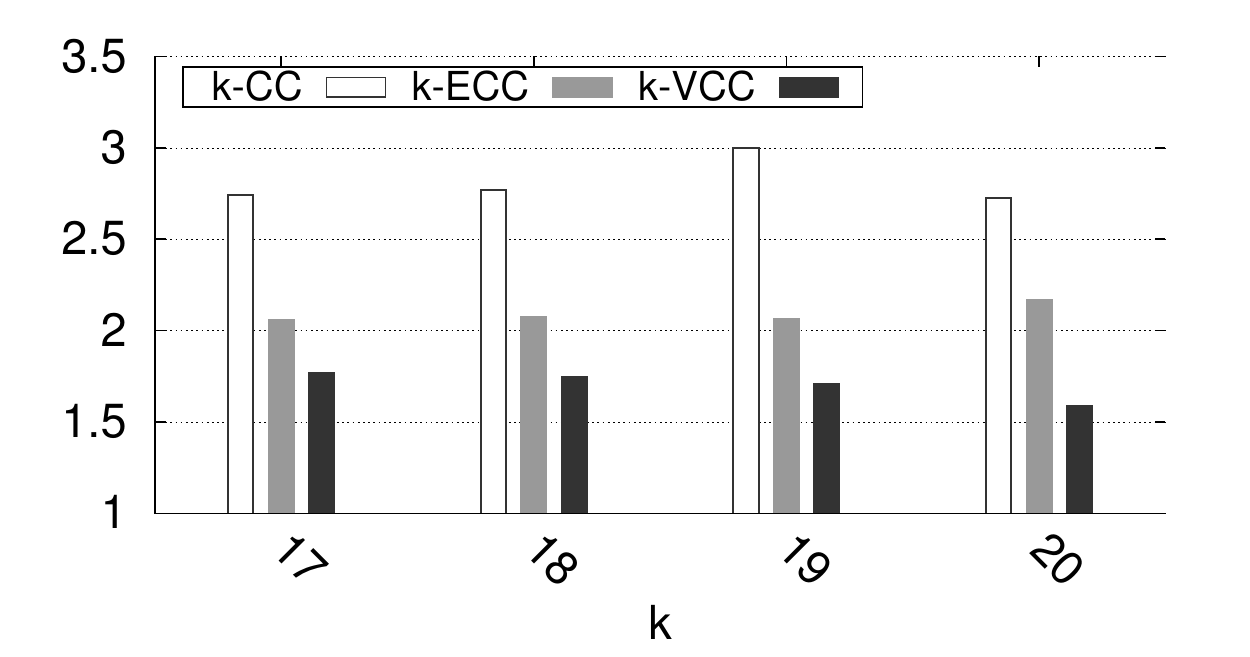}
}
\end{tabular}
\end{center}
\topcaption{Average Diameter}
\label{fig:exp:ad}
\end{figure}

\reffig{exp:ad} presents the average diameter of all $k$-\ccs, $k$-\eccs and $k$-\vccs under the different parameter $k$ in real datasets. Similarly, \reffig{exp:aed} and \reffig{exp:cco} give the statistics of edge density and clustering coefficient respectively. We choose four datasets \ytb, \dblp, \google, and \cnr as representatives in this experiment. In the experimental results, we can see that for the same parameter $k$ value, $k$-\vccs have the smallest average diameter, the largest average edge density and the largest clustering coefficient in all three tested metrics. The result shows that our $k$-\vccs are more cohesive than the $k$-\eccs and $k$-\ccs.

\begin{figure}[t!]
\begin{center}
\begin{tabular}[h]{c}
\hspace{-1em}
\subfigure[\ytb]{
	\includegraphics[width=0.5\columnwidth]{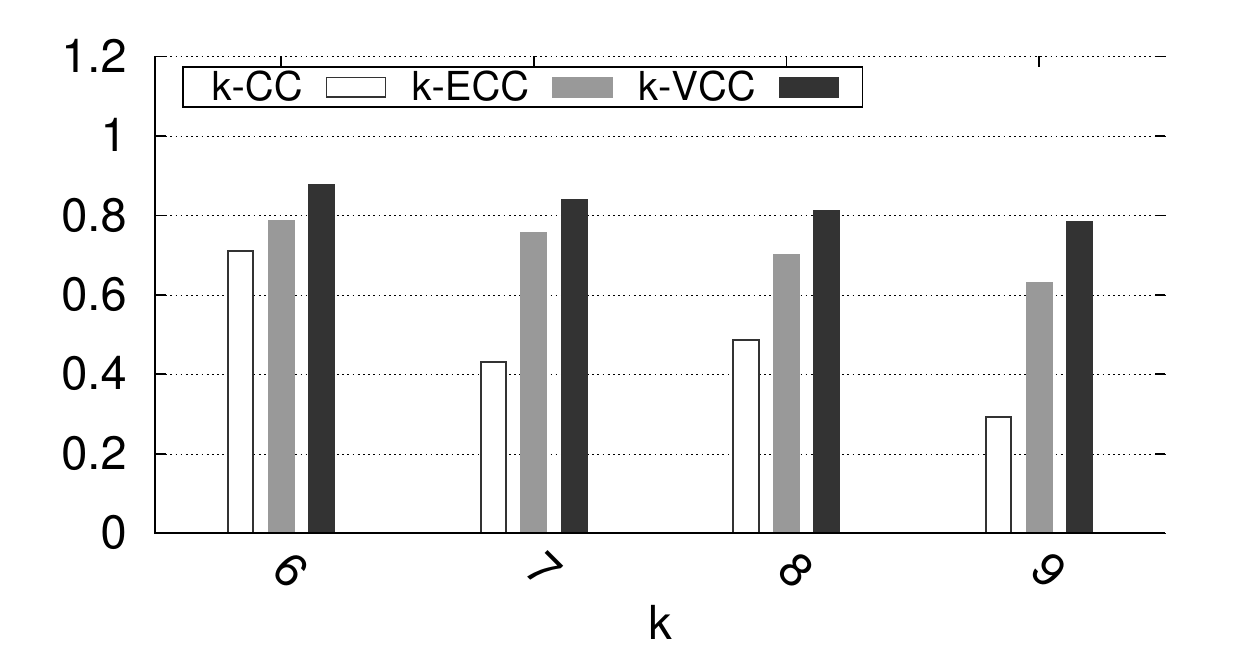}
}
\hspace{-1em}
\subfigure[\dblp]{
	\includegraphics[width=0.5\columnwidth]{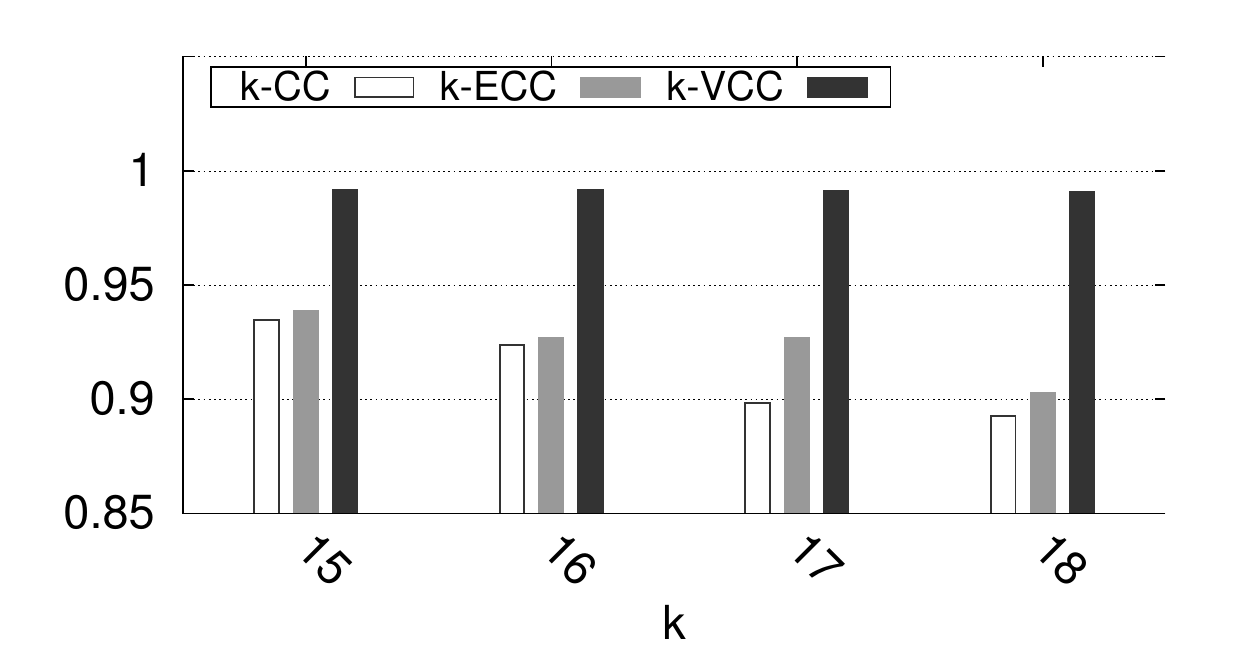}
}
\\
\hspace{-1em}
\subfigure[\google]{
	\includegraphics[width=0.5\columnwidth]{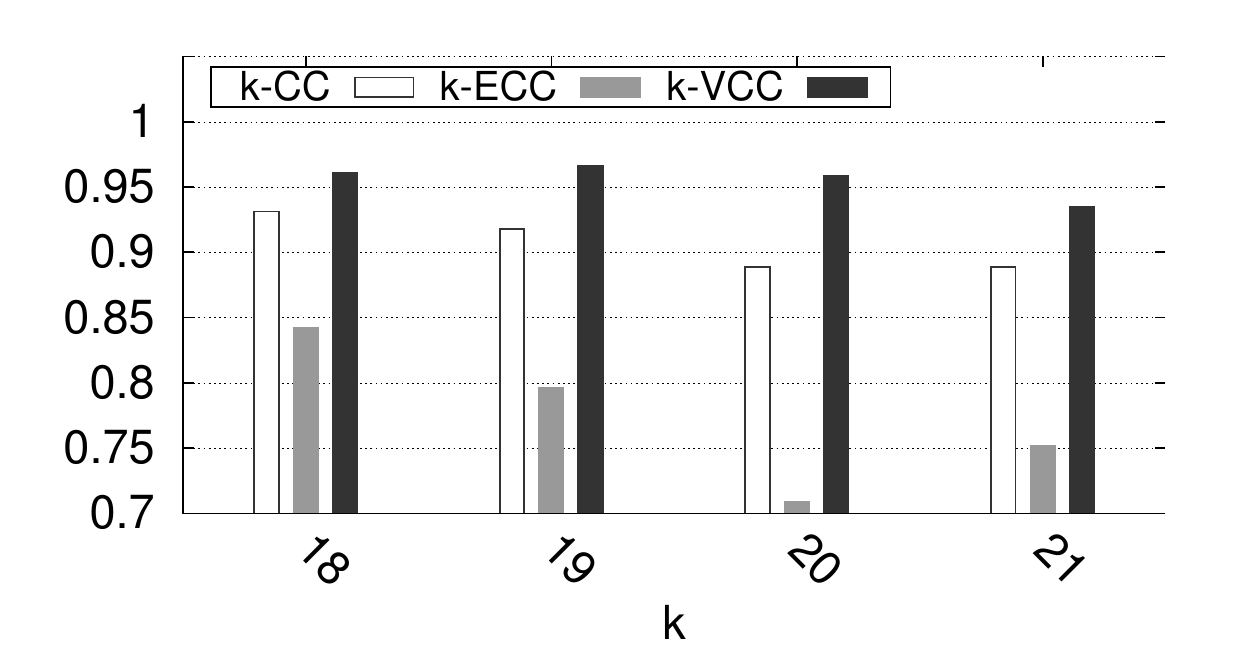}
}
\hspace{-1em}
\subfigure[\cnr]{
	\includegraphics[width=0.5\columnwidth]{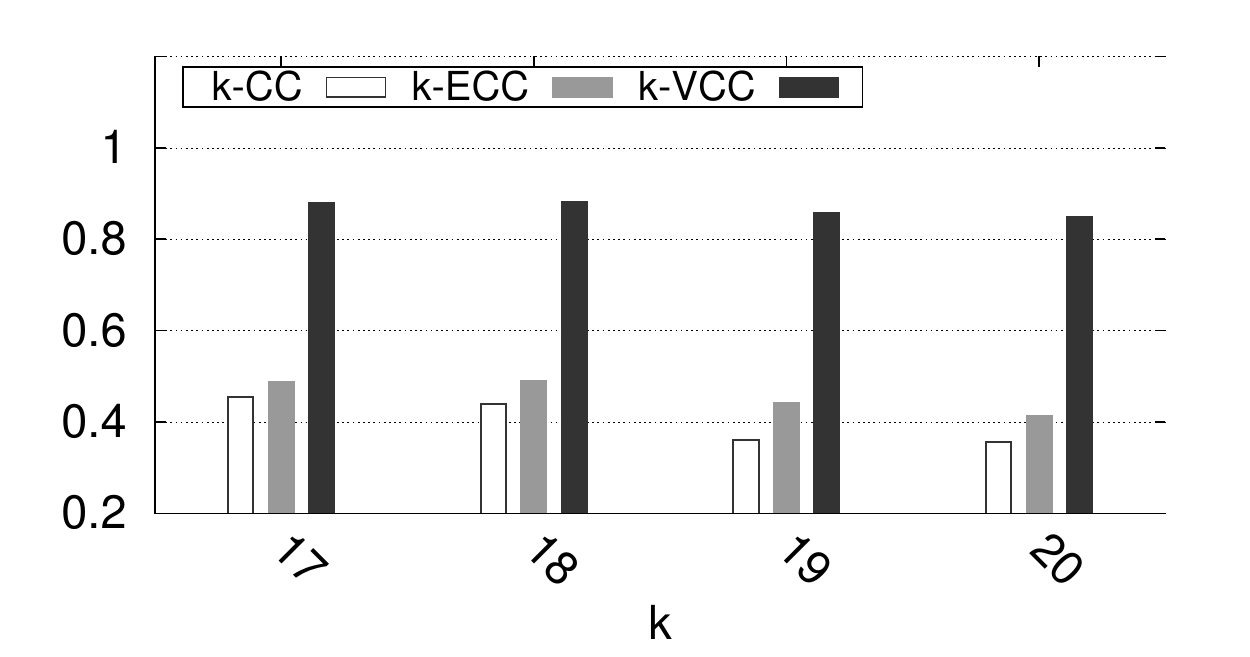}
}
\end{tabular}
\end{center}
\topcaption{Average Edge Density}
\label{fig:exp:aed}
\end{figure}

It is worth to mention that in \reffig{exp:ad}, when $k$ increases, the diameter of the obtained $k$-\ccs, $k$-\vccs, and $k$-\eccs can either increase or decrease. As an example, the average diameter of $k$-\vccs decreases slightly while increasing $k$ from $7$ to $8$ in the \ytb dataset. This is because when $k$ increases, the $k$-\vccs obtained are more cohesive, thus the diameter for the $k$-\vccs containing a certain vertex becomes smaller. Such reason also leads to the increase for edge density and clustering coefficient for some $k$ values in these datasets. As another example,  the average diameter of $k$-\vccs increases slightly while increasing $k$ from $20$ to $21$ in the \google dataset. The reason for this phenomenon is that there exist some small $20$-\vccs in which no vertex belongs to any $21$-\vcc. Here the small $k$-\vcc means there exist small number of vertices inside. These small $20$-\vccs have small diameter, which makes the average diameter small. Such reason also leads to the decrease for edge density and clustering coefficient for some $k$ values in these datasets. 

\begin{figure}[t!]
\begin{center}
\begin{tabular}[h]{c}
\hspace{-1em}
\subfigure[\ytb]{
	\includegraphics[width=0.5\columnwidth]{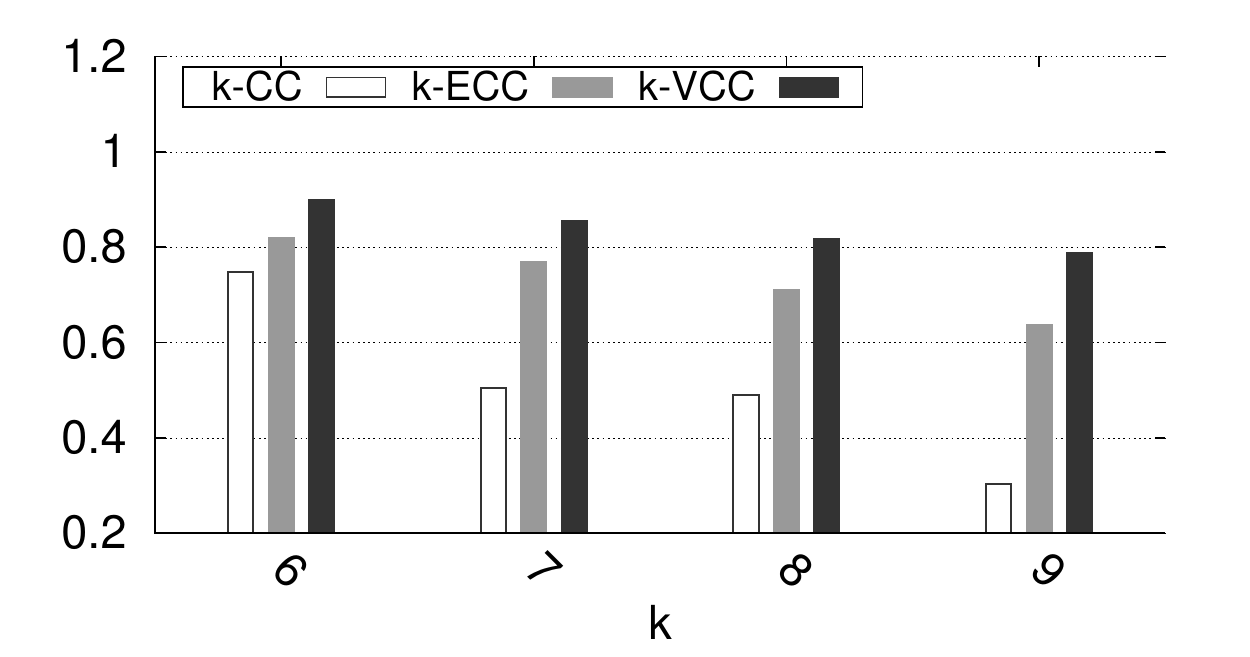}
}
\hspace{-1em}
\subfigure[\dblp]{
	\includegraphics[width=0.5\columnwidth]{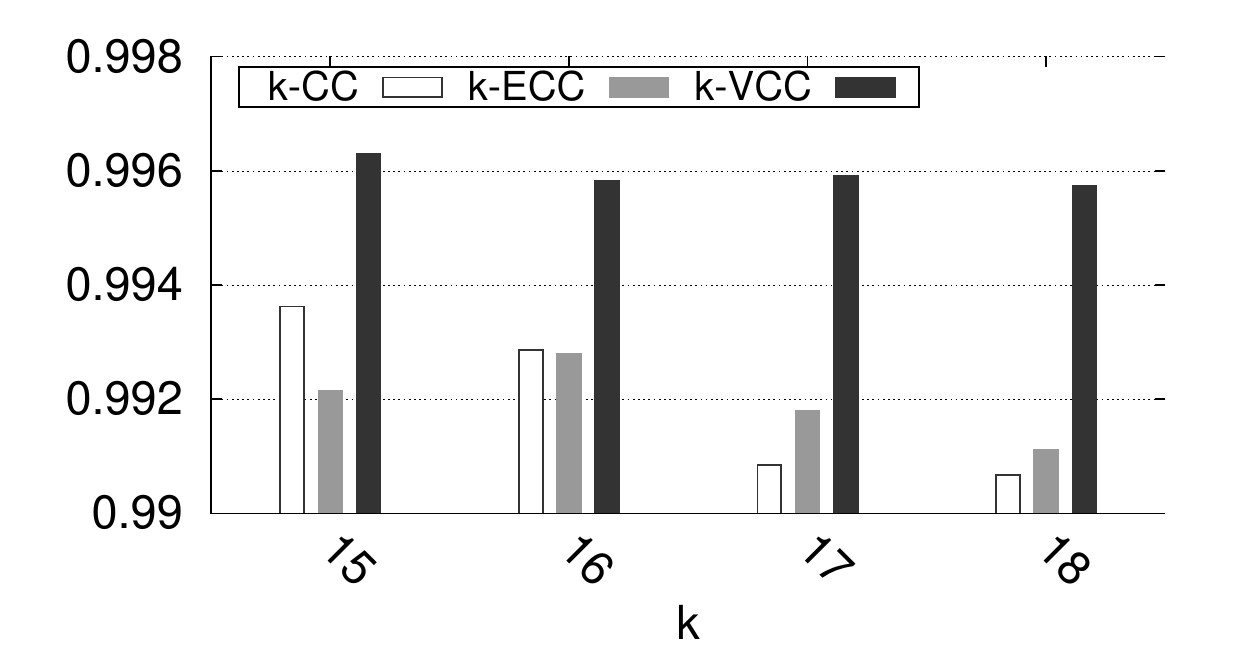}
}
\\
\hspace{-1em}
\subfigure[\google]{
	\includegraphics[width=0.5\columnwidth]{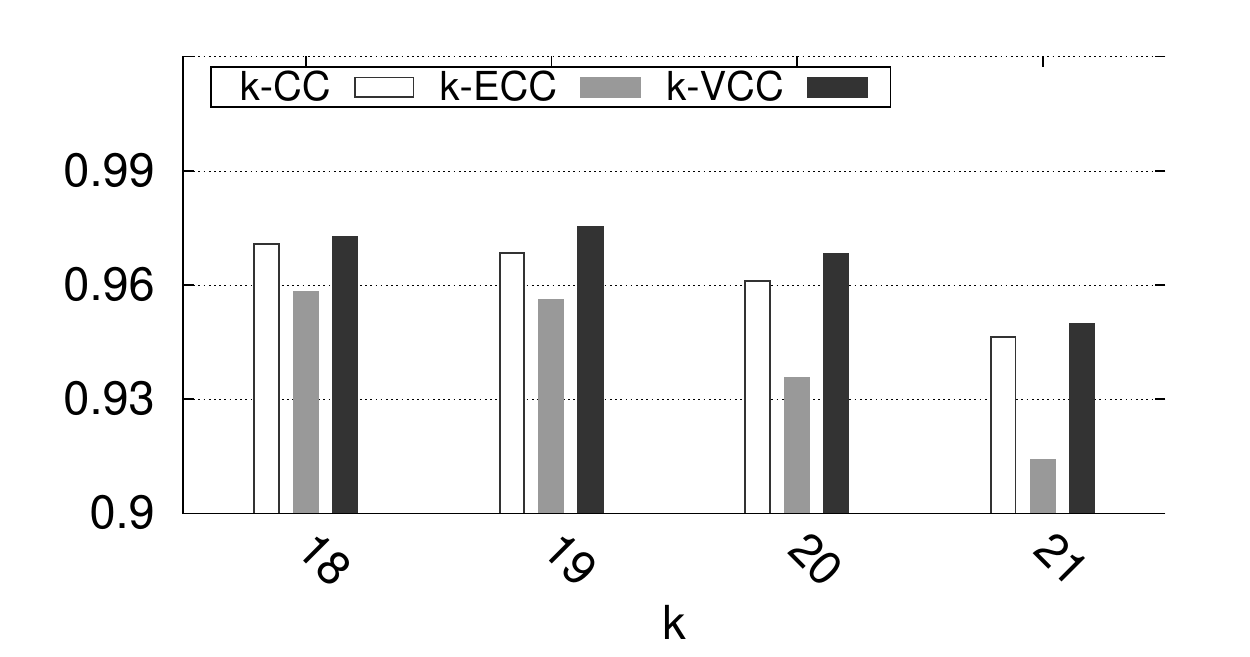}
}
\hspace{-1em}
\subfigure[\cnr]{
	\includegraphics[width=0.5\columnwidth]{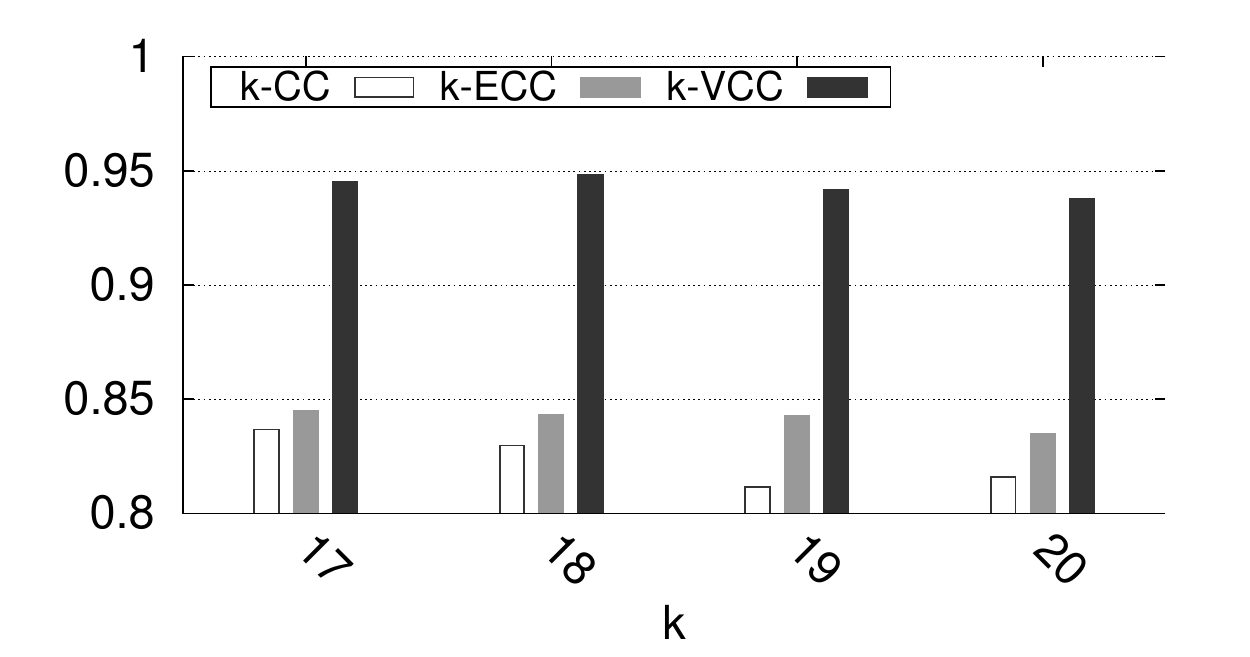}
}
\end{tabular}
\end{center}
\topcaption{Average Clustering Coefficient}
\label{fig:exp:cco}
\end{figure}

A case study is shown in \refsubsec{exp:cs} to further demonstrate the effectiveness of $k$-\vccs. 

\subsection{Efficiency Evaluation}
\label{subsec:exp:qp}

\begin{figure}[t!]
\begin{center}
\begin{tabular}[t]{c}
\hspace{-1em}
\subfigure[\sfd]{
	\includegraphics[width=0.5\columnwidth]{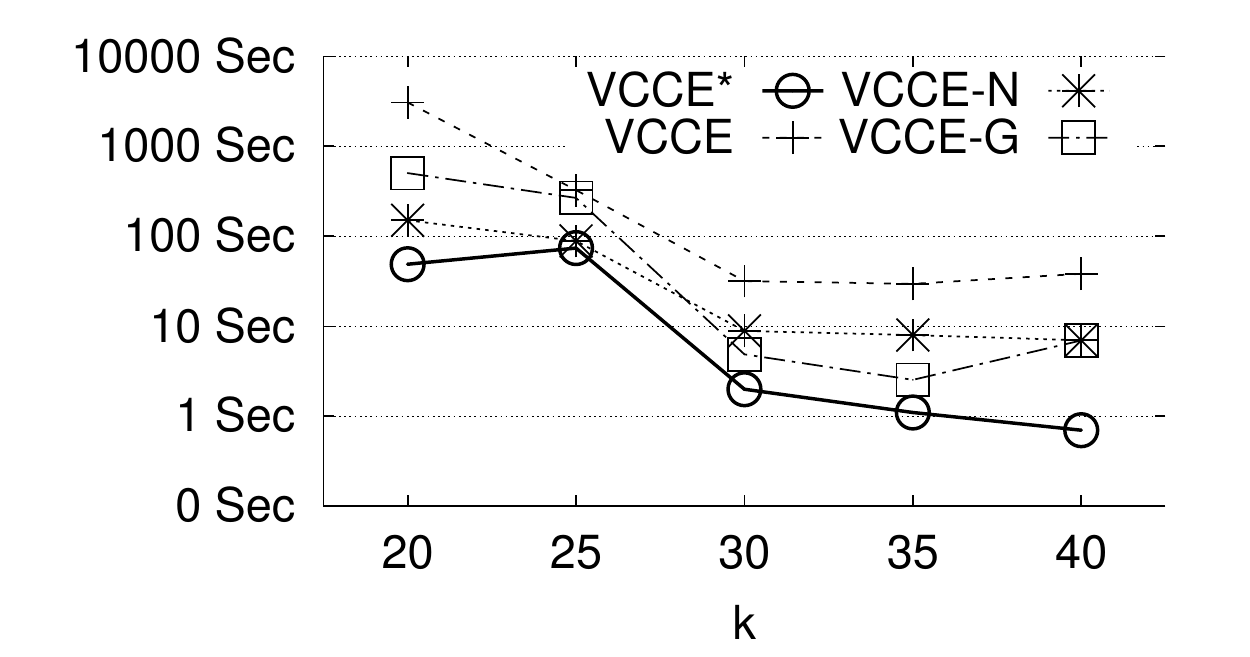}
}
\hspace{-1em}
\subfigure[\dblp]{
	\includegraphics[width=0.5\columnwidth]{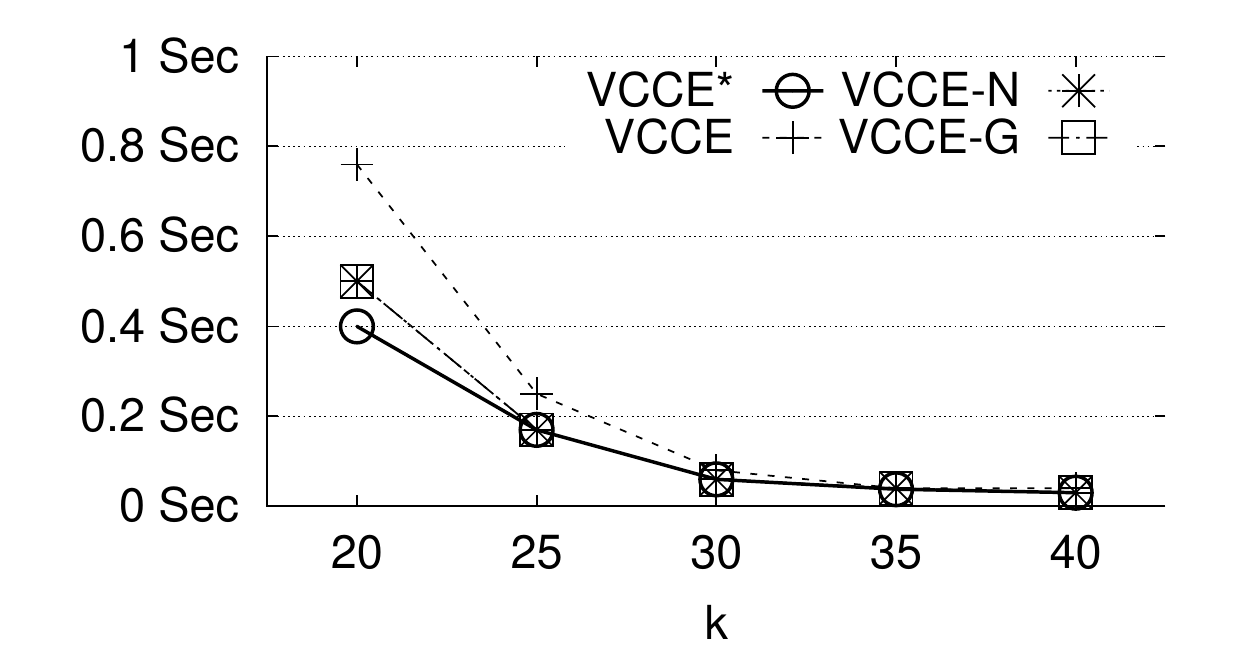}
}
\\
\hspace{-1em}
\subfigure[\nd]{
	\includegraphics[width=0.5\columnwidth]{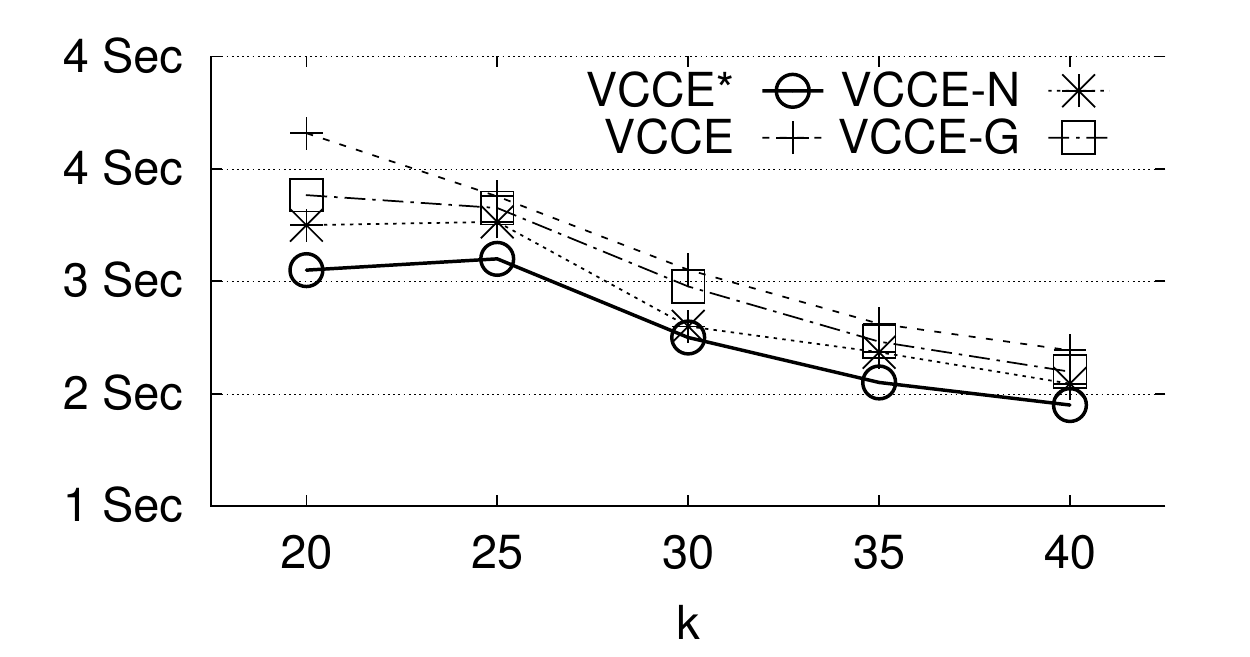}
}
\hspace{-1em}
\subfigure[\google]{
	\includegraphics[width=0.5\columnwidth]{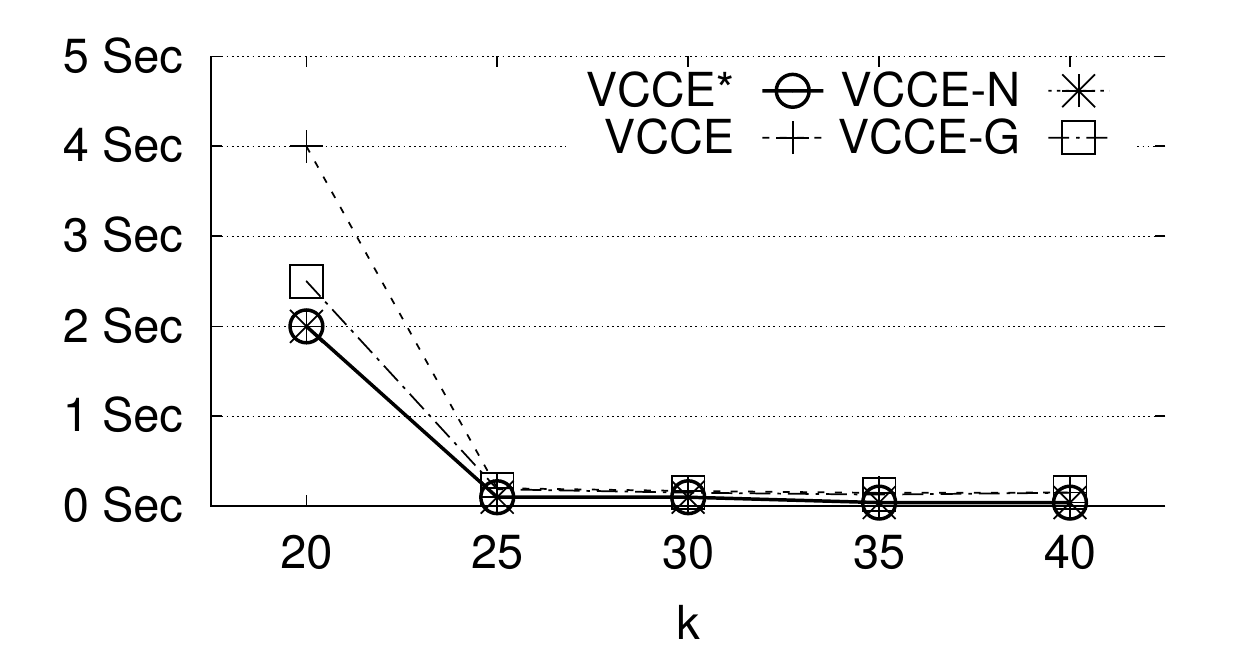}
}
\\
\hspace{-1em}
\subfigure[\cit]{
	\includegraphics[width=0.5\columnwidth]{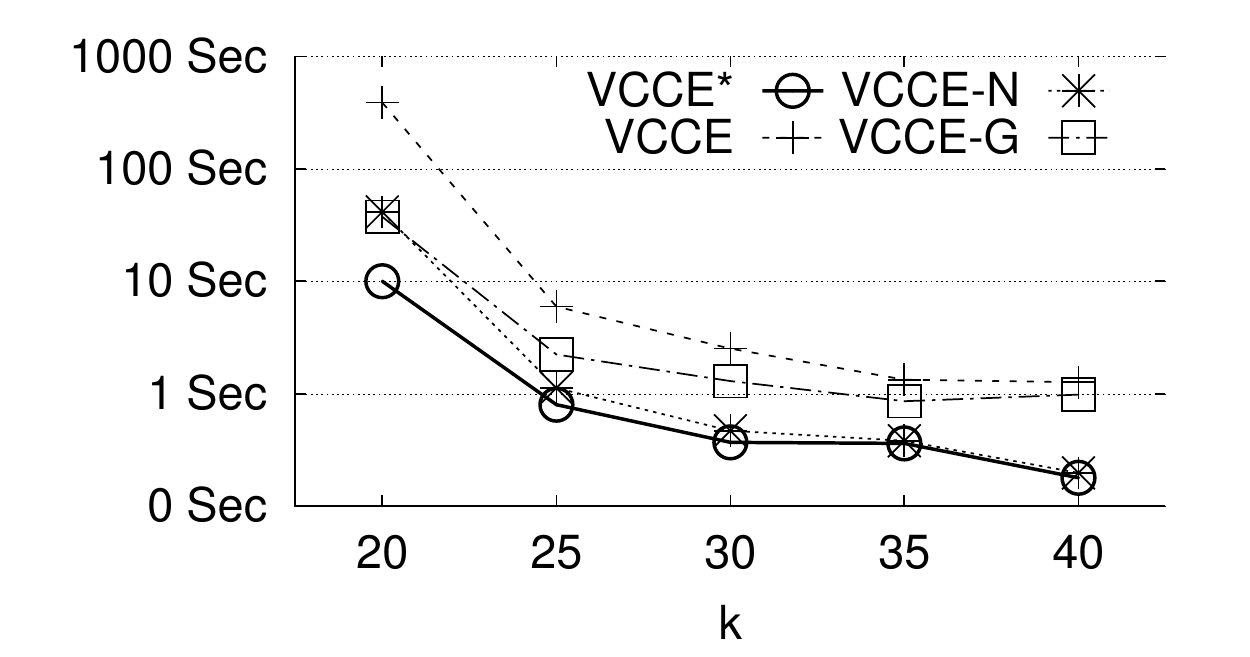}
}
\hspace{-1em}
\subfigure[\cnr]{
	\includegraphics[width=0.5\columnwidth]{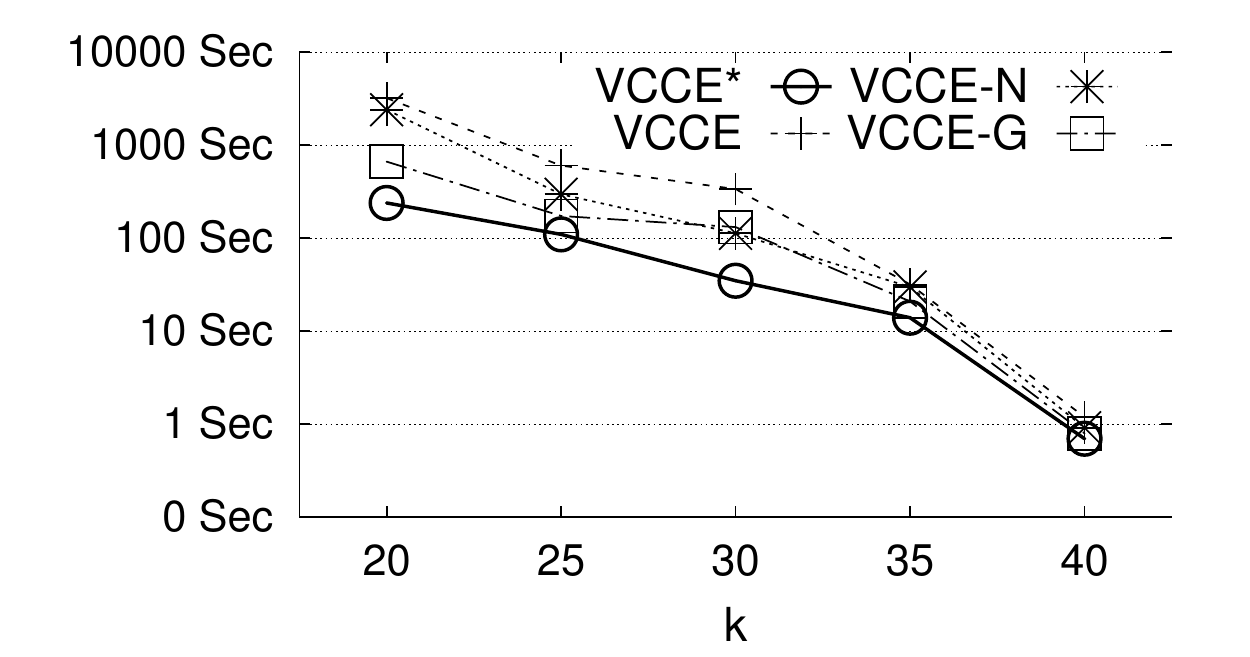}
}
\end{tabular}
\end{center}
\topcaption{Processing time}
\label{fig:exp:performance}
\end{figure}

To test the efficiency of our proposed techniques, we compare the following four algorithms to compute the $k$-\vccs. For each dataset, we show statistics of algorithms under different parameters $k$ varying from $20$ to $40$.

\begin{itemize}
\item \bsa: Our basic algorithm introduced in \refsec{base}.
\item \nsa: The basic algorithm with the neighbor sweep strategy introduced in \refsec{opt:ns}.
\item \gsa: The basic algorithm with the group sweep strategy introduced in \refsec{opt:gs}.
\item \mra: The algorithm with both neighbor sweep and group sweep strategies.
\end{itemize} 

\stitle{Testing the Time Cost.} As we can see from \reffig{exp:performance}, the time cost of each algorithm generally presents a decreasing trend while parameter $k$ increases. That is because a higher value of parameter $k$ leads to a smaller number of $k$-\vccs. Intuitively, the algorithm will test less local connectivity during the processing when $k$ increases. A special case here is that algorithm \mra spends a little more time under $k=25$ than under $k=20$ in the \sfd dataset. This phenomenon happens due to the structure of the \sfd graph in which $k=25$ leads to more partitions than $k=20$. We also find that both algorithms \nsa and \gsa are more efficient than the basic algorithm in all testing cases. Considering the specific structures of different datasets, we find that the group sweep strategy is more effective on graph \cnr, and the neighbor sweep strategy is more effective on other datasets. Our \mra algorithm outperforms all other algorithms in all test cases.

\begin{table}[t]

\topcaption{\sc Proportion for Different Rules}
\label{tab:ruleprop}
\begin{center}
{\small
\begin{tabular}{l|c|c|c|c|c|c}
\hline
{\bf Rules} & \sfd & \dblp & \nd & \google & \cit & \cnr \\\hline\hline
\textit{NS\_1} & 14\% & 67\% & 1\% & 29\% & 12\% & 11\% \\
\textit{NS\_2} & 40\% & 21\% & 42\% & 36\% & 68\% & 32\% \\
\textit{GS} & 13\% & 4\% & 1\% & 9\% & 12\% & 48\% \\
\textit{Non-Pru} & 33\% & 8\% & 56\% & 26\% & 8\% & 9\% \\\hline
\end{tabular}
}
\end{center}
\end{table}

\stitle{Testing the Effectiveness of Sweep Rules.} To further investigate the effectiveness of our sweep rules, we also track each processed vertex during the performance of \mra and record the number of vertices pruned by each strategy. Specifically, when performing sweep procedure, we separately mark the vertices pruned by \textit{neighbor sweep rule 1} (strong-side vertex), \textit{neighbor sweep rule 2} (neighbor deposit) and group sweep. Here, we divide neighbor sweep into two detailed sub-rules since the both of them perform well and the effectiveness of these two strategies is not very consistent in different datasets. For each vertex $v$ in line~11 of \fcstar, we increase the count for corresponding strategy if $v$ is pruned (line 13). We also record the number of vertices which are non-pruned and really tested (line 12). For each dataset, we record these data under different $k$ from $20$ to $40$ and obtain the average value. The result is shown in \reftab{ruleprop}. \textit{NS\_1} and \textit{NS\_2} represent \textit{neighbor sweep rule 1} and \textit{neighbor sweep rule 2} respectively. \textit{GS} is group sweep and \textit{Non-Pru} means the proportion of non-pruned vertices. Note that there is a large number of vertices which are pruned in advance by the $k$-core technique.

The result shows our pruning strategies are effective. Over $90\%$ vertices are pruned in \dblp, \cit and \cnr. The proportion of totally pruned vertices is smallest in \nd, which is about $45\%$. Among these pruning strategies, the effectiveness of \textit{neighbor sweep rule 1} and group sweep depends on the specific structure of datasets. \textit{neighbor sweep rule 1} performs much better than group sweep in \dblp. The pruned vertices due to such strategy accounts for $67\%$ of total (including really tested vertices). Group sweep is more effective than \textit{neighbor sweep rule 1} in \cnr. The percentage for group sweep is about $48\%$ while it is only $11\%$ for \textit{neighbor sweep rule 1}. These two strategies are of about the same effectiveness in other datasets. As comparison, the \textit{neighbor sweep rule 2} closely relies on the existing processed vertices. It becomes more and more effective with vertices tested or pruned constantly. Our result shows it is very powerful and stable. The percentage for such strategy reaches to $68\%$ in \cit and is over $20\%$ in all other datasets.


\begin{figure}[t!]
\begin{center}
\begin{tabular}[t]{c}
\hspace{-1em}
\subfigure[\sfd]{
	\includegraphics[width=0.5\columnwidth, height=2cm]{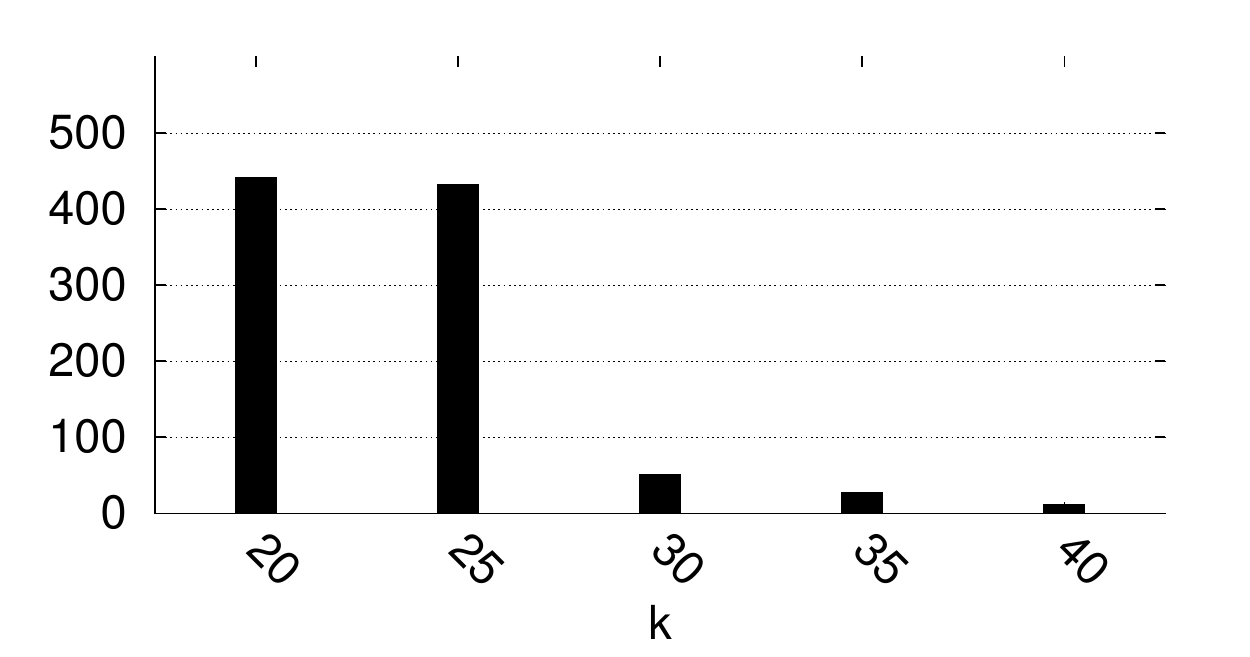}
}
\hspace{-1em}
\subfigure[\dblp]{
	\includegraphics[width=0.5\columnwidth, height=2cm]{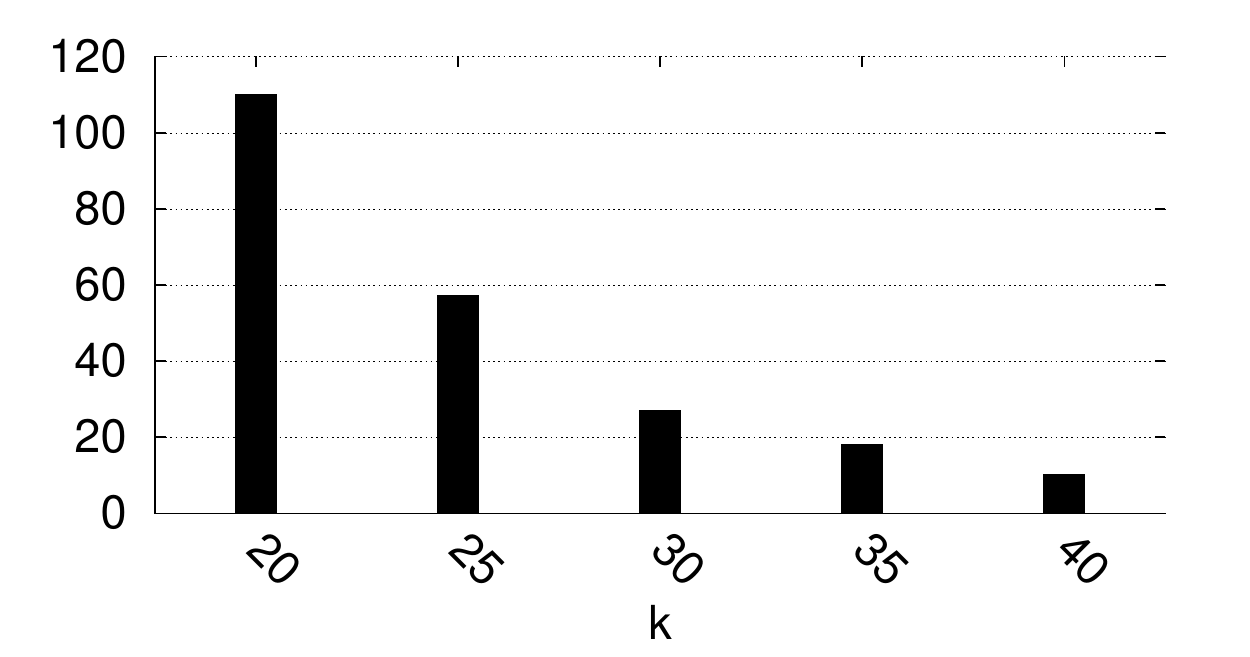}
}
\\
\hspace{-1em}
\subfigure[\nd]{
	\includegraphics[width=0.5\columnwidth, height=2cm]{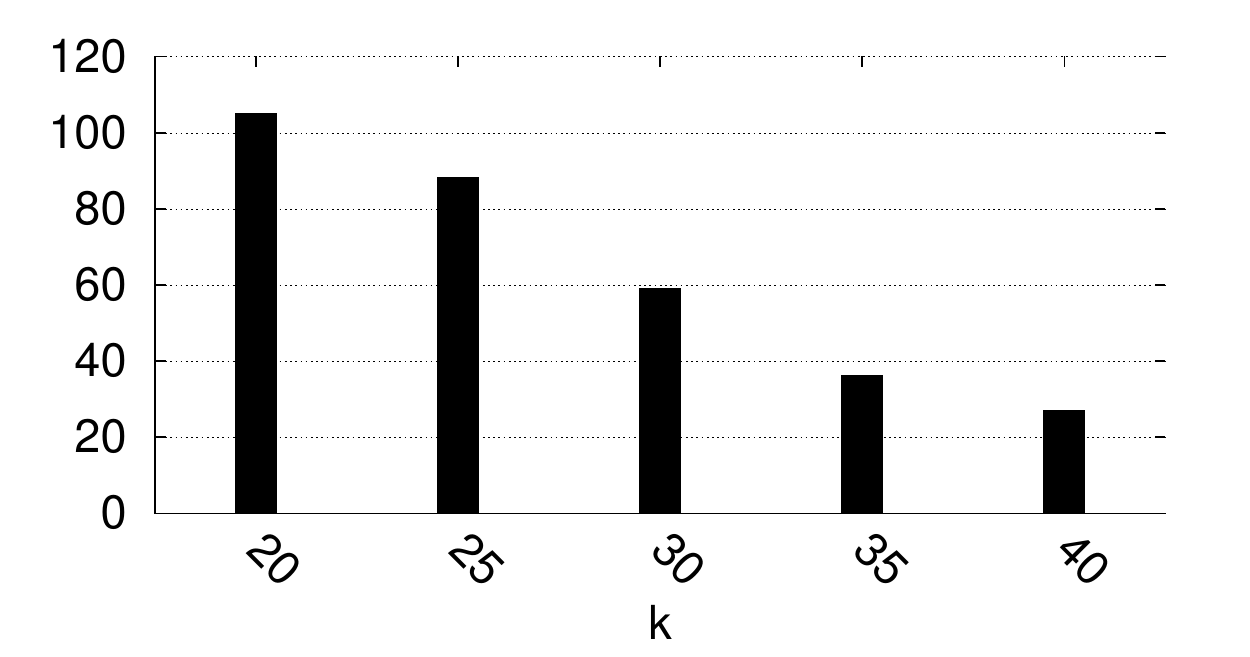}
}
\hspace{-1em}
\subfigure[\google]{
	\includegraphics[width=0.5\columnwidth, height=2cm]{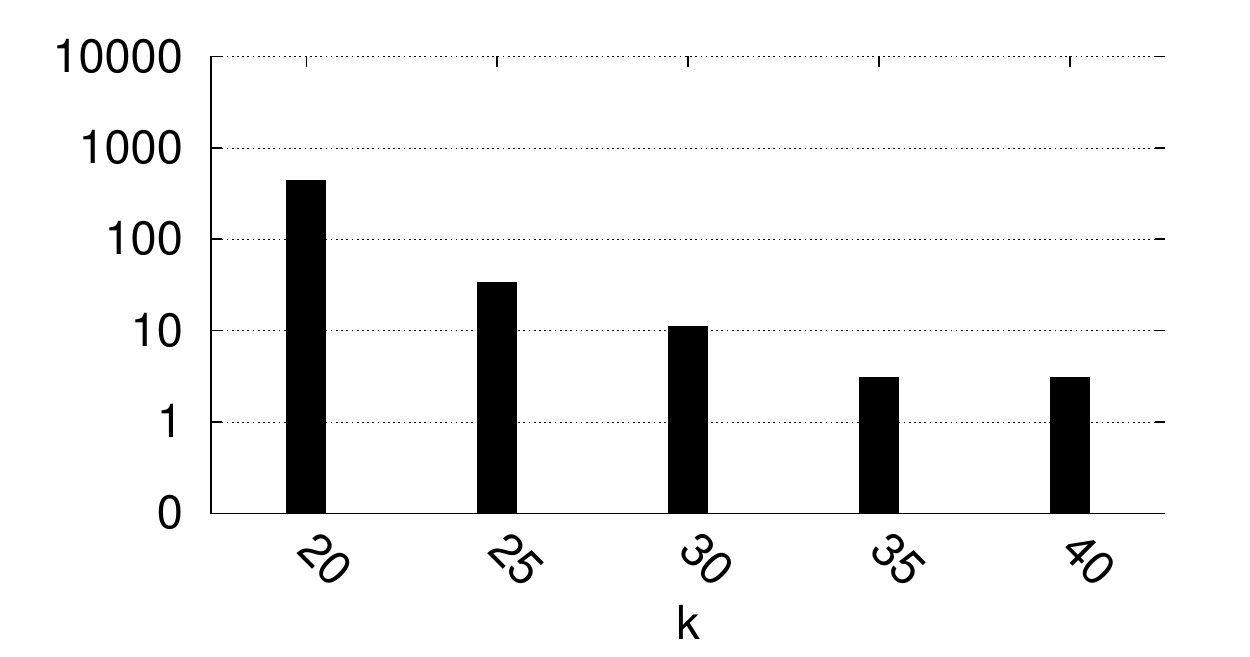}
}
\\
\hspace{-1em}
\subfigure[\cit]{
	\includegraphics[width=0.5\columnwidth, height=2cm]{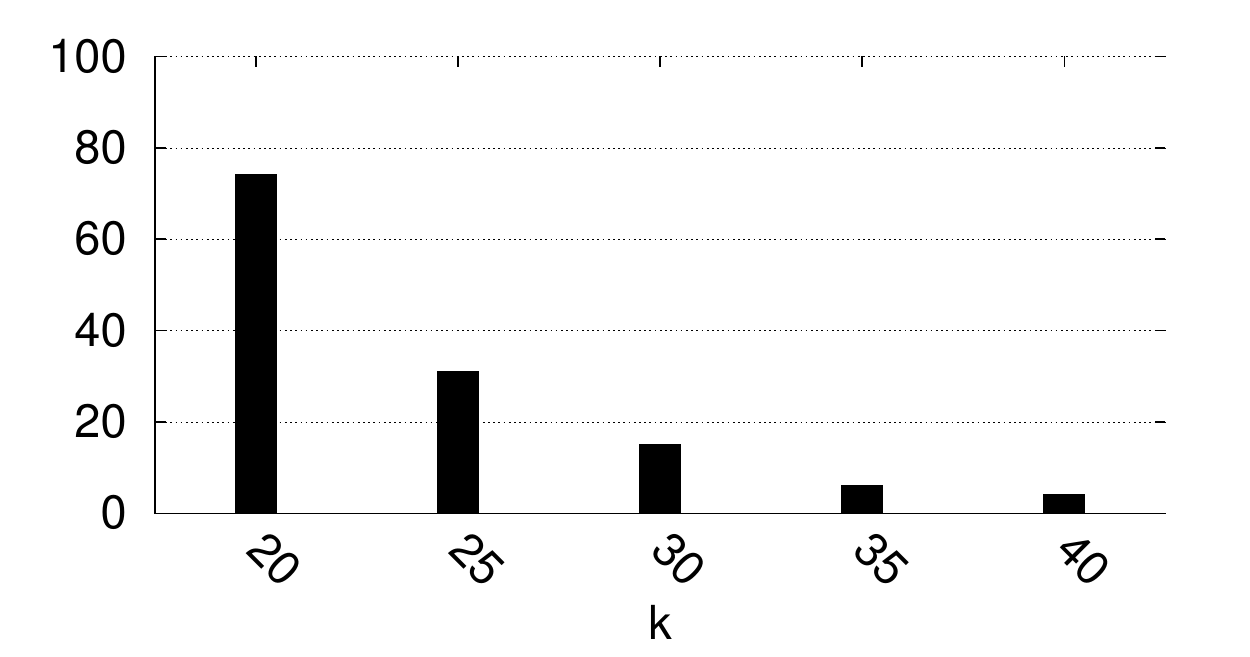}
}
\hspace{-1em}
\subfigure[\cnr]{
	\includegraphics[width=0.5\columnwidth, height=2cm]{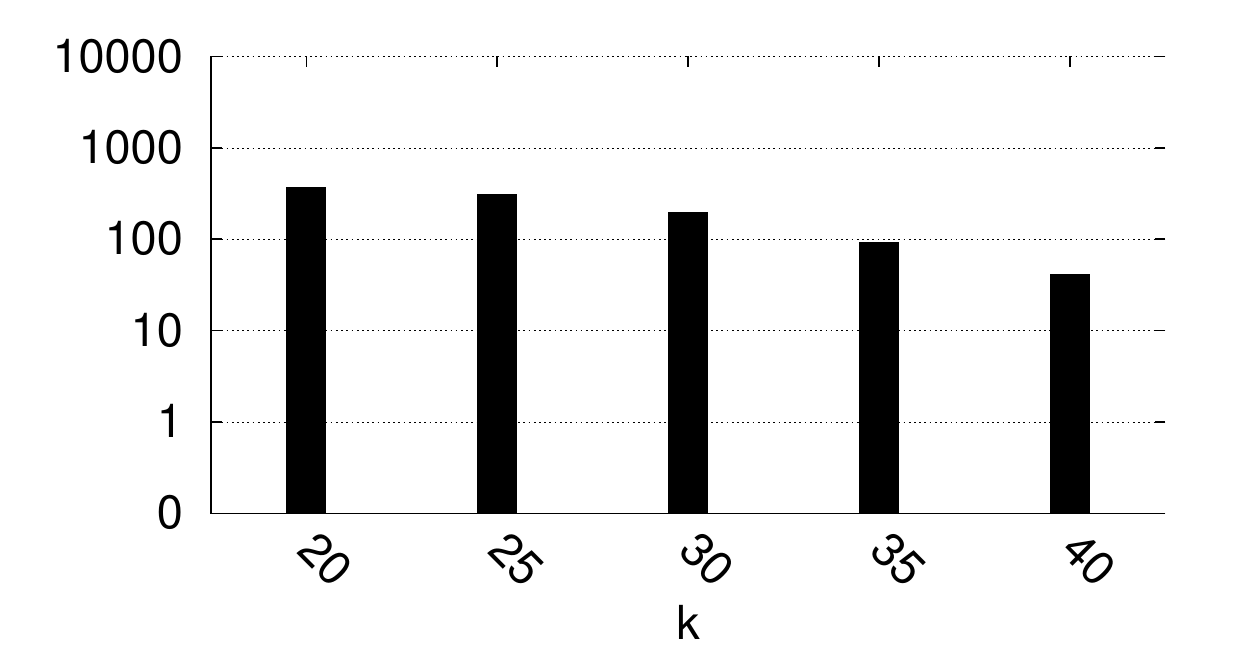}
}
\end{tabular}
\end{center}
\topcaption{Number of $k$-\vccs}
\label{fig:exp:cnts}
\end{figure}

\stitle{Testing the Number of $k$-\vccs.} The numbers of $k$-\vccs under different $k$ values for each dataset are given in \reffig{exp:cnts}. The numbers of $k$-\vccs on all tested datasets have a decreasing trend when varying $k$ from $20$ to $40$ in \reffig{exp:cnts}. The reason is that when increasing $k$, some $k$-\vccs cannot satisfy the requirement and thus will not appear in the result list. The trend of the number of $k$-\vccs explains why the processing time of our algorithms decreases when $k$ increases in \reffig{exp:performance}. Note that the number of $k$-\vccs may vary a lot in different datasets for the same $k$ value.  In the same dataset, when $k$ increases, the number of $k$-\vccs may drop sharply. For example, for when $k$ increases from $20$ to $25$, the number of $k$-\vccs in \google decreases by $10$ times. The number of $k$-\vccs depends on the graph structure of each specific graph. 

\begin{figure}[t!]
\begin{center}
\begin{tabular}[t]{c}
\hspace{-1em}
\subfigure[\sfd]{
	\includegraphics[width=0.5\columnwidth, height=2cm]{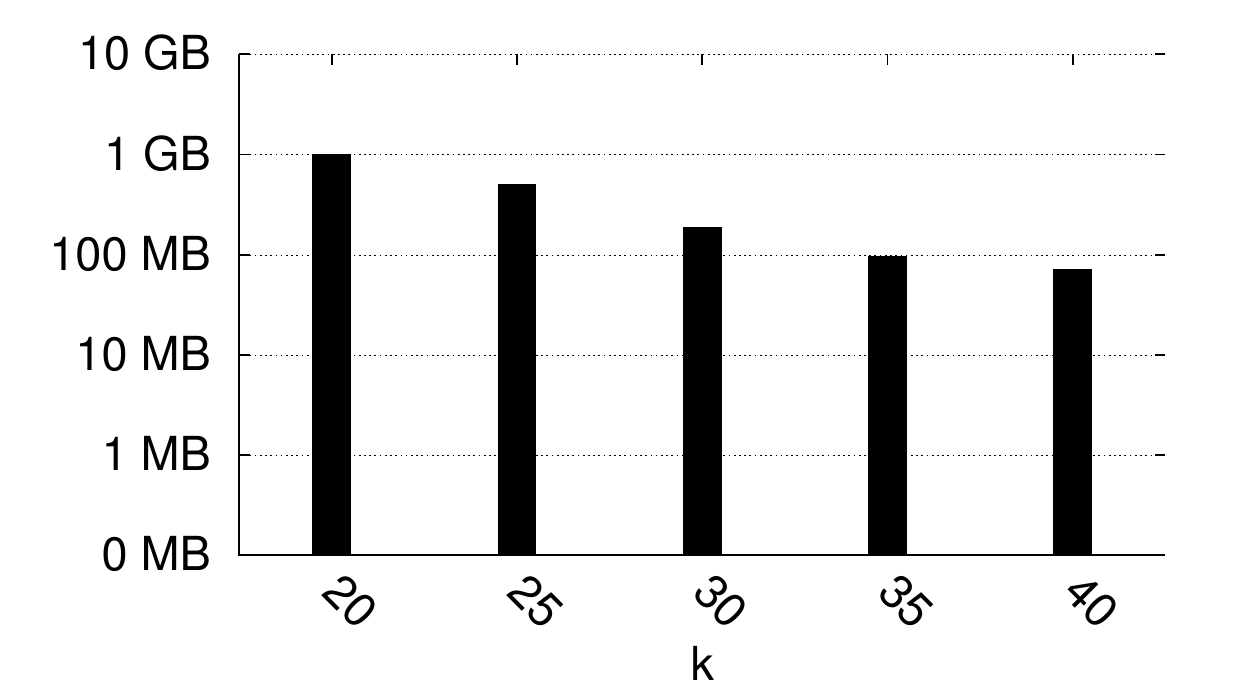}
}
\hspace{-1em}
\subfigure[\dblp]{
	\includegraphics[width=0.5\columnwidth, height=2cm]{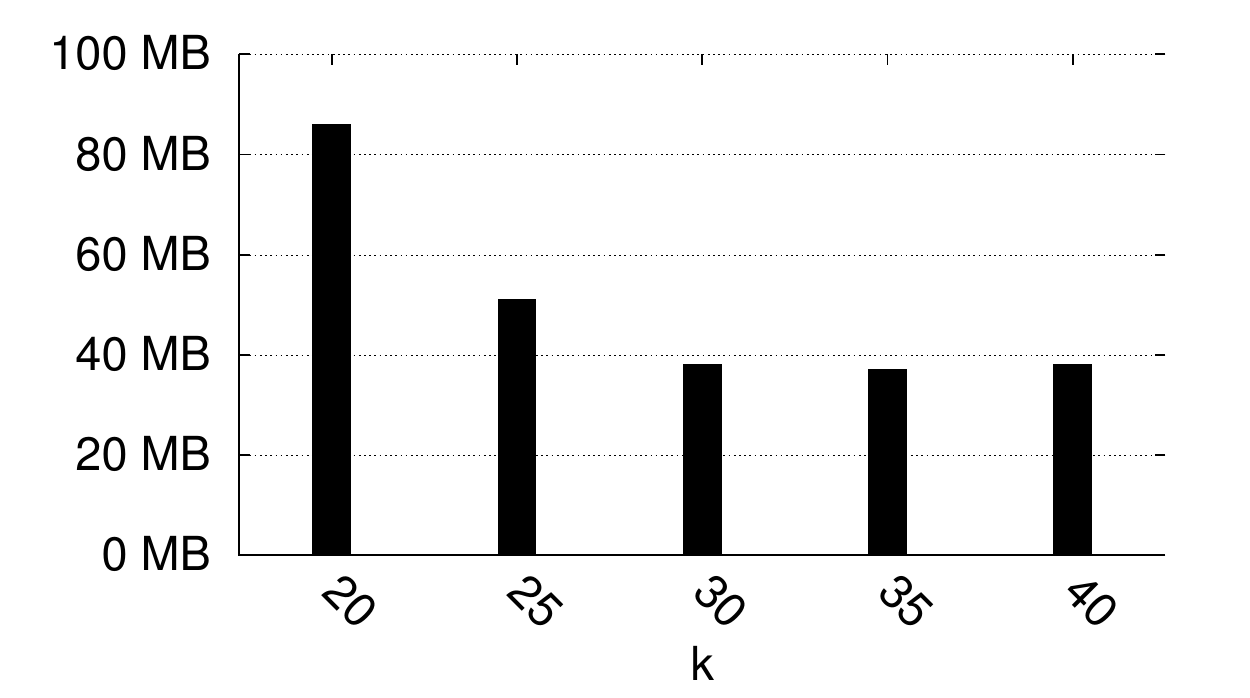}
}
\\
\hspace{-1em}
\subfigure[\nd]{
	\includegraphics[width=0.5\columnwidth, height=2cm]{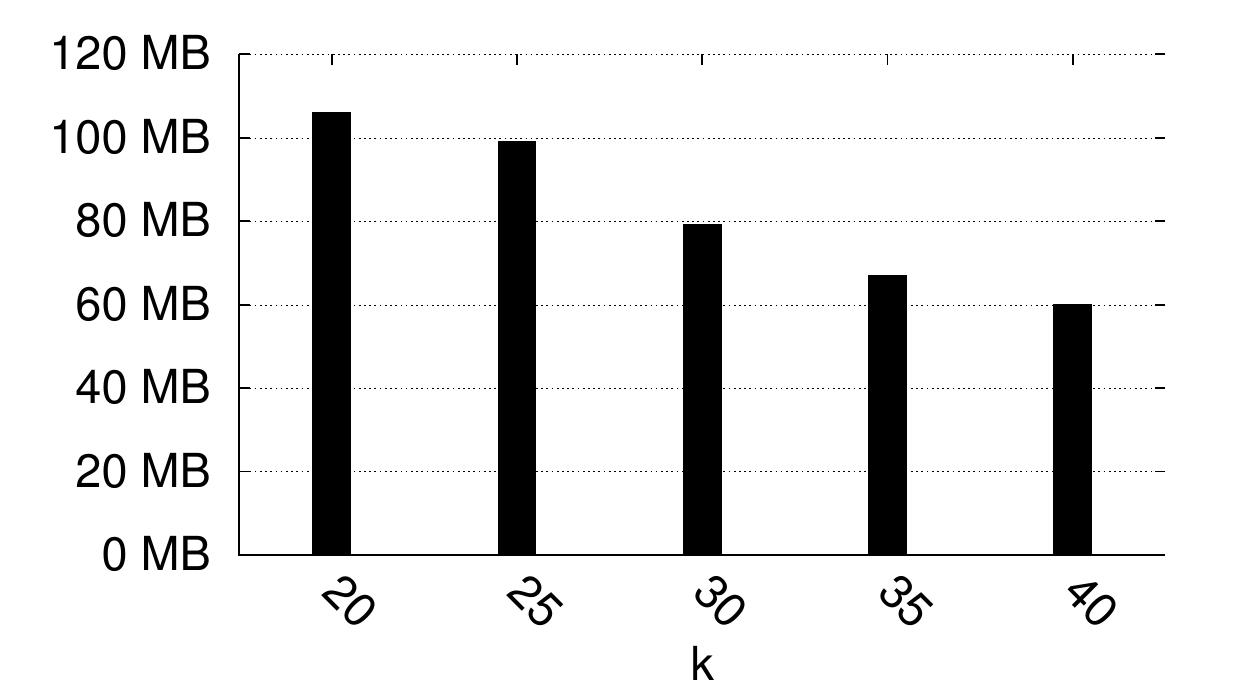}
}
\hspace{-1em}
\subfigure[\google]{
	\includegraphics[width=0.5\columnwidth, height=2cm]{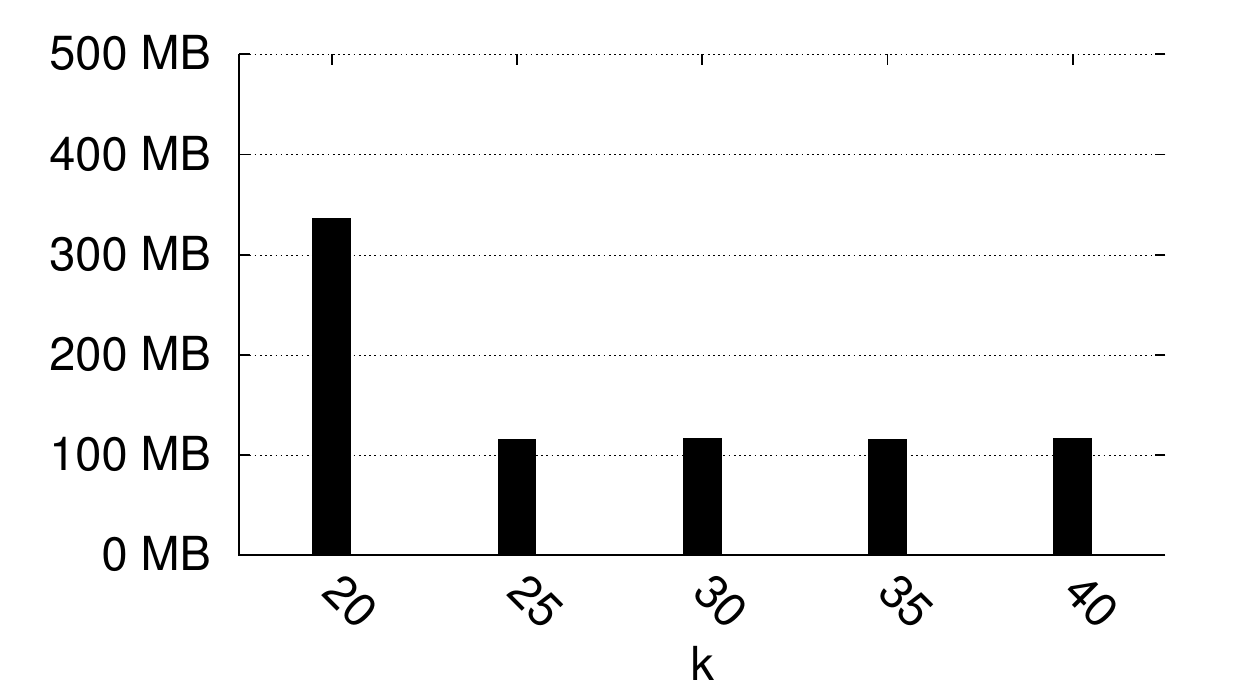}
}
\\
\hspace{-1em}
\subfigure[\cit]{
	\includegraphics[width=0.5\columnwidth, height=2cm]{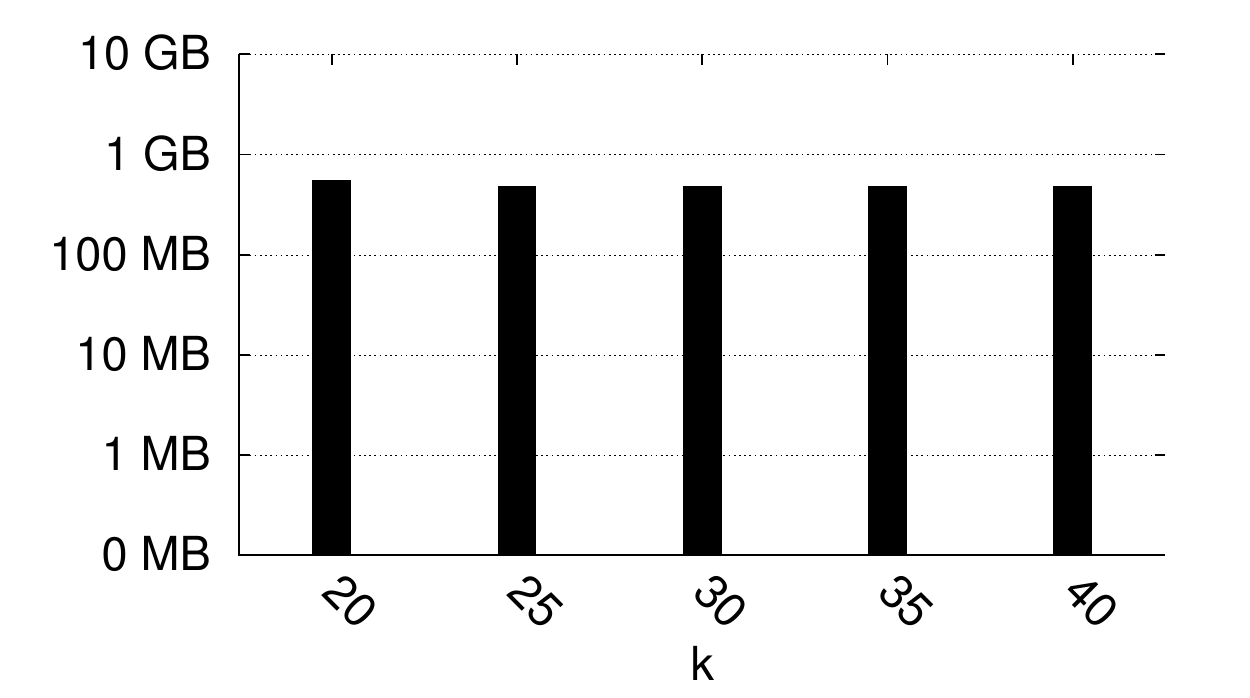}
}
\hspace{-1em}
\subfigure[\cnr]{
	\includegraphics[width=0.5\columnwidth, height=2cm]{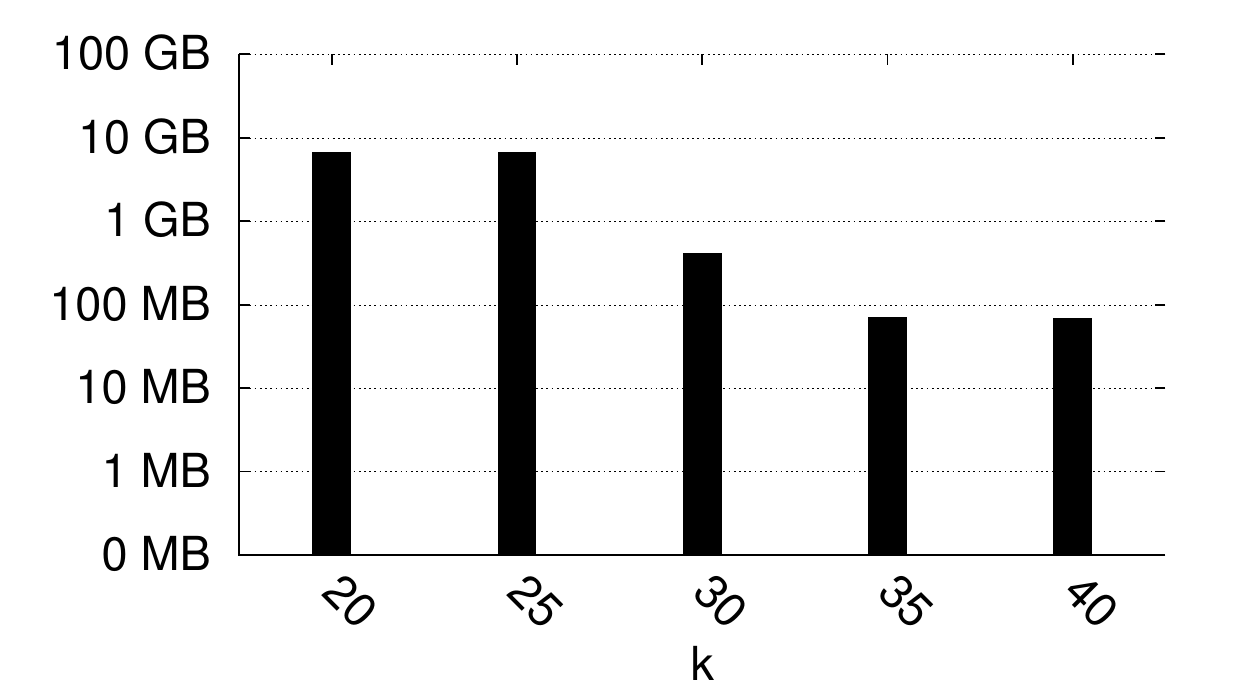}
}
\end{tabular}
\end{center}
\topcaption{Memory Usage of Algorithm \mra}
\label{fig:exp:ms}
\end{figure}

\stitle{Testing the Memory Usage.} \reffig{exp:ms} presents the memory usage of algorithm \mra on different datasets while verying parameter $k$. Note that the memory usage of all the four algorithms \bsa, \nsa, \gsa, and \mra are very close since they follow the same framework to cut the graph recursively, and the memory usage mainly depends on the size of graph and the number of partitioned graphs which are the same for all the four algorithms. Therefore, we only show the memory usage of \mra.  

As we can see from the figure, the memory usage on most of the datasets has a decreasing trend when $k$ increases. The reasons are twofold. First, recall that all vertices with degree less than $k$ are firstly removed for each subgraph during the algorithm. A higher $k$ must lead to more removed vertices, and therefore, makes the graph smaller. Second, when $k$ increases, the number of $k$-\vccs decreases and the number of partitioned graphs also decreases, which lead to a smaller memory usage. For some cases, the memory usage increases when $k$ increases, this is because when $k$ increases, the sparse certificate of the graph becomes denser, which requires more memory. Generally, the memory usage keeps in a reasonable range in all testing cases. 

\subsection{Scalability Evaluation}
\label{subsec:exp:scal}

\begin{figure}[t!]
\begin{center}
\begin{tabular}[h]{c}
\hspace{-1em}
\subfigure[Vary Size (\google)]{
	\includegraphics[width=0.5\columnwidth]{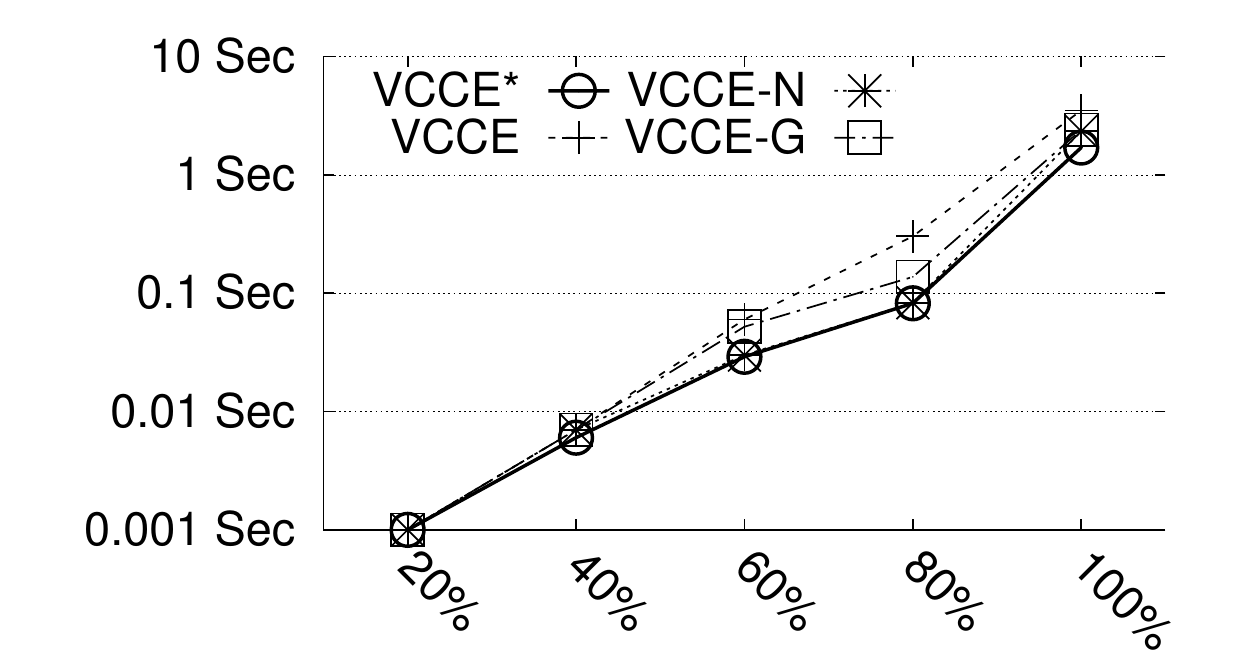}
}
\hspace{-1em}
\subfigure[Vary Density (\google)]{
	\includegraphics[width=0.5\columnwidth]{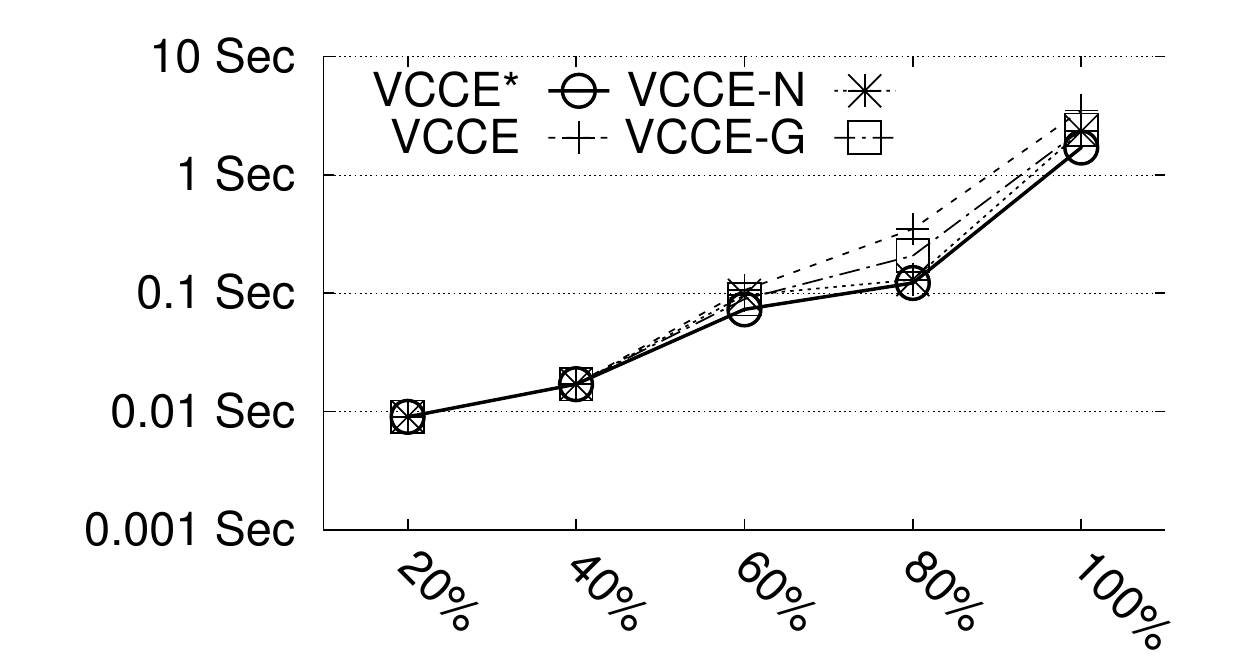}
}
\\
\hspace{-1em}
\subfigure[Vary Size (\cit)]{
	\includegraphics[width=0.5\columnwidth]{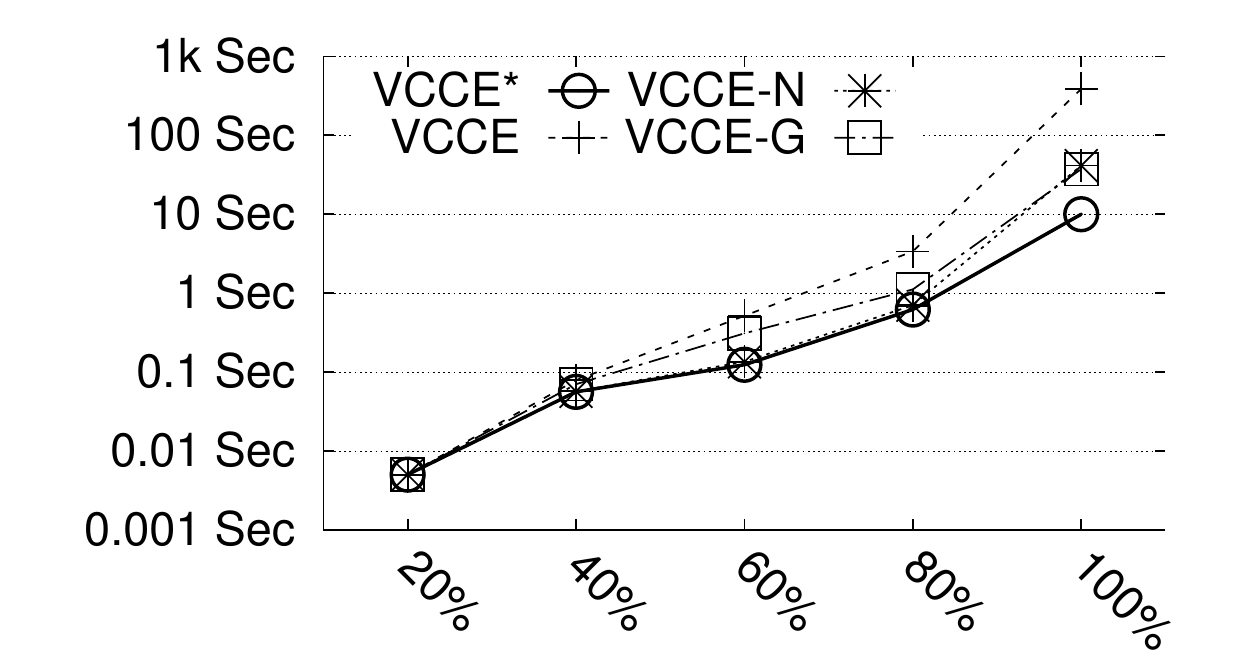}
}
\hspace{-1em}
\subfigure[Vary Density (\cit)]{
	\includegraphics[width=0.5\columnwidth]{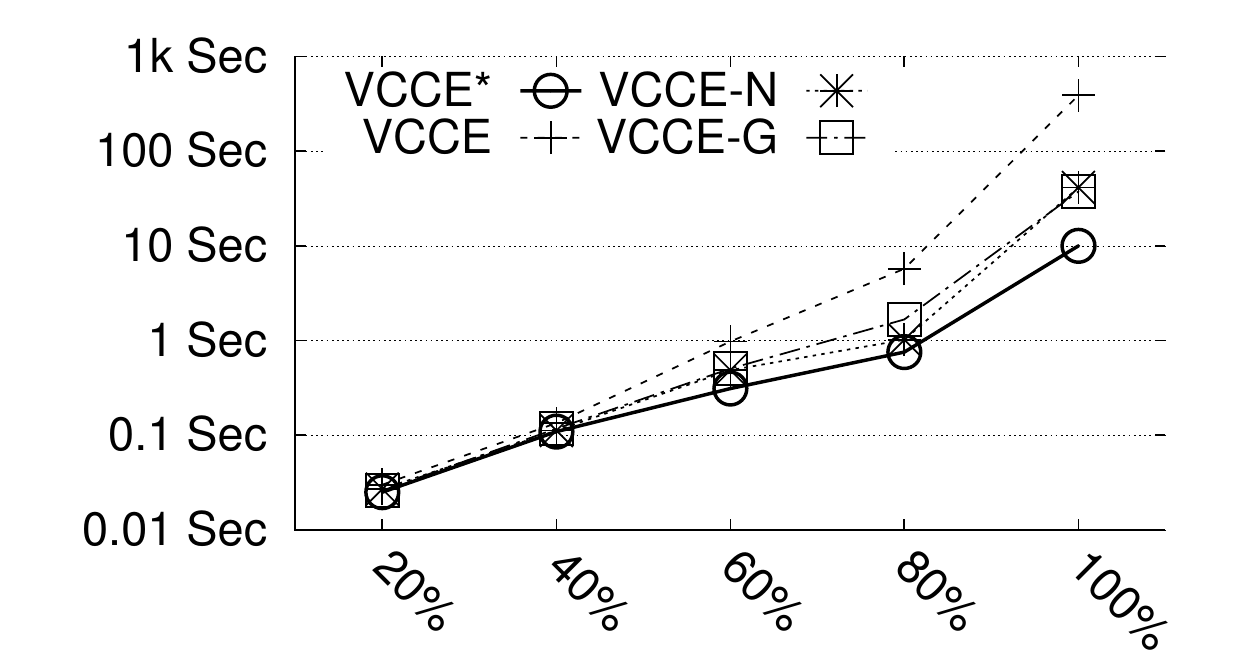}
}

\end{tabular}
\end{center}
\topcaption{Scalability evaluation}
\label{fig:exp:scal}
\end{figure}

In this section, we test the scalability of our proposed algorithms. We choose two real graph datasets \google and \cit as representatives. For each dataset, we vary the graph size and graph density by randomly sampling vertices and edges respectively from $20\%$ to $100\%$. When sampling vertices, we get the induced subgraph of the sampled vertices, and when sampling edges, we get the incident vertices of the edges as the vertex set. Here, we only report the processing time. The memory usage is linear to the number of vertices. We compare four algorithms in the experiments. The experimental results are shown in \reffig{exp:scal}.

\reffig{exp:scal} (a) and (c) report the processing time of our proposed algorithms when varying $|V|$ in \google and \cit respectively. When $|V|$ increases, the processing time for all algorithms increases. \mra performs best in all cases and \bsa is the worst one. The curves in \reffig{exp:scal} (b) and (d) report the processing time of our algorithms in \google and \cit respectively when varying $|E|$. Similarly, \mra is the fastest algorithm in all tested cases. In addition, the gap between \mra and \bsa increases when $|E|$ increases. For example, in \cit, the processing time of \mra is $20$ times faster than that of \bsa when $|E|$ reaches $100\%$. The result shows that our pruning strategy is effective and our optimized algorithm is more efficient and scalable than the basic algorithm.

\subsection{Case Study}
\label{subsec:exp:cs}

\begin{figure}[t!]
\begin{center}
\vspace*{-0.25cm}
\begin{tabular}[t]{c}
\hspace{-2em}
\subfigure[\small{4-\vccs}]{
	\includegraphics[width=0.99\columnwidth]{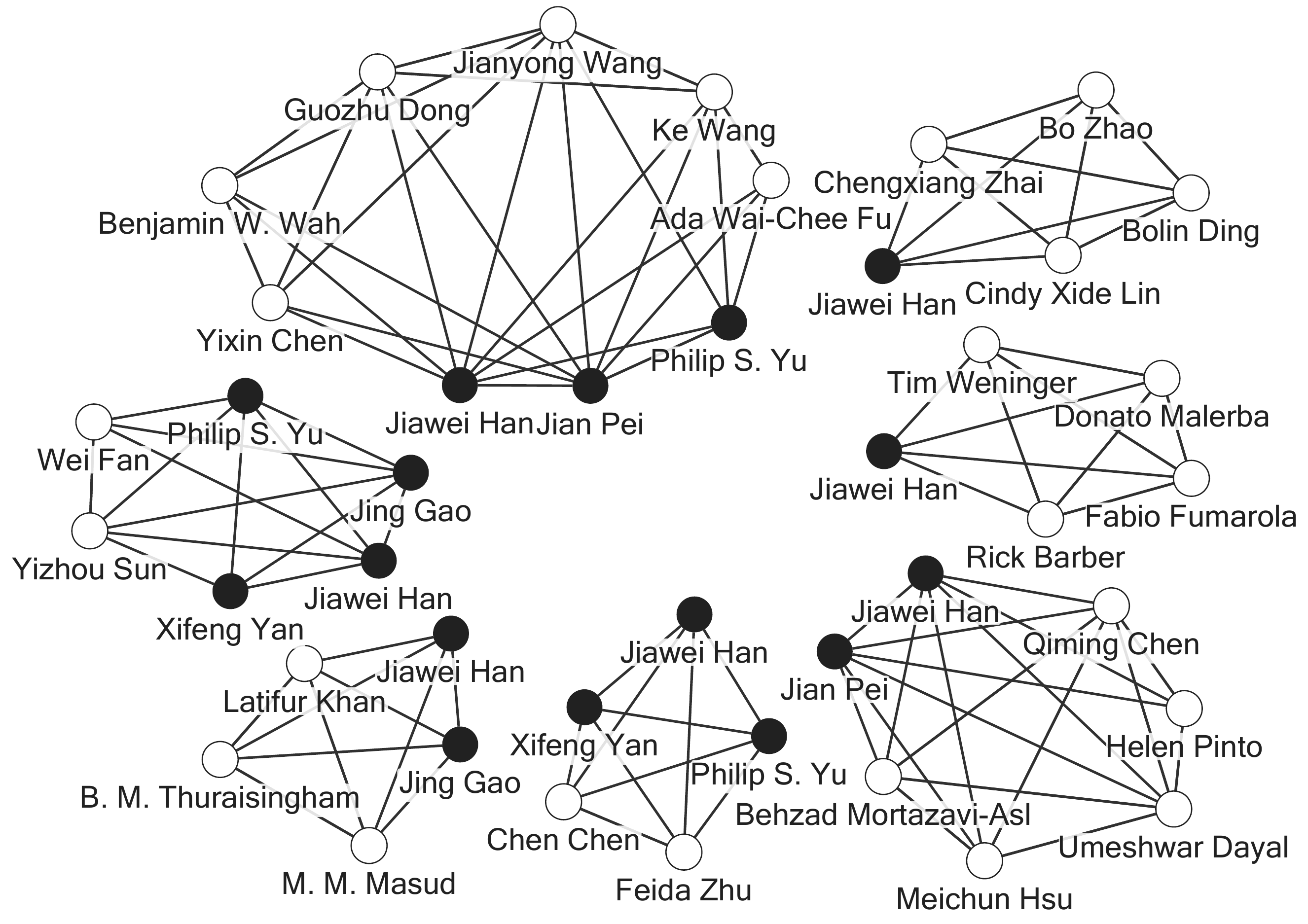}
}
\\
\vspace*{-0.3cm}
\hspace{-1em}
\subfigure[\small{4-\eccs and 4-core}]{
	\includegraphics[width=0.92\columnwidth]{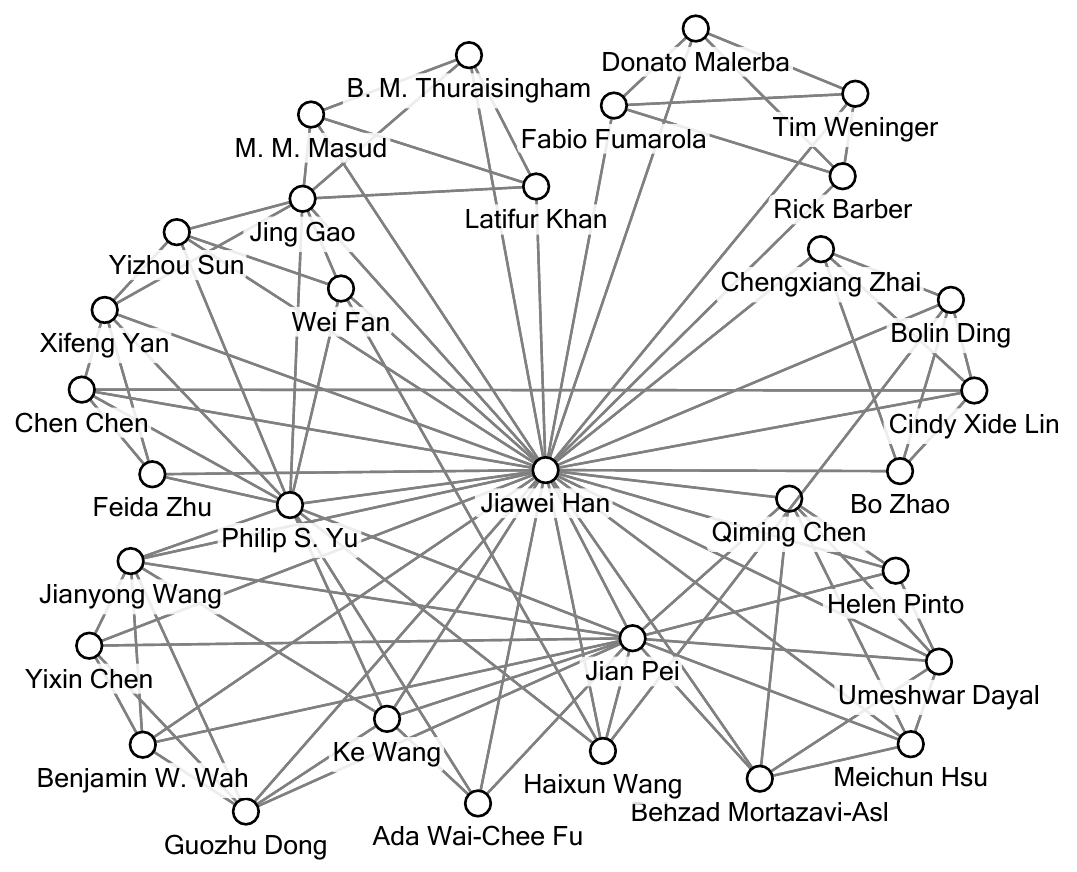}
}
\end{tabular}
\end{center}
\vspace*{-0.1cm}
\topcaption{Case Study on DBLP}
\vspace*{-0.1cm}
\label{fig:exp:caseStudy}
\end{figure}

In this experiment, we conduct a case study to visually reveal the quality of $k$-\vccs. We construct a collaboration graph from the \dblp (http://dblp.uni-trier.de/). Each vertex of the graph represents an author and an edge exists between two authors if they have $3$ or more common publications. Since the $k$-\vccs in the original graph are too large to show, we pick up the author `Jiawei Han' and his neighbors. We use the induced subgraph of these vertices to conduct this case study.

We query all $4$-\vccs containing `Jiawei Han' and the result is shown in \reffig{exp:caseStudy} (a). We obtain seven $4$-\vccs. Each of them is dense. A vertex is marked black if it appears in more than one $4$-\vccs. The result clearly reveals different research groups related to `Jiawei Han'. Some core authors appear in multiple groups, such as `Philip S. Yu' and 'Jian Pei'. As a comparison, we get only one $4$-\ecc, which contains the authors in all $4$-\vccs. The result of $4$-core is the same as $4$-\ecc in this experiment. Note that author `Haixun Wang' appears in $4$-\eccs and $4$-cores but not in any $4$-\vcc. That means he has cooperations with some authors in the research groups of `Jiawei Han', but those authors are from different identified groups and he does not belong to any of these groups. 


\section{Related Work}
\label{sec:relatedwork}
\vspace*{-0.1cm}
\stitle{Cohesive Subgraph.} Efficiently computing cohesive subgraphs, based on a designated metric, has drawn a large number of attentions recently. \cite{Chang2013F, Cheng2011} propose algorithms for maximal clique problem. However, the definition of clique is too strict. For relaxation, some clique-like metrics are proposed. These metrics can be roughly classified into three categories, 1) global cohesiveness, 2) local degree and triangulation, and 3) connectivity cohesiveness.

\sstitle{1. Global Cohesiveness.} \cite{luce1950connectivity} defines an $s$-clique model to relax the clique model by allowing the distance between two vertices to be at most $s$, i.e., there are at most $s-1$ intermediate vertices in the shortest path. However, it does not require that all intermediate vertices are in the $s$-clique itself. To handle this problem, \cite{Mokken1979} proposes an $s$-club model requiring that all intermediate vertices are in the same $s$-club. In addition, $k$-plex allows each vertex in such subgraph can miss at most $k$ neighbors \cite{Berlowitz2015, Stephen1978}. Quasi-clique is a subgraph with $n$ vertices and at least $\gamma*{n \choose 2}$ edges \cite{Zeng2006}. These kinds of metrics globally require the graph to satisfy a designated density or other certain criterions. They do not carefully consider the situation of each vertex and thus cannot effectively reduce the free rider effect \cite{Wu2015}.

\sstitle{2. Local Degree and Triangulation.} $k$-core is maximal subgraph in which each vertex has a degree at least $k$ \cite{Batagelj2003}. It only requires the minimum number of neighbors for each vertex in the graph to be no smaller than $k$. Therefore the number of non-neighbors for each vertex can be large. It is difficult to retain the familiarity when the size of a $k$-core is large. $k$-truss has also been investigated in \cite{cohen2008trusses, Wang2012, Shao2014}. It requires that each edge in a $k$-truss is contained in at least $k-2$ triangles. $k$-truss has similar problem as $k$-core. It is easy to see that two cohesive subgraphs can be simply identified as one $k$-truss if they share mere one edge. In addition, $k$-truss is invalid in some popular graphs such as bipartite graphs. This model is also independently defined as $k$-mutual-friend subgraph and studied in \cite{zhao2012large}. Based on triangles, \textit{DN}-graph \cite{Wang2010} with parameter $k$ is a connected subgraph $G'(V',E')$ satisfying following two conditions: 1) Every connected pair of vertices in $G'$ shares at least $\lambda$ common neighbors. 2) For any $v \in V \backslash V', \lambda (V' \cup \{v\}) < \lambda$; and for any $v \in V', \lambda (V' \backslash \{v\}) \le \lambda$. Such metric seems a little strict and generates many redundant results. Also, detecting all \textit{DN}-graphs is NP-Complete. Approximate solutions are given and the time complexity is still high \cite{Wang2010}.

\sstitle{3. Connectivity Cohesiveness.} In this category, most of existing works only consider the edge connectivity of a graph. The edge connectivity of a graph is the minimum number of edges whose removal disconnect the graph. \cite{Yan2005} first proposes algorithm to efficiently compute frequent closed k-edge connected subgraphs from a set of data graphs. However, a frequent closed subgraph may not be an induced subgraph. To conquer this problem, \cite{Zhou2012} gives a cut-based method to compute all $k$-edge connected components in a graph. To further improve efficiency, \cite{Chang2013} proposes a decomposition framework for the same problem and achieves a high speedup in the algorithm. 

\stitle{Vertex Connectivity.} \cite{even1975network} proves that the time complexity of computing maximum flow reaches $O(n^{0.5}m)$ in an unweighted directed graph while each vertex inside has either a single edge emanating from it or a single edge entering it. This result is used to test the vertex connectivity of a graph with given $k$ in $O(n^{0.5}m^2)$ time. \cite{SEven1975} further reduces the time complexity of such problem to $O(k^3m+knm)$. There are also other solutions for finding the vertex connectivity of a graph \cite{zvi1980, esfahanian1984}. To speed up the computation of vertex connectivity, \cite{CheriyanKT93} finds a sparse certificate of $k$-vertex connectivity which can be obtained by performing scan-first search $k$ times.


\section{Conclusions}
\label{sec:conclusion}
Cohesive graph detection is an important graph problem with a wide spectrum of applications. Most of existing models will cause the free rider effect that combines irrelevant subgraphs into one subgraph. In this paper, we first study the problem of detecting all $k$-vertex connected components in a given graph where the vertex connectivity has been proved as a useful formal definition and measure of cohesion in social groups. This model effectively reduces the free rider effect while retaining many good structural properties such as bounded diameter, high cohesiveness, bounded graph overlapping, and bounded subgraph number. We propose a polynomial time algorithm to enumerate all $k$-\vccs via a overlapped graph partition framework. We propose several optimization strategies to significantly improve the efficiency of our algorithm. We conduct extensive experiments using seven real datasets to demonstrate the effectiveness and the efficiency of our approach. 
\balance

{\small
\bibliographystyle{abbrv}
\bibliography{scRef}

\begin{thebibliography}{10}

\bibitem{kalvarez2005}
J.~I. Alvarez{-}Hamelin, L.~Dall'Asta, A.~Barrat, and A.~Vespignani.
\newblock k-core decomposition: a tool for the visualization of large scale
  networks.
\newblock {\em CoRR}, abs/cs/0504107, 2005.

\bibitem{bader2003}
G.~Bader and C.~Hogue.
\newblock An automated method for finding molecular complexes in large protein
  interaction networks.
\newblock {\em BMC Bioinformatics}, 4(1), 2003.

\bibitem{Batagelj2003}
V.~Batagelj and M.~Zaversnik.
\newblock An o(m) algorithm for cores decomposition of networks.
\newblock {\em CoRR}, cs.DS/0310049, 2003.

\bibitem{Berlowitz2015}
D.~Berlowitz, S.~Cohen, and B.~Kimelfeld.
\newblock Efficient enumeration of maximal k-plexes.
\newblock In {\em Proc. of SIGMOD'15}, 2015.

\bibitem{Chang2013F}
L.~Chang, J.~X. Yu, and L.~Qin.
\newblock Fast maximal cliques enumeration in sparse graphs.
\newblock {\em Algorithmica}, 66, 2013.

\bibitem{Chang2013}
L.~Chang, J.~X. Yu, L.~Qin, X.~Lin, C.~Liu, and W.~Liang.
\newblock Efficiently computing k-edge connected components via graph
  decomposition.
\newblock In {\em Proc. of SIGMOD'13}, 2013.

\bibitem{Cheng2011}
J.~Cheng, Y.~Ke, A.~W.-C. Fu, J.~X. Yu, and L.~Zhu.
\newblock Finding maximal cliques in massive networks.
\newblock {\em ACM Trans. Database Syst.}, 36, 2011.

\bibitem{CheriyanKT93}
J.~Cheriyan, M.~Kao, and R.~Thurimella.
\newblock Scan-first search and sparse certificates: An improved parallel
  algorithms for k-vertex connectivity.
\newblock {\em {SIAM} J. Comput.}, 22(1):157--174, 1993.

\bibitem{cohen2008trusses}
J.~Cohen.
\newblock Trusses: Cohesive subgraphs for social network analysis.
\newblock {\em National Security Agency Technical Report}, page~16, 2008.

\bibitem{Cui2014}
W.~Cui, Y.~Xiao, H.~Wang, and W.~Wang.
\newblock Local search of communities in large graphs.
\newblock In {\em Proc. of SIGMOD'14}, 2014.

\bibitem{EdacherySB99}
J.~Edachery, A.~Sen, and F.~Brandenburg.
\newblock Graph clustering using distance-k cliques.
\newblock In {\em Proc. of GD'12}, 1999.

\bibitem{esfahanian1984}
A.~H. Esfahanian and S.~Louis~Hakimi.
\newblock On computing the connectivities of graphs and digraphs.
\newblock {\em Networks}, 14(2):355--366, 1984.

\bibitem{SEven1975}
S.~Even.
\newblock An algorithm for determining whether the connectivity of a graph is
  at least k.
\newblock {\em {SIAM} J. Comput.}, 4, 1975.

\bibitem{even1975network}
S.~Even and R.~E. Tarjan.
\newblock Network flow and testing graph connectivity.
\newblock {\em {SIAM} J. Comput.}, 4, 1975.

\bibitem{zvi1980}
Z.~Galil.
\newblock Finding the vertex connectivity of graphs.
\newblock {\em {SIAM} J. Comput.}, 9, 1980.

\bibitem{Huang2014}
X.~Huang, H.~Cheng, L.~Qin, W.~Tian, and J.~X. Yu.
\newblock Querying k-truss community in large and dynamic graphs.
\newblock In {\em Proc. of SIGMOD'14}, 2014.

\bibitem{luce1950connectivity}
R.~D. Luce.
\newblock Connectivity and generalized cliques in sociometric group structure.
\newblock {\em Psychometrika}, 15, 1950.

\bibitem{menger1927}
K.~Menger.
\newblock Zur allgemeinen kurventheorie.
\newblock {\em Fundamenta Mathematicae}, 10(1):96--115, 1927.

\bibitem{Mokken1979}
R.~J. Mokken.
\newblock Cliques, clubs and clans.
\newblock {\em Quality and Quantity}, 13, 1979.

\bibitem{Moody00structuralcohesion}
J.~Moody and D.~R. White.
\newblock Structural cohesion and embeddedness: A hierarchical conception of
  social groups.
\newblock {\em American Sociological Review}, 68, 2000.

\bibitem{Moon1965}
J.~W. Moon and L.~Moser.
\newblock On cliques in graphs.
\newblock {\em Israel Journal of Mathematics}, 3(1):23--28, 1965.

\bibitem{PattilloYB13}
J.~Pattillo, N.~Youssef, and S.~Butenko.
\newblock On clique relaxation models in network analysis.
\newblock {\em European Journal of Operational Research}, 226, 2013.

\bibitem{Stephen1978}
S.~B. Seidman and B.~L. Foster.
\newblock {A graph-theoretic generalization of the clique concept}.
\newblock {\em Journal of Mathematical Sociology}, 6, 1978.

\bibitem{Shao2014}
Y.~Shao, L.~Chen, and B.~Cui.
\newblock Efficient cohesive subgraphs detection in parallel.
\newblock In {\em Proc. of SIGMOD'14}, 2014.

\bibitem{Stoer1997}
M.~Stoer and F.~Wagner.
\newblock A simple min-cut algorithm.
\newblock {\em J. ACM}, 44(4), 1997.

\bibitem{verma2012}
A.~Verma and S.~Butenko.
\newblock Network clustering via clique relaxations: A community based
  approach.
\newblock {\em Graph Partitioning and Graph Clustering}, 588, 2012.

\bibitem{Wang2012}
J.~Wang and J.~Cheng.
\newblock Truss decomposition in massive networks.
\newblock {\em Proc. VLDB Endow.}, 5, 2012.

\bibitem{Wang2010}
N.~Wang, J.~Zhang, K.-L. Tan, and A.~K.~H. Tung.
\newblock On triangulation-based dense neighborhood graph discovery.
\newblock {\em Proc. VLDB'10}, 4, 2010.

\bibitem{white2001}
D.~R. White and F.~Harary.
\newblock The cohesiveness of blocks in social networks: Node connectivity and
  conditional density.
\newblock {\em Sociological Methodology}, 31(1):305--359, 2001.

\bibitem{hassler1932}
H.~Whitney.
\newblock Congruent graphs and the connectivity of graphs.
\newblock {\em American Journal of Mathematics}, 54(1):150--168, 1932.

\bibitem{Wu2015}
Y.~Wu, R.~Jin, J.~Li, and X.~Zhang.
\newblock Robust local community detection: On free rider effect and its
  elimination.
\newblock {\em Proc. VLDB Endow.}, 8, 2015.

\bibitem{Yan2005}
X.~Yan, X.~J. Zhou, and J.~Han.
\newblock Mining closed relational graphs with connectivity constraints.
\newblock In {\em Proc. of ICDE'05}, 2005.

\bibitem{Zeng2006}
Z.~Zeng, J.~Wang, L.~Zhou, and G.~Karypis.
\newblock Coherent closed quasi-clique discovery from large dense graph
  databases.
\newblock In {\em Proc. of KDD'06}, 2006.

\bibitem{zhang2010}
H.~Zhang, H.~Zhao, W.~Cai, J.~Liu, and W.~Zhou.
\newblock Using the k-core decomposition to analyze the static structure of
  large-scale software systems.
\newblock {\em The Journal of Supercomputing}, 53(2), 2010.

\bibitem{Zhang2012}
Y.~Zhang and S.~Parthasarathy.
\newblock Extracting analyzing and visualizing triangle k-core motifs within
  networks.
\newblock In {\em Proc. of ICDE'12}, 2012.

\bibitem{zhao2012large}
F.~Zhao and A.~K. Tung.
\newblock Large scale cohesive subgraphs discovery for social network visual
  analysis.
\newblock In {\em Proc. of the VLDB'12}, volume~6, 2012.

\bibitem{Zhou2012}
R.~Zhou, C.~Liu, J.~X. Yu, W.~Liang, B.~Chen, and J.~Li.
\newblock Finding maximal k-edge-connected subgraphs from a large graph.
\newblock In {\em Proc. of EDBT'12}, 2012.

\end{thebibliography}
}

\end{document}